\documentclass[ecta,nameyear,draft]{econsocart}
\usepackage{threeparttable}
\usepackage{amsmath,blkarray,bm,dsfont}
\usepackage{amssymb}
\usepackage{tikz}
\usepackage{subfigure}
\usepackage[mathscr]{euscript}
\usepackage{caption}
\usepackage[normalem]{ulem}
\usepackage{xcolor}
\usepackage{colortbl} 
\usepackage{arydshln} 
\usepackage{float}
\usepackage{multirow}
\usetikzlibrary{positioning}
\usepackage{enumitem}

\usepackage{xr-hyper}
\usepackage{pdfpages}

\usetikzlibrary{patterns,intersections}

\newcommand{\imply}{\Rightarrow}

\newcommand{\ra}{\rightarrow}
\newcommand{\inter}{\cap}

\newcommand{\ess}{\operatorname{ess}}

\newcommand{\cl}{\operatorname{cl}}

\RequirePackage[colorlinks,citecolor=blue,linkcolor=blue,urlcolor=blue, filecolor=blue, pagebackref]{hyperref}

\newcommand{\E}{\mathcal{E}} 
\newcommand{\e}{\varepsilon} 
\renewcommand{\ss}{\mathcal{S}} 
\newcommand{\sF}{\mathcal{F}} 
\newcommand{\sPF}{\mathcal{PF}} 
\newcommand{\sY}{\mathcal{Y}} 
\newcommand{\sX}{\mathcal{X}} 
\newcommand{\bS}{\mathbb{S}} 
\newcommand{\sU}{\mathcal{U}} 
\newcommand{\bB}{\mathbb{B}} 
\newcommand{\cC}{\mathcal{C}} 
\newcommand{\sIM}{\mathcal{IM}} 

\renewcommand{\v}{v} 
\newcommand{\q}{q} 
\newcommand{\s}{s} 

\newcommand{\R}{{\mathbb R}} 
\newcommand{\B}{B} 
\newcommand{\G}{G} 

\newcommand{\X}{X} 
\newcommand{\Y}{Y} 
\newcommand{\LL}{L} 
\newcommand{\bLL}{\mathbf{\LL}} 
\newcommand{\cM}{\mathcal{M}} 
\newcommand{\Z}{Z} 

\newcommand{\h}{h} 
\newcommand{\hn}{\widehat{\h}} 
\newcommand{\btheta}{\boldsymbol{\theta}} 
\newcommand{\bpi}{\pi}
\newcommand{\boeta}{\boldsymbol{\eta}} 
\newcommand{\balpha}{\alpha} 
\newcommand{\bbeta}{\boldsymbol{\beta}} 
 
\newcommand{\bmu}{\boldsymbol{\mu}} 
\newcommand{\Jstar}{J^*} 

\newcommand{\HH}{\mathsf{H}} 
\newcommand{\bP}{\mathbb{P}} 
\newcommand{\bEs}{\mathbb{E}^*} 
\newcommand{\bE}{\mathbb{E}} 
\newcommand{\one}{\mathds{1}} 

\startlocaldefs
\theoremstyle{remark}
\newtheorem{theorem}{Theorem}[section]
\newtheorem{lemma}{Lemma}[section]
\newtheorem{corr}{Corollary}[section]
\newtheorem{defn}{Definition}

\newtheorem{prop}{Proposition}[section]

\newtheorem{asm}{Assumption}
\newtheorem{remark}{Remark}[section]
\newtheorem{proc}{Procedure}

\defcitealias{lia:lu:mu:oku24}{LLMO} 
\defcitealias{obe:pow:vog:mul19}{OPVM} 
\defcitealias{liu:mol:vel26OA}{Online Appendix}
\endlocaldefs

\begin{document}
\begin{frontmatter}

\title{Identification and Inference\\
for Algorithmic Frontiers with Selective Labels}

\begin{aug}
%
%
%
\author[id=au1,addressref={add1}]{\fnms{Yiqi}~\snm{Liu}\ead[label=e1]{yiqi@cmsa.fas.harvard.edu}}
\author[id=au2,addressref={add2}]{\fnms{Francesca}~\snm{Molinari}\ead[label=e2]{f-molinari@jhu.edu}}
\author[id=au3,addressref={add3}]{\fnms{Amilcar}~\snm{Velez}\ead[label=e3]{amilcare@cornell.edu}}
\address[id=add1]{%
\orgdiv{Center of Mathematical Sciences and Applications},
\orgname{Harvard University}}

\address[id=add2]{%
\orgdiv{School of Government and Policy},
\orgname{Johns Hopkins University}}

\address[id=add3]{%
\orgdiv{Department of Economics},
\orgname{Cornell University}}
\end{aug}

\vspace{-2.3em}

\support{Preliminary version. This draft: October 14, 2026. We thank Jos{\'e} Montiel-Olea and participants at the Cowles Conference on Econometrics, ES North American Summer Meeting, ESIF: Economics and AI+ML, Joint Statistical Meetings, and Stanford Empirical Methods in the Age of AI Conference for helpful comments. Liu gratefully acknowledges support from Harvard CMSA.}
\begin{abstract}
This paper provides identification results to characterize a fairness-accuracy (FA) frontier, and statistical inference tools to test hypotheses and build a confidence set for the FA-frontier, when outcomes are observed only for selected individuals.
When the selection process is unrestricted but loss is measured in specific ways, we provide a characterization of the sharp identification region of the FA-frontier.
When the observed distribution rules out frontier membership for a binary status quo algorithm, we construct a single alternative algorithm that improves both groups' accuracy without worsening fairness, uniformly over all admissible distributions of missing labels.
For any loss functions, under selection on observables and strict overlap assumptions, we obtain point identification and propose a debiased machine learning estimator, derive its asymptotic distribution, and show how this can be used to carry out inference for the FA-frontier.
\end{abstract}


\bigskip

\end{frontmatter}  

\newpage
\section{Introduction}\label{sec:intro}
Over the last decade, algorithms have come to play an increasingly prevalent role in everyday life, often by supporting high-stake decisions through predictions of quantities such as re-offense risk, repayment likelihood, and health care needs.
These predictions enter institutional decision processes that determine who is detained or released, who is granted credit, and which patients receive outreach and care-management interventions.
The increasing reliance on algorithms has prompted the emergence of a set of regulatory and oversight processes, known as \emph{AI governance}, that rule the design, deployment, documentation, and monitoring of AI systems, including, in some settings, restrictions on the use of protected attributes as model inputs.
A distinctive feature of these regimes is their evidence based orientation: for high-risk AI systems, the EU AI Act requires technical documentation, logging capabilities, and appropriate accuracy and robustness over the lifecycle \citep[Articles 11, 12, 15, 26]{eu_ai_act}, while related governance proposals emphasize dataset documentation, model reporting, and internal auditing \citep[see, e.g.,][]{gebru_etal_2021,mitchell_etal_2019,raji_etal_2020}.

This paper extends the research agenda set out in \citet{liu:mol25v2} to provide a toolkit for nonparametric estimation of and statistical inference on the fairness-accuracy tradeoff that characterizes algorithms.
Fairness and accuracy evaluations play a central role in governance, as they determine both the validity of model supported decisions and whether those decisions impose disproportionate harms on legally protected groups.\footnote{Algorithms may exhibit systematic differences in performance across legally protected groups, both in predictive accuracy and in the decisions they induce \citep[e.g.,][]{ang:lar:mat:kir16, arn:dob:hul21, obe:pow:vog:mul19, cow:tuc20,ber:jei:jab:kea:rot21}. Yet, improving parity in expected losses often requires sacrificing accuracy for at least one group, and conversely increasing accuracy for one group may worsen disparities.}
At the same time, governance questions are rarely about a single fixed model.
In public oversight and non-discrimination doctrine, the relevant question is often whether a challenged procedure can be replaced by an alternative that achieves comparable accuracy with less adverse impact, shifting attention from a single algorithm to a potentially large class of admissible decision rules, e.g., those obtained under constraints on inputs, model classes, or the use of sensitive attributes.
Regulators and oversight bodies are called to adjudicate (i) whether a deployed rule is Pareto efficient, (ii) whether design restrictions enable less discriminatory alternatives (LDAs) with comparable accuracy, and (iii) whether such LDAs exist within an admissible class of decision rules.\footnote{For example, US employment discrimination law prohibits selection procedures that cause disparate impact unless the employer demonstrates that they are ``job related for the position in question and consistent with business necessity.'' Even then the practice is unlawful if an alternative practice would serve the employer's legitimate interests with less discriminatory impact \citep{civil_rights_act_1991}.}
As in \citet{liu:mol25v2}, we build on the theoretical framework put forward by \citet*[LLMO henceforth]{lia:lu:mu:oku24}, which in our view is well-suited to formalize these questions because it characterizes a fairness-accuracy (FA) frontier under broad preferences over group-specific expected losses as well as the implications of banning attributes as inputs to the algorithms.

Our first innovation relative to \citet{liu:mol25v2}'s groundwork is to allow for accuracy and fairness of an algorithm to be assessed through two different loss functions, thereby aligning our treatment with the literature in computer science \citep{corbettdavies2017cost,menon2018cost,kim2020fact,gillis2024lda}.
We show that the support function of a convex feasible set $\E\subset\R^3$ of expected accuracy losses for the two groups and the difference in expected fairness losses between the groups plays a crucial role in characterizing the FA-frontier and carrying out inference.\footnote{In \citet{liu:mol25v2}'s analysis and \citetalias{lia:lu:mu:oku24}'s main text, the same loss function is used to measure fairness and accuracy, so that the feasible set $\E$ is in $\R^2$. \citet{liu:mol25v2} propose to use the support function of $\E\subset\R^2$ as the unifying tool for statistical inference.}
We extend the results in \citetalias{lia:lu:mu:oku24}'s Appendix O by providing a simple and testable necessary and sufficient condition under which the FA-frontier coincides with the Pareto frontier.
When this condition is not satisfied, a policymaker who values fairness may choose to deploy an algorithm that is Pareto dominated in accuracy loss, to ameliorate its performance in fairness loss.

Our second innovation addresses the reality that in many high-stakes settings, outcomes are observed only for selected individuals: pretrial misconduct (or lack thereof) is observed only if bail is granted; repayment behavior is observed only when credit is extended; certain health outcomes are observed only for patients who receive treatment or remain under follow-up within a system; etc.
Hence, instead of observing the outcome vector of interest $\Y^*$, one observes $\Z\Y^*$, where $\Z\in\{0,1\}$ indicates observability.
This \emph{selective labels} problem creates challenges for governance compliance similar to the ones long recognized in the vast econometrics literature on selectively observed data \citep[e.g.,][]{heckman_1979,manski_1989}. Accuracy metrics computed using only observed outcomes may misrepresent true performance, potentially overstating accuracy where it matters most (e.g., if those denied bail would not have re-offended). 
Hence, assessments of whether a system achieves ``appropriate levels of accuracy'' \citep[EU AI Act,][Article 15]{eu_ai_act} or whether LDAs exist \citep{civil_rights_act_1991} may be incorrect and required documentation may systematically exclude the most policy-relevant counterfactual outcomes.

When accuracy is measured by classification loss and fairness by statistical parity \citep[as in][]{menon2018cost}, we provide, without restricting the selection process, a criterion for whether some admissible distribution of the missing labels places the expected losses of a status quo algorithm $a^*$ on the associated FA-frontier.
When that is not the case, we construct a single alternative algorithm that improves both groups' accuracy without worsening fairness, uniformly over all admissible distributions.
We also characterize the sharp identification region for the FA-frontier in the spirit of \citet{manski_1989} \citep[see][for reviews of the partial identification literature]{mol20,can:ill:velez26}.
The challenge is that frontier membership is a relational property: the FA-frontier consists of expected losses $\e(a)\in\E$ that are not FA-dominated by those of any other algorithm in $\mathcal{A}(\sX)$, the class of measurable maps from $\sX$ to $[0,1]$ (see Definition~\ref{def:R3_frontier}).
Selective labels leave both $\e(a)$ and $\E$ partially identified, requiring frontier membership to be assessed across all admissible distributions for the missing label, using the \emph{same} distribution to construct both $\e(a)$ and $\E$.
We reduce this daunting task to a finite-dimensional optimization problem.

When the selection process is restricted by a missing-at-random (MAR) assumption, where conditional on attributes, outcome observability does not convey additional information on $\Y^*$, we obtain point identification of the FA-frontier.
For any loss function, we provide a debiased machine learning (DML) estimator for the FA-frontier based on inverse probability weighting. 
Our DML approach is particularly appropriate for governance contexts, where high-dimensional administrative data enable flexible nonparametric modeling of selection and formal inference tools are needed for legally defensible determinations.
While ideas transfer from the program evaluation and high-dimensional inference literature \citep[e.g.,][]{CHKSS}, we extend the methods to carry out inference on a new set-valued estimand: the FA-frontier.
We derive the asymptotic distribution of our estimator and test statistics building on \citet{liu:mol25v2,sem23,cha:che:mol:sch18,ChernozhukovDML,ber:mol08}, and \citet{fan:san19}, provide inverse-probability weighting correction in the DML estimator, and asymptotically valid confidence sets for the FA-frontier and tests for LDA existence. 

Accompanying software implementing both the MAR-based and partial identification approaches is under development, along with empirical applications of our methods.

\textbf{Related Literature.}
A large literature in computer science, statistics, and economics studies algorithmic fairness; see \citet{cho:rot18}, \citet{bar:har:nar23}, and \citet{Corbett24} for overviews.
Much of this work treats fairness as a constraint or regularizer in an objective that prioritizes predictive performance \citep[e.g.,][]{dwork12, ber:hei:jab:jos:kea:mor:jam:nee:roth17}, or studies impossibility results and tradeoffs among fairness criteria \citep{kleinberg2016inherent}.
A related literature computes fairness-accuracy Pareto frontiers or ranges of disparities for specific model classes, typically as an optimization or auditing exercise rather than as a problem of statistical inference \citep[e.g.,][]{wei:nie21, lit:wey:all22, cos:ram:cho21}.
When labels are perfectly observed, \citet{aue:lia:tab:oku24} provide a sample-splitting based test for LDA existence, and \citet{fal:jor:uli26} study minimax-optimality of the sample analog of a Pareto-optimal linear rule in a homoskedastic linear model with squared loss.

Within the algorithmic fairness literature, other work has focused on selective labels \citep[e.g.,][]{lak:kle:les:lud:mul17, kle:lak:les:lud:mul17, ram:cos:ken25, kha:tam:yao}.
\citet{ram:cos:ken25} develop a general framework for robust evaluation and design of predictive algorithms under selective labels and unobserved confounding, delivering bounds for performance measures under alternative restrictions on the selection process.
We view their contribution as addressing a distinct question from ours: how to evaluate predictive performance in the presence of counterfactual outcomes.
We focus instead on identification and inference for the fairness-accuracy tradeoff of group losses across all feasible algorithms using a given information set, and on hypothesis tests and confidence statements about frontier properties such as LDA existence that directly map to governance questions and design restrictions. 

\textbf{Outline.} Section \ref{sec:setup} lays out notation and derives the FA-frontier extending \citetalias{lia:lu:mu:oku24} and \citet{liu:mol25v2}.
Section \ref{sec:identif_and_estim} formalizes the selective labels problem and derives the sharp identification region for the FA-frontier for specific loss functions and unrestricted selection processes. When a status quo algorithm is not on the frontier, we provide an alternative algorithm that improves over it uniformly across admissible distributions for the missing labels.
We also establish, for any loss functions, point identification of the FA-frontier under an MAR assumption.
Section~\ref{sec:DML} puts forward a DML estimator and derives its asymptotic distribution.
Section~\ref{sec:tests} proposes a test for equality between the Pareto and FA-frontier.
It then presents our test for existence of an LDA in the presence of selectively observed labels and derives its asymptotic distribution, building on \citet{liu:mol25v2}.
Section \ref{sec:conclude} concludes.
Our main proofs are in Appendix \ref{appn:A};
Appendix \ref{app:auxiliary} reports auxiliary results. Additional proof details can be found in the \hyperref[online_appendix]{Online Appendix}.

\section{Setup}\label{sec:setup}
Let a population of individuals be described by an outcome vector $\Y^*\in\sY\subset\R^{d_Y}$, a binary group identity $\G\in\{r,b\}$ (red or blue), and a vector of covariates $\X\in\sX\subset\mathbb{R}^{d_\X}$.
For example, one component of $\Y^*$ may denote an individual's number of active chronic illnesses in the subsequent year and the other component the cost of care; $\G$ may denote the individual's race, and $\X$ may include age, gender, biomarkers, comorbidity, and medication variables.
As in \citet{liu:mol25v2}, we leave the relation between $\G$ and $\X$ unspecified, but require $\G$ not to be a deterministic function of $\X$.
Each individual receives a binary decision $D\in\{0,1\}$, e.g., whether they are automatically enrolled in a high-risk care management program; or granted bail; or granted or denied a loan (while we assume $\G$ and $D$ to be binary, the results can be extend to multiple groups and multiple decisions).
An algorithm $a:\sX \mapsto [0,1]$ assigns a probability distribution to $D$ conditional on $\X$; e.g., the algorithm assigns each patient a health risk score in $[0,1]$; or a recidivism probability; or a repayment probability. 
For simplicity, we take $a(\X)$ to be the only input to the decision, and hence $D\sim a(\X)$.
We let $\mathcal A(\sX)$ denote the set of all (measurable) algorithms that map from the input space $\sX$ to a probability distribution over $D$.
Extending the framework in \citet{liu:mol25v2}, we let $\ell^A:\{0,1\}\times\sY\mapsto\mathbb R$ be a loss function used to measure the \emph{accuracy} of an algorithm for an individual with outcome $y\in\sY$, and $\ell^F:\{0,1\}\times\sY\mapsto\mathbb R$ a loss function used to measure the algorithm's \emph{fairness}.
For example, one could use classification error to measure accuracy ($\ell^A(d,y)=\one(d\neq y_1)$), which in a health care example returns the value $1$ if the algorithm mistakenly enrolls a healthy person in the high-risk care program or fails to enroll someone who is very sick.
And one could use statistical parity to measure fairness ($\ell^F(d,y)=d$), hence judging an algorithm more fair if the proportion of either group receiving the two decisions is closer.

This paper is concerned with the case where instead of observing $Y^*$, we observe $\Y=\Z\Y^*$, where $\Z \in \{0,1\}$ is an indicator variable that records when the outcome of interest $\Y^*$ is observed in the testing data. 
In other words, we allow for selectively observed labels.
Before analyzing this problem, however, we provide a tractable characterization of the FA-frontier, under the assumption that $(\Y^*,\X,\G,\Z)\sim\bP^*$ is observed, when the loss function used to measure accuracy differs from that used to measure fairness.
In doing so, we extend the analysis in Appendix O.1 of \citetalias{lia:lu:mu:oku24}.
Armed with our novel characterization, we then study what can be learned about the FA-frontier and how can valid statistical inference be carried out, when in fact only the selectively observed labels $\Y$ are available.

\subsection{The Fairness-Accuracy Frontier}\label{subsec:FAfrontier}
We denote $\mu_g \equiv \bP^*(G=g)$ the population proportion of group $g \in \{r,b\}$ and, for $d \in \{0,1\}$, the (unobserved) \emph{ideal labels} associated with loss function $\iota\in\{A,F\}$ by
\begin{align}
\LL_d^{g,\iota} &\equiv \tfrac{\ell^\iota(d,Y^*)\one\{G=g\}}{\mu_g}.\label{eq:defLdgA}
\end{align} 
We denote the conditional expectation of $\LL_d^{g,\iota}$ taken with respect to $\bP^*(\Y^*,\G|\X)$ by 
\begin{align*}
\theta_d^{g,\iota}(\X) &\equiv \bEs[\LL_d^{g,\iota}|\X] = \tfrac{\bEs[\ell^\iota(d,Y^*)\one\{G=g\}|\X]}{\mu_g}.
\end{align*} 
\begin{defn}[Ideal labels and conditional expectations]\label{def:ideal-labels}
    For $d\in\{0,1\}$, define
\begin{align}
\bLL_d
&\equiv
\begin{bmatrix}
\LL_d^{r,A}\\
\LL_d^{b,A}\\
\LL_d^{r,F}-\LL_d^{b,F}
\end{bmatrix}\in \R^3,\label{eq:defLvector}
\end{align}
and let $\btheta_d(\X)\equiv \bEs[\bLL_d|\X]$, $\Delta\bLL\equiv\bLL_1-\bLL_0$, and $\Delta\btheta(\X)\equiv \btheta_1(\X)-\btheta_0(\X)\in\R^3$.
\end{defn}
To make sure that the ideal labels are well-defined, we assume:
\begin{asm}[Moment restrictions for both losses]\label{asm:moments}
\textit{For all $d\in\{0,1\}, g\in\{r,b\}$, and $\iota\in\{A,F\}$, for some constants $0<c_1<1$, $0<c_2<\infty$, $\mu_g\in(c_1,1-c_1)$ and $\ess\sup_{X\in\sX}\bE^*\left[\left(\LL_d^{g,\iota}\right)^2\big|\X\right]<c_2$.} 
\end{asm}

\begin{defn}[Algorithm space]\label{def:algorithm}
An \emph{algorithm} is a Borel-measurable map $a: \sX \to [0,1]$, where $a(x)$ represents the probability of choosing $D=1$ given $\X=x$. Denote $\mathcal{A}(\sX)$ the set of all such measurable maps. Write $\mathcal{A}$ when the covariate space is clear from context.
\end{defn}
For each $a \in \mathcal{A}$ and $g \in \{r,b\}$, define the group-wise expected loss associated with loss function $\ell^\iota, \iota \in \{A,F\}$, when $D\sim a(\X)$ as:
\begin{align}
e_g^\iota(a) &\equiv \bEs[\ell^\iota(D,\Y^*)|\G=g]=\bEs[a(\X)\theta_1^{g,\iota}(\X)+(1-a(\X))\theta_0^{g,\iota}(\X)].\label{eq:eg_iota}
\end{align}
We measure the \emph{unfairness} of an algorithm $a$ through $|e^F_r(a)-e^F_b(a)|$. We define the \emph{feasible set} as the collection of all fairness-accuracy triplets attainable by algorithms in $\mathcal A$:
\begin{defn}[Feasible set]\label{def:feasible}
Given $\bP^*$ and $\mathcal A$, as $a$ ranges in $\mathcal A$, the \emph{feasible set} for $\left(e^A_r(a),e^A_b(a),e^F_r(a)-e^F_b(a)\right)$ is
\begin{align*}
\E \equiv \left\{(\e_1,\e_2,\e_3) \in \R^3: \exists a \in \mathcal{A} \text{ such that } \e_1 = e_r^A(a), \e_2 = e_b^A(a), \e_3 = e_r^F(a)-e_b^F(a)\right\}.
\end{align*}
\end{defn}
\noindent To simplify notation, we omit the dependence of $\E$ on $(\bP^*,\mathcal A)$.
In Lemma~\ref{prop:compact_convex}, we show that under Assumptions \ref{asm:moments}-\ref{asm:margin}, $\E$ is non-empty, compact, and convex.

We adopt the FA preference ordering and FA-frontier of \citetalias{lia:lu:mu:oku24}, extending their $\R^2$ formulation to the feasible set $\E \subset \R^3$ in Definition~\ref{def:feasible} \citep[][p.8, use a similar definition of FA dominance]{aue:lia:tab:oku24}:
\begin{defn}[FA-dominance and FA-frontier]\label{def:R3_frontier}
Given $\e,\e' \in \E$, $\e' \succ_{FA} \e$ if $\e'_1 \leq \e_1,~\e'_2 \leq \e_2,~|\e'_3| \leq |\e_3|$, with at least one inequality strict.
The FA-frontier is:
\begin{align}
    \sF\equiv\{\e \in \E: \nexists ~\e' \in \E \text{ such that } \e' \succ_{FA} \e\}\subset\R^3.\label{eq:sF:R3}
\end{align}
\end{defn}
We also define the Pareto frontier, which ignores considerations of fairness:
\begin{align}
 \sPF&\equiv\left\{\e\in\E:\nexists\,\e'\in \E \text{ s.t. }\e'_1\le \e_1,\ \e'_2\le \e_2,\text{ with at least one strict inequality}\right\},\label{eq:sPF}
\end{align}
and the points in $\E\subset\R^3$ that minimize, respectively, $e_r^A$, $e_b^A$, and $|e_r^F-e_b^F|$:
\begin{align}
    R \equiv \arg\min_{\e\in\E}\e_1,\qquad B \equiv \arg\min_{\e\in\E}\e_2,\qquad F\equiv\arg\min_{\e\in\E}|\e_3|.\label{eq:define_R_B_F}
\end{align}

\noindent In Appendix~\ref{appn:B:geometry-sF} we show that the shape of $\sF$ depends on the position of $\E$, $\sPF$, $R$ and $B$ relative to the hyperplane $\{\e_3=0\}=\{e_r^F=e_b^F\}$. 
Figure~\ref{fig:F} illustrates several examples.

\begin{remark}
    \label{rem:LLMO-frontier}
    \citetalias{lia:lu:mu:oku24} (Appendix O.1) derive the FA-frontier and study its properties when the loss function used to measure accuracy differs from that used to measure fairness, which is the framework used in this paper. But they focus on an alternative definition of feasible set in $\R^2$, which projects $\E$ from Definition \ref{def:feasible} onto the first two accuracy components. In Appendix \ref{app:auxiliary:LLMO-frontier}, we establish that our FA-frontier in Definition \ref{def:R3_frontier} and the FA-frontier of \citetalias{lia:lu:mu:oku24} (Definition O.2) coincide in what we argue is the only way that matters.
\end{remark}

\subsection{Support Function-Based Characterization of the FA-Frontier}\label{subsec:supp-func-FAfrontier}
To begin, we characterize \citep[similarly to][Propositions 3.1-3.2]{liu:mol25v2} the support function of $\E$, denoted $\h_\E(\q)$ for $\q$ a direction vector in $\bS^2 \equiv \{\q\in\R^3:\|\q\|=1\}$.
This functional turns out to be crucial for all our results. 
We assume:
\begin{asm}[Margin condition]\label{asm:margin}
\textit{There exists a constant $m\in(0,1]$ such that for every $\delta>0$, $\sup_{\q\in\bS^2}\bP^*\big(|\q^\intercal \Delta\btheta(\X)|\le \delta\big)\lesssim\delta^m$.}
\end{asm}

\begin{figure}\captionsetup[subfigure]{font=footnotesize}
\vspace{-.8cm}

\centering
\subfigure[$\E\subset\{e_r^F>e_b^F\}$ has no kinks; $\sPF\subsetneq\sF$]{
\includegraphics[width=0.425\linewidth]{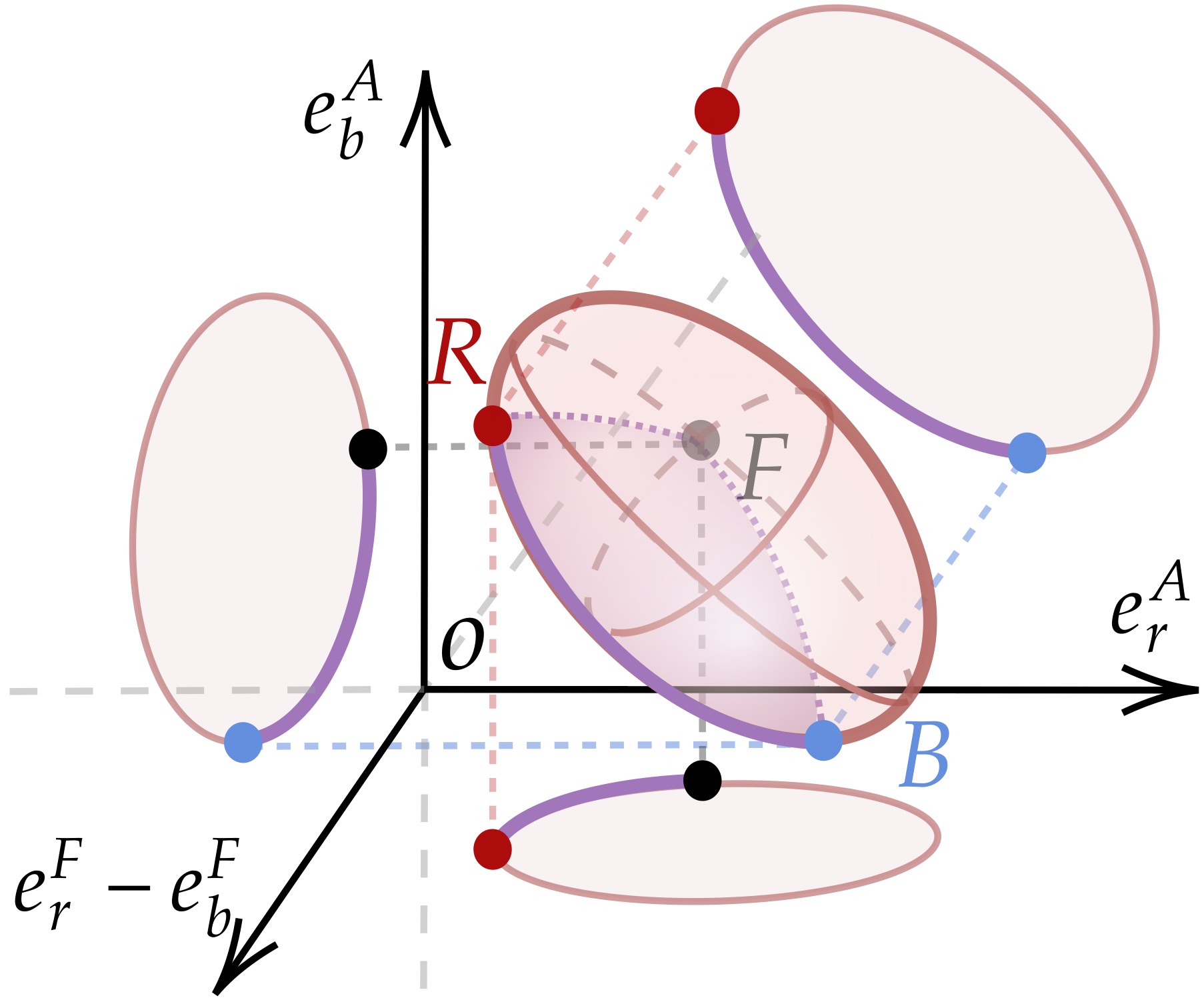}
} \hspace{1em}
\subfigure[$\E\subset\{e_r^F>e_b^F\}$ has kinks; $\sPF=\sF$]{
\includegraphics[width=0.425\linewidth]{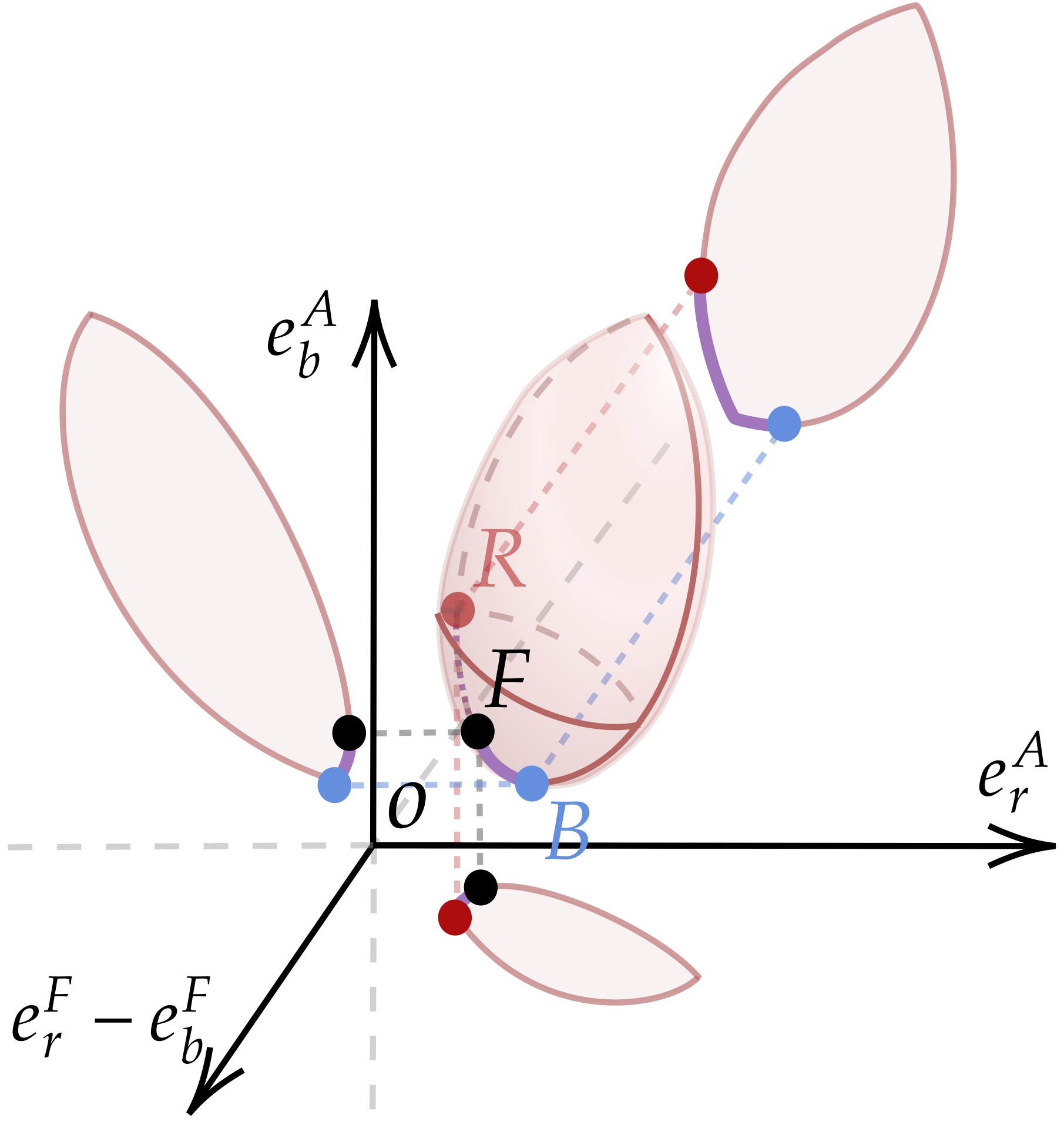}\hspace{.1cm}
}
\subfigure[\mbox{$R\in\{e_r^F>e_b^F\},\B\in\{e_r^F<e_b^F\}$; $\sPF\subsetneq\sF$}]{
\includegraphics[width=0.425\linewidth]{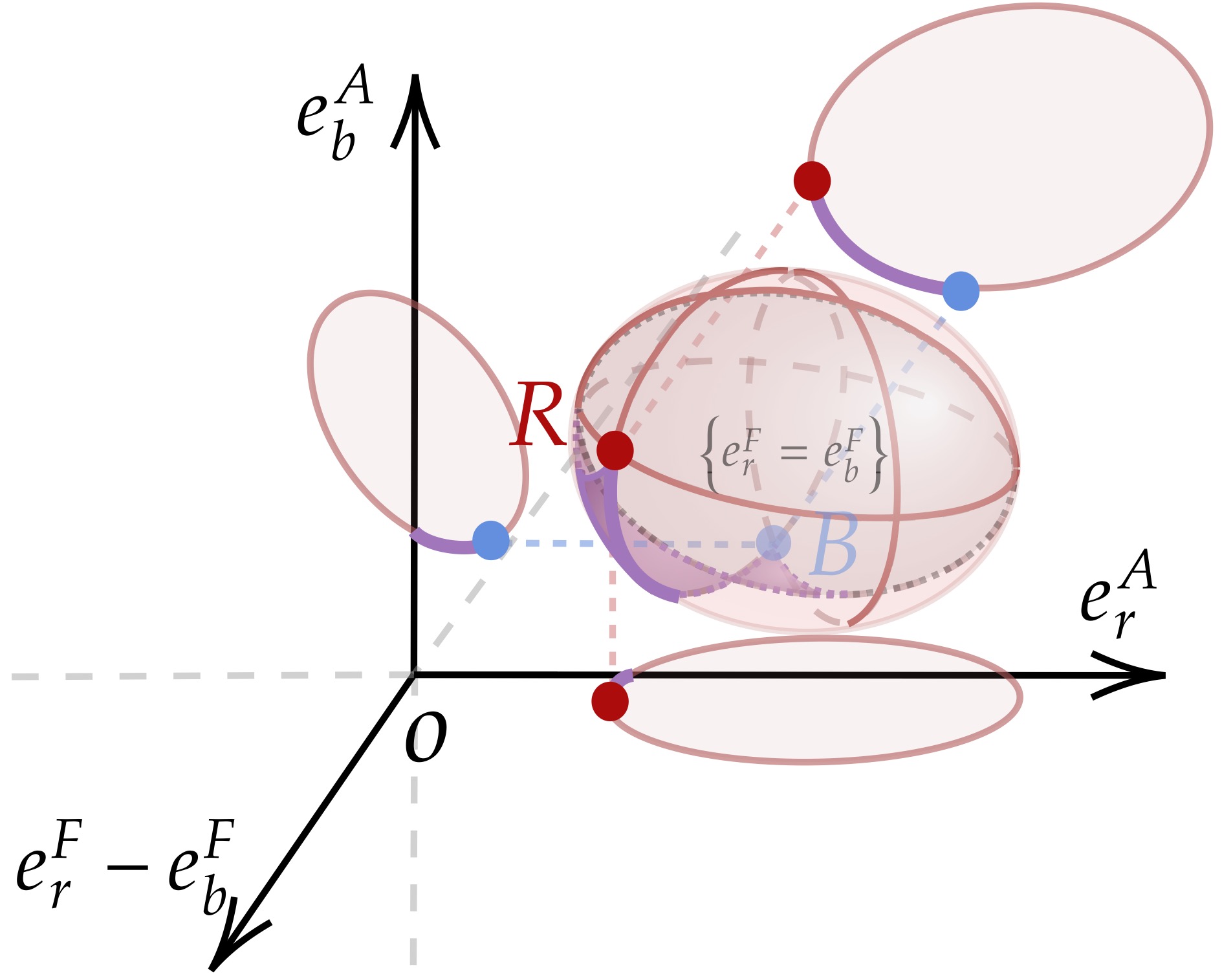}
}
\hspace{1em} \subfigure[\mbox{$\sPF \in\{e_r^F>e_b^F\},\E\cap\{e_r^F=e_b^F\}\neq\varnothing;\sPF\subsetneq\sF$}]{
\includegraphics[width=0.425\linewidth]{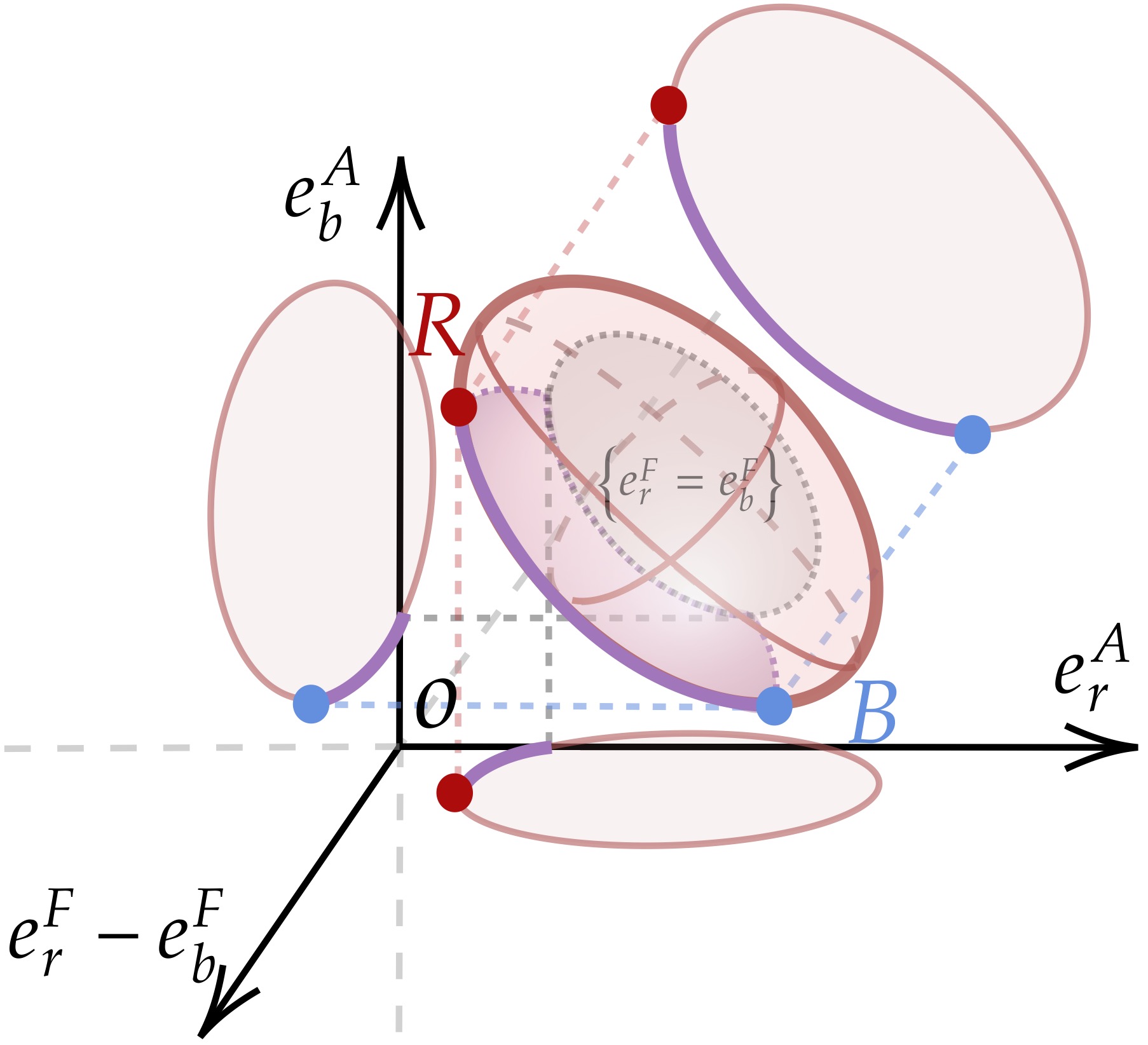}
}
\subfigure[$\sPF\subset\{e_r^F=e_b^F\}$; $\sF=\sPF$]{
\includegraphics[width=0.425\linewidth]{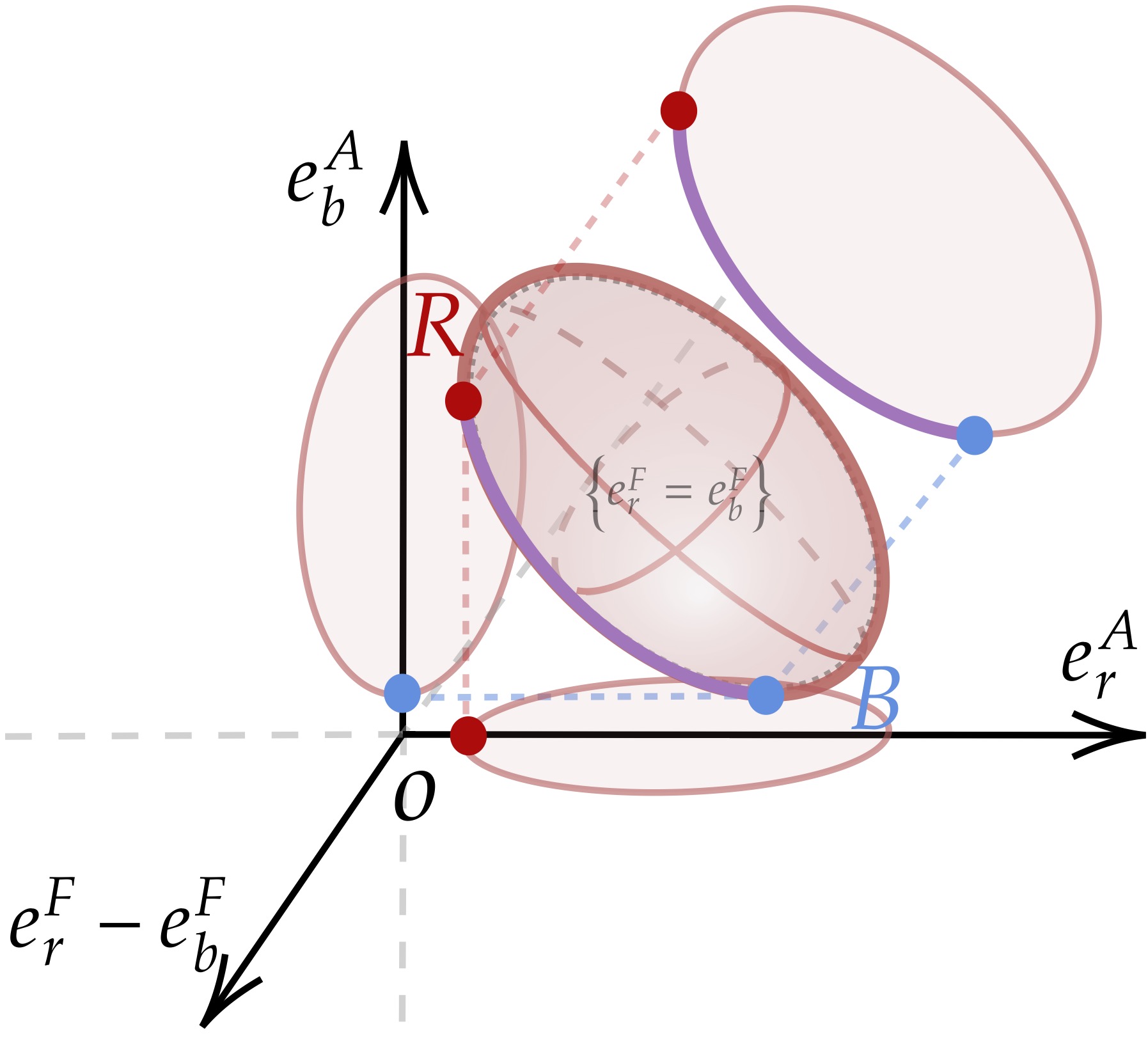}
}
\hspace{1em}
\subfigure[\mbox{$R,B\in\{e_r^F=e_b^F\},\sPF\nsubseteq\{e_r^F=e_b^F\}$; $\sPF\subsetneq\sF$}]{\includegraphics[width=0.425\linewidth]{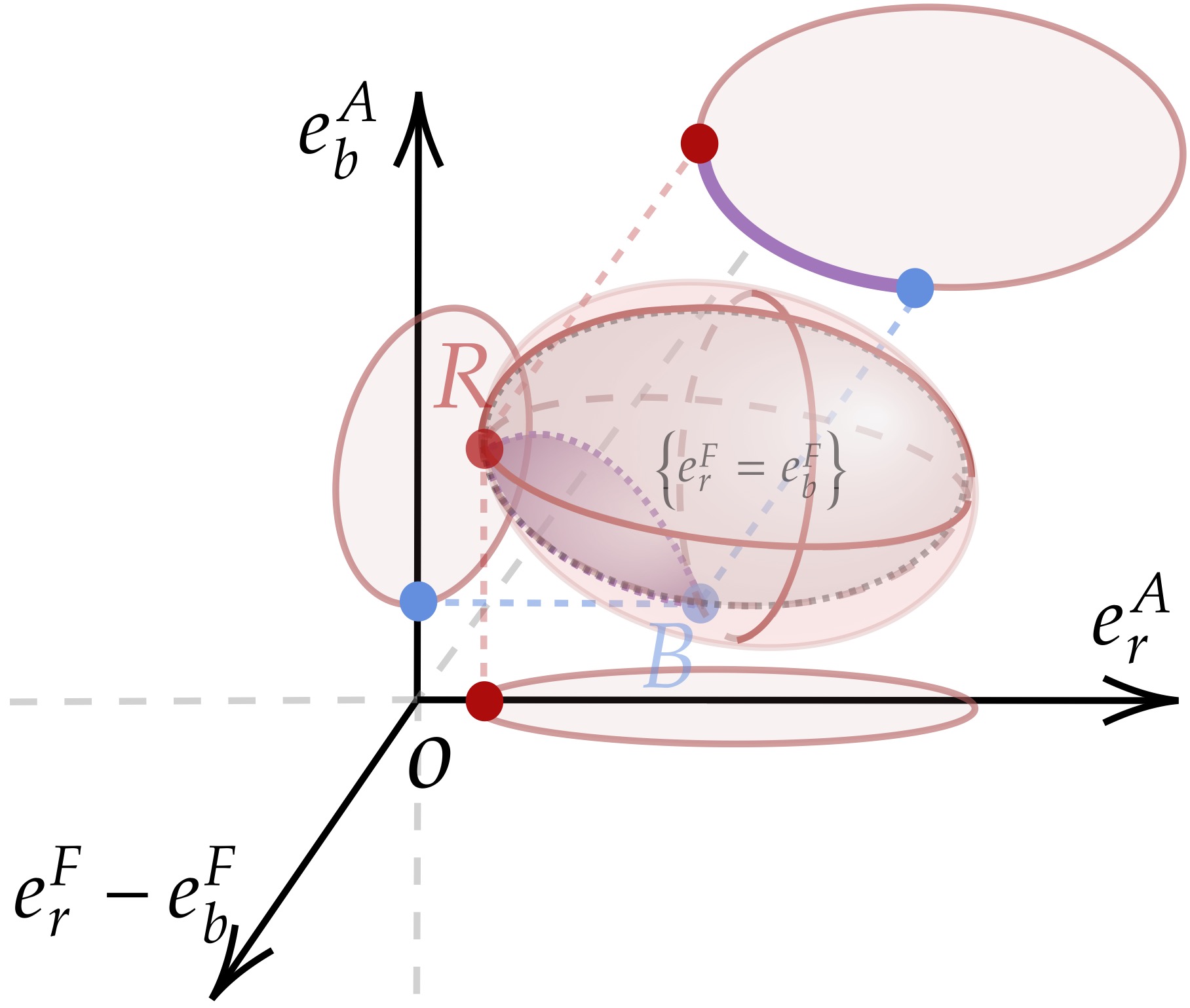}}
\caption{\footnotesize{$\E$ is the pink ellipsoid; its projections onto the two-dimensional coordinate planes are plotted as ellipses. $\sF$ is the purple shaded area/curve; $\sPF$ is the purple curve connecting~$R$ and $B$. Rear-surface curves on $\E$ are dashed; front-surface curves are solid. Illustrations inspired by analytic expressions.}}\label{fig:F}
\end{figure}

\begin{theorem}\label{thm:support_E}
\textit{Under Assumptions \ref{asm:moments}-\ref{asm:margin}, the support function of $\E$, denoted $\h_\E(\q)$, and its gradient with respect to $\q \in \bS^2$ are, respectively,
\begin{align}
\h_\E(\q)&\equiv\sup_{\e\in\E}\q^\intercal\e=\bEs[\q^\intercal \bLL_0 + \q^\intercal\Delta\bLL\one\{\q^\intercal \Delta\btheta(\X)>0\}\Big]\label{eq:support},\\
    \nabla_\q\h_{\E}(\q)&=
    \bEs\bigl[\bLL_0+\Delta\bLL\one\{\q^\intercal \Delta\btheta(\X)>0\}\bigr],\label{eq:ss}
\end{align}
and the set $\E$ is strictly convex. }
\end{theorem}
Eq.~\eqref{eq:ss} implies that $\ss_\E(\q)\equiv\{\e\in\E:\q^\intercal\e=\h_\E(\q)\}$, the support set of $\E$ in direction $\q\in\bS^2$, equals the singleton $\nabla_\q\h_{\E}(\q)$ \citep[Corollary 1.7.3]{sch93}. It yields that
\begin{align*}
    R =\bEs\bigl[\bLL_0+\Delta\bLL\one\{[-1,0,0]^\intercal\Delta\btheta(\X)>0\}\bigr],\hspace{.2cm}
    B =\bEs\bigl[\bLL_0+\Delta\bLL\one\{[0,-1,0]^\intercal\Delta\btheta(\X)>0\}\bigr].
\end{align*}

The next corollary, proven in \citet[Corollary 3.1]{liu:mol25v2}, characterizes precisely which algorithms yield expected losses on the boundary of $\E$:
\begin{corr}\label{corr:alg_as}
\textit{
    Let Assumptions~\ref{asm:moments}-\ref{asm:margin} hold.
    Let $\partial\E=\{\ss_\E(\q):\q\in\bS^2\}$ denote the boundary of $\E$.
    Then for any algorithm $a\in\mathcal{A}(\sX)$, $\e(a)\equiv\left(e^A_r(a),e^A_b(a),e^F_r(a)-e^F_b(a)\right)\in\partial\E$ if and only if for some $\q\in\bS^2$, $a(\X)=\mathds{1}\{\q^\intercal\Delta\btheta(X)>0\},~\bP^*_\X\text{-a.s.}$}
\end{corr}
Corollary \ref{corr:alg_as} echoes the generalized Neyman-Pearson Lemma \citep[Theorem 3.6.1]{leh:rom08}.\footnote{ To apply that theorem, in its notation set $m=1,\,c_1=0,\,\mu=\bP$, and let $f_1(x)=0$ for all $x\in\sX$ and $f_2(x)=\q^\intercal\Delta\btheta(x)$. The critical functions $\phi$ are our algorithms $a\in\mathcal{A}$.} One of its important consequences is that \emph{no} algorithm such that $a(\X)\in(0,1)$ for a set of $\X$ of positive probability can yield group risks on $\partial\E$.

As the goal of this section is to characterize $\sF$, given $\e^*=(\e_1^*,\e_2^*,\e_3^*)\in \E$, we define
\begin{align}
    \cC(\e^*) \equiv \{\e=(\e_1,\e_2,\e_3)\in\R^3: \e_1\le \e_1^*,\ \e_2\le \e_2^*,\ |\e_3|\le |\e_3^*|\},\label{eq:def_C}
\end{align}
the set of group-wise expected accuracy losses and fairness disparity $\e\in\R^3$---whether or not they are feasible---that are both weakly more accurate and weakly fairer than $\e^*$. 
In Lemma~\ref{lem:supp_C}, we show that the support function of $\cC(\e^*)$ equals:
\begin{align}
    \h_{\cC(\e^*)}(\q)\equiv\sup_{\e\in\cC(\e^*)}\q^\intercal\e=
\begin{cases}
\q_1 \e_1^* + \q_2 \e_2^* + |\q_3||\e_3^*| & \text{if } \q_1\ge 0,\ \q_2\ge 0,\\
+\infty & \text{otherwise.}
\end{cases}\label{eq:supportC}
\end{align}

Given an algorithm $a^*\in\mathcal A$ and $\e^*\equiv(e^A_r(a^*),e^A_b(a^*),e^F_r(a^*)-e^F_b(a^*))$, we establish in Theorem~\ref{thm:FequivLLMO} that $\e^*\in\sF$ if and only if the set $\cC(\e^*)$, depicted in the two panels of Figure \ref{fig:separate} as the shaded green hyper-rectangle corresponding to two different values of $\e^*$, can be properly separated from $\E$ \citep[for a definition of proper separation, see][p.12]{sch93}.
Panel (a) depicts a case where $\e^*\in\sF$ and proper separation occurs; panel (b) depicts a case where $\e^*\notin\sF$ and proper separation fails. 

\begin{figure}
\centering\captionsetup[subfigure]{font=footnotesize}
\centering
\hspace{1cm}
\subfigure[$\e^*$ is on $\sF$]{\includegraphics[width=0.35\linewidth]{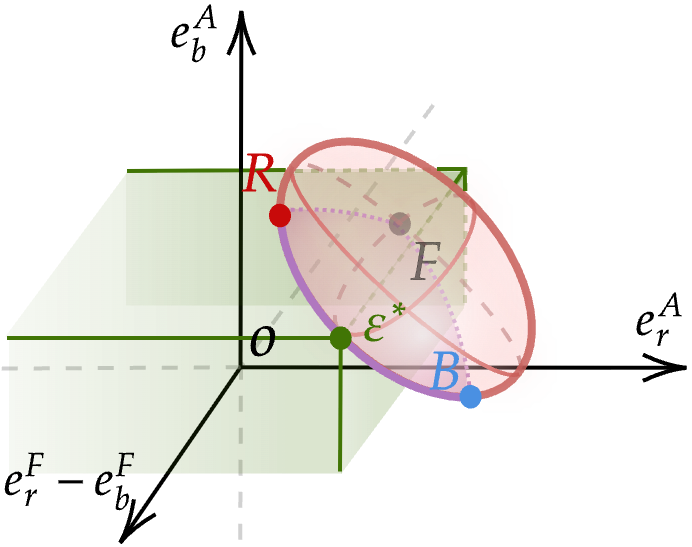}}
\hspace{1.5cm}
\subfigure[$\e^*$ is off $\sF$]{\includegraphics[width=0.35\linewidth]{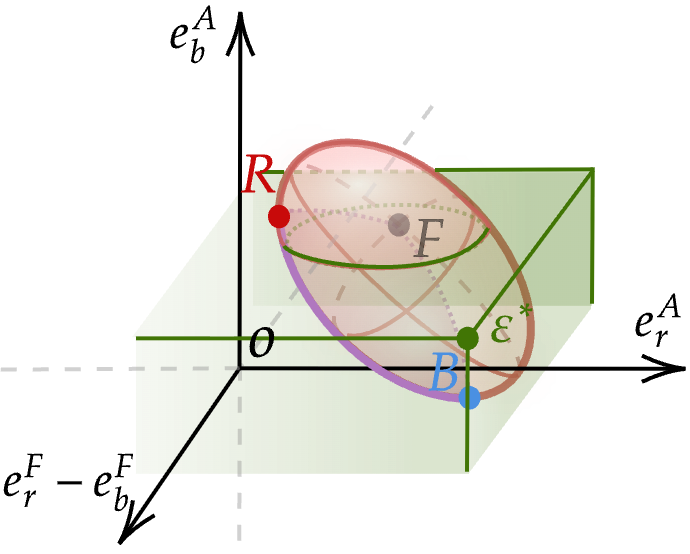}}
\caption{\footnotesize{The set $\cC(\e^*)$, which collects all FA improvements relative to $\e^*$, is the region shaded in green. Panel (a) shows an example where $\e^*$ lies on the frontier and there exists a hyperplane that properly separates $\cC(\e^*)$ and $\E$, whereas in panel (b) $\e^*$ is not on the frontier and no hyperplane can separate $\cC(\e^*)$ and $\E$. 
}}
\label{fig:separate}
\end{figure}

\begin{theorem}\label{thm:FequivLLMO}
\textit{Let $\e^*=(\e_1^*,\e_2^*,\e_3^*)\in \E$ and define
$\widetilde{\bS}^2\equiv\{q\in \bS^2:q_1\ge 0,\ q_2\ge 0\}$. Under Assumptions~\ref{asm:moments}-\ref{asm:margin}, the following are equivalent:
\begin{enumerate}[label=(\roman*)]
\item\label{thm:FequivLLMO_e_in_sF} $\e^*\in \sF$;
\item\label{thm:FequivLLMO_E_single_inters_C} $\E\cap\cC(\e^*)=\{\e^*\}$;
\item\label{thm:FequivLLMO_support} $\min_{q\in \widetilde{\bS}^2}
\left( \h_{\cC(\e^*)}(q)+\h_{\E}(-q) \right)=0$.
\end{enumerate}}
\end{theorem}

Using the equivalence $\ref{thm:FequivLLMO_e_in_sF}\Leftrightarrow\ref{thm:FequivLLMO_support}$ in Theorem~\ref{thm:FequivLLMO}, similarly to \citet[Proposition 3.3]{liu:mol25v2}, we have a closed-form expression for the FA-frontier that does not rely on the position of $\E$ in $\R^3$, sidestepping the complexity of knowing its geometry:
\begin{equation}\label{eq:support_frontier}
\sF=\left\{\e^*\in\E:~\min_{\q \in \tilde\bS^2} \left[\h_{\cC(\e^*)}(\q) + \h_{\E}(-\q)\right] = 0\right\}.
\end{equation}

When the same loss function is used to measure accuracy and fairness, as in \citetalias{lia:lu:mu:oku24}'s main text and in \citet{liu:mol25v2}, it is simple to characterize when the FA-frontier coincides with the Pareto frontier in \eqref{eq:sPF}.\footnote{In that case, $\ell^A=\ell^F$ and $e_g^A=e_g^F,\,g\in\{r,b\}$, so $\E$ can be defined in $\R^2$ and \citetalias{lia:lu:mu:oku24} show that the Pareto and the FA-frontiers coincide if and only if $R$ and $B$, defined in \eqref{eq:define_R_B_F}, lie on opposite sides of the $45^o$ line.}
When different loss functions are used, such characterization is harder to obtain, and \citetalias{lia:lu:mu:oku24} (Proposition O.1) provide only a sufficient condition for it.
Working with the feasible set $\E\subset\R^3$ and its support function, we are able to provide a simple, necessary and sufficient, testable condition  under which $\sPF=\sF$.
The characterization is valuable because when $\sF$ does not coincide with $\sPF$, the tradeoff between fairness and accuracy becomes substantial: a policymaker who values fairness may choose to deploy an algorithm that is Pareto dominated in accuracy loss, to ameliorate its performance relative to fairness loss.
To derive the result, it is useful to distinguish between cases in which $\E$ has kinks, and cases where it does not, through the following condition.
\begin{asm}[No kinks]
    \label{asm:no-kink}
    \textit{For any $\q,\v\in\bS^2,\q\neq\v$, 
    \begin{align}
\bP^*\big(\q^\intercal\Delta\btheta(X)>0,v^\intercal\Delta\btheta(\X)<0\big)+\bP^*\big(\q^\intercal\Delta\btheta(X)<0,v^\intercal\Delta\btheta(\X)>0\big)>0.\label{eq:no_kinks_cond}
\end{align}
}
\end{asm} 
\noindent Assumption \ref{asm:no-kink} is a necessary and sufficient condition under which $\E$ has no kinks (see Supplemental Appendix B.2.3 in \citeauthor{bon:mag:mau12}, \citeyear{bon:mag:mau12}). It yields injectivity of the support map of $\E$, $\ss_\E(\q) = \arg \max_{\e \in \E} \q^\intercal \e$ (see Theorem~\ref{thm:support_E} and the discussion following it). 
The next result characterizes when $\sPF$ and $\sF$ are exactly the same.
\begin{theorem}\label{thm:Pareto_equal_FA}
   \textit{Under Assumptions \ref{asm:moments}-\ref{asm:margin}, $\sPF=\sF$ if and only if for all $\e^*\in\sF$ there exists $\alpha\in[0,1]$ such that for $\q(\alpha)=\tfrac{[\alpha,1-\alpha,0]^\intercal}{\Vert[\alpha,1-\alpha,0]^\intercal\Vert}$, $\e^*=\arg\min_{\eta\in \E} \q(\alpha)^\intercal \eta$.
   If Assumption~\ref{asm:no-kink} also holds, then $\sPF=\sF$ if and only if $\sPF\subset\{\e_3=0\}$.
}
\end{theorem}
Intuitively, $\sPF=\sF$ when every point on the FA-frontier is Pareto-optimal. 
When each point on the boundary of $\E$ has a unique supporting direction (which is guaranteed by Assumption~\ref{asm:no-kink}), this happens if and only if $\sPF\subset\{\e_3=0\}=\{\e:e_r^F=e_b^F\}$; in this case, all points along the Pareto frontier are fair and there is no conflict between improving fairness and improving accuracy.
On the other hand, when $\E$ has kinks, one may have $\sPF=\sF$ even though $\sPF\nsubseteq\{\e_3=0\}$.
The complexity of the geometry of $\sF$ in relation to the geometry of $\E$ and the usefulness of our Theorem \ref{thm:Pareto_equal_FA} are illustrated in Figure~\ref{fig:F}. 
If $\E$ has no kinks and lies completely to one side of the hyperplane $\{\e_3=0\}$ (Panel (a), where $\E\subset\{\e_3>0\}=\{e_r^F>e_b^F\}$) or $R$ and $B$ lie on opposite sides of that hyperplane (Panel (c)), then $\sPF\subsetneq\sF$ since $\sPF \not\subset \{\e_3=0\}$. 
On the other hand, if $\E$ has kinks properly facing the origin, then it can be that $\sPF=\sF$ even when $\E$ lies completely to one side of the hyperplane $\{\e_3=0\}$ (Panel (b)).
If $\E$ crosses the hyperplane $\{\e_3=0\}=\{e_r^F=e_b^F\}$ with $\sPF$ fully contained in one of $\{\e_3>0\}$ or $\{\e_3<0\}$, then $\sPF\subsetneq\sF$ regardless of the presence or absence of kinks (Panel (d)); see Corollary~\ref{corr:crossing-one-sided-pareto}.
If $\E$ has no kinks and $\sPF\subset\{e_3=0\}$, then $\sPF=\sF$ (Panel (e)), yet having only $\{R,B\} \in\{e_3=0\}$ but not the entire Pareto frontier may yield $\sPF\subsetneq\sF$ (Panel (f)).
Hence, the characterization in Theorem~\ref{thm:Pareto_equal_FA} is particularly useful to avoid case-by-case distinctions and pre-tests.

\section{Selectively observed Labels}\label{sec:identif_and_estim}
\newsavebox{\LMVstrikebox}
\newcommand{\LMVdeletedmath}[1]{%
  \mbox{\begingroup
    \sbox{\LMVstrikebox}{\sout{$\displaystyle #1$}}%
    \ifdim\wd\LMVstrikebox>.90\linewidth
      \resizebox{.90\linewidth}{!}{\usebox{\LMVstrikebox}}%
    \else
      \usebox{\LMVstrikebox}%
    \fi
  \endgroup}%
}

Having fully characterized the FA-frontier in the idealized case when 
$(\Y^*,\X,\G,\Z)\sim\bP^*$ is observed, we turn to our goal of providing identification results and inference procedures for $\sF$ when only $(\Y,\X,\G,\Z)\sim\bP$ is observed, with $\Y=\Z\Y^*$ the selected labels and $\Z$ a selection indicator.
This problem has been studied, e.g., by \citet{lak:kle:les:lud:mul17,kle:lak:les:lud:mul17,ram:cos:ken25,kha:tam:yao}.
When evaluating whether algorithmic policies can improve upon human decision-making, selectively observed labels are compounded by the possibility that $\Z$ is determined by a decision maker (e.g., a judge) who observes information beyond the recorded covariates made available to the algorithm.
In this case, for a specific choice of loss functions, we report partial identification results.
We then provide point identification results and statistical inference procedures under conditional missing at random and strict overlap assumptions.

\subsection{Partial Identification of the FA-Frontier}\label{subsec:PI-CR-SP}
In this section, we allow $\Z$ to depend on variables unavailable to the algorithms in $\mathcal A(\sX)$.
This creates substantial challenges for identification and inference, because not only the feasible set $\E$ is partially identified, but even for a fixed algorithm $a^*\in\mathcal A(\sX)$, the vector of expected losses $\e^*\equiv\e(a^*):=(e_r^A(a^*),e_b^A(a^*),e_r^F(a^*)-e_b^F(a^*))$ is also only partially identified.
Since efficiency is a relational property of $(\e^*,\E)$, extra care needs to be taken to couple the possible values of $\E$ and $\e^*$ associated with each admissible distribution for the unobserved labels to obtain a sharp characterization of the FA-frontier that exhausts all information from the observed distribution and maintained assumptions.

To make progress on this task, we specialize our analysis to the case where $\Y^*$ is a binary scalar variable, the accuracy loss measures classification error, $\ell^A(d,y) = \one\{d \neq y\}$, and the fairness loss measures statistical parity, $\ell^F(d,y) = \one\{d=1\}$.
Write $\Y=\Z\Y^*$ for the observed binary outcome. In this case the fairness coordinate depends only on the decision rule and the observable distribution of $(\X,\G)$, so lack of point identification comes entirely from the selectively observed outcomes entering the two accuracy coordinates.

Recall that, at the true distribution $\bP^*$ and under Assumptions~\ref{asm:moments}-\ref{asm:margin}, an algorithm $a\in\mathcal A(\sX)$ can generate expected losses on the FA-frontier only if $a(X)\in\{0,1\}$, $\bP_X$-a.s.
Hence, for the frontier membership problem below, we restrict attention to a fixed candidate $a^*\in\mathcal A(\sX)$ with $a^*(X)\in\{0,1\}$, $\bP_X$-a.s. and compare it to the full class $\mathcal A(\sX)$.

Denote the unobserved conditional distribution of $\Y^*$ when $\Z=0$ by
\begin{align*}
    \lambda_g^*(x)\equiv \bP^*[\Y^*=1| \X=x,\G=g,\Z=0],
\qquad g\in\{r,b\},
\end{align*}
and $\lambda^*(x)\equiv[\lambda_r^*(x),\lambda_b^*(x)]^\intercal$. 
Without restrictions on the selection process, $\lambda^*$ is unknown. 
As $\Y^*\in\{0,1\}$, every measurable map $\lambda\equiv(\lambda_r,\lambda_b):\sX\to[0,1]^2$ is compatible with $\bP$, the observed law of $(\Y,\X,\G,\Z)$; i.e., $\lambda$ yields a law of $(\Y^*,\X,\G,\Z)$ that reproduces $\bP$ and satisfies $\Y=Z\Y^*$ almost surely.
Let $\Lambda\equiv \left \{ \lambda: \sX \to [ 0,1]^2 \right \}$ denote this class of functions. 

We next derive the sharp identification region for $\h_\E(\q)$.
To do so, we need to introduce some notation.
Define the observable functions
\begin{align}
A_0(x)
&\equiv
\bE\left[
\left.
\begin{pmatrix}
\tfrac{Z\Y\one\{\G=r\}}{\mu_r},&
\tfrac{Z\Y\one\{\G=b\}}{\mu_b},&
0
\end{pmatrix}^\intercal
\right|\X=x
\right], \label{eq:PI-A0-main}\\
A_1(x)
&\equiv
\bE\left[
\left.
\begin{pmatrix}
\tfrac{(1-\Z\Y)\one\{\G=r\}}{\mu_r},&
\tfrac{(1-\Z\Y)\one\{\G=b\}}{\mu_b},&
\tfrac{\one\{\G=r\}}{\mu_r}-\tfrac{\one\{\G=b\}}{\mu_b}
\end{pmatrix}^\intercal
\right|\X=x
\right], \label{eq:PI-A1-main}\\
B(x)
&\equiv
\bE\left[
\left.
\begin{pmatrix}
\tfrac{(1-\Z)\one\{\G=r\}}{\mu_r} & 0\\
0 & \tfrac{(1-\Z)\one\{\G=b\}}{\mu_b}\\
0 & 0
\end{pmatrix}
\right|\X=x
\right].\label{eq:PI-B-main}
\end{align}
Lemma~\ref{lem:PI-support-functions} shows that the functions $\btheta_d(\cdot)$ in Definition~\ref{def:ideal-labels} associated with distribution $\lambda\in\Lambda$ satisfy $\btheta_0(x;\lambda)=A_0(x)+B(x)\lambda(x)$ and $\btheta_1(x;\lambda)=A_1(x)-B(x)\lambda(x)$.
Write $\Delta\btheta(x;\lambda)\equiv\btheta_1(x;\lambda)-\btheta_0(x;\lambda)=A_1(x)-A_0(x)-2B(x)\lambda(x)$.
The feasible set associated with distribution $\lambda$ and its support function are, respectively,
\begin{align}
    \E(\lambda)&=\left\{\bE\left[(1-a(X))\btheta_0(X;\lambda)+a(X)\btheta_1(X;\lambda)\right],\,a\in\mathcal{A}(\sX)\right\},\label{eq:PI_Elambda}\\
    \h_{\E(\lambda)}(\q)&=\bE\left[
\max\left\{
\q^\intercal\btheta_0(X;\lambda),
\q^\intercal\btheta_1(X;\lambda)
\right\}
\right].\label{eq:PI-hE-main}
\end{align}
Accordingly, the sharp identified set for $\h_\E(\q)$ in direction $\q$ is $\{\h_{\E(\lambda)}(\q):\lambda\in\Lambda\}$. 

The next assumption on the observable distribution $\bP$ requires that an analog of the margin condition in Assumption \ref{asm:margin} holds for two key distributions compatible with the data: the one where $\bP^*[\Y^*=1| \X=x,\G=g,\Z=0]=0$ for $g\in\{r,b\}$ and the one where $\bP^*[\Y^*=1| \X=x,\G=g,\Z=0]=1$ for $g\in\{r,b\}$.
Define $\lambda_j=(j,j)^\intercal,j\in\{0,1\}$.
\begin{asm}[Endpoint margins]\label{asm:endpoint-margin}
\textit{There exists a constant $m\in(0,1]$ such that, for some $C>0$ and every $\delta>0$, $\sup_{\q\in\bS^2}
\bP_X\!\left(\left|\q^\intercal\Delta\btheta(X;\lambda_j)\right|\le\delta\right)
\le C\delta^m$, $j\in\{0,1\}$.}
\end{asm}
\noindent{Assumption~\ref{asm:endpoint-margin} does \emph{not} require that every observationally compatible completion satisfies Assumption~\ref{asm:margin} but does require that for $y\in\{0,1\}$ and $g\in\{r,b\}$},
\begin{align}
\bP(Z=1,Y=y\mid X)>0\qquad\text{and}\qquad
\bP(G=g\mid X)>0,\qquad \bP_X\text{-a.s.}
\label{eq:PI-endpoint-support}
\end{align}
It also requires $A_{1,3}(X)\ne0$, $\bP_X$-a.s.; see Lemma \ref{lem:PI-endpoint-consequences}.
Hence, positive probability covariate strata without observed labels, or with only one observed outcome class, are excluded. 
Importantly, the conditions in \eqref{eq:PI-endpoint-support} are testable implications of Assumption~\ref{asm:endpoint-margin}.

Next, we characterize the distribution in $\Lambda$ that is \emph{least favorable} when looking for an LDA algorithm to the binary candidate $a^*$.
Define
\begin{equation}
\lambda(x;a^*)\equiv a^*(x)\lambda_1.
\label{eq:PI-favorable-completion}
\end{equation}
This distribution imputes for every missing label the status quo algorithm's decision, hence making the classification loss for such imputed observations equal zero. 
Lemma~\ref{lem:PI-favorable-comparison} shows that $\lambda(\cdot;a^*)$ makes both accuracy losses generated by $a^*$ simultaneously hardest to improve by any challenger $a\in\mathcal{A}(\sX)$ (hence the ``least favorable'' descriptor):\footnote{The third coordinate, which measures fairness, is independent of $\lambda\in\Lambda$. We write $\e_3(a)=\bE[a(X)A_{1,3}(X)]$.}
\begin{align*}
    \e_j(a^*;\lambda(\cdot))-\e_j(a;\lambda(\cdot))\ge\e_j(a^*;\lambda(\cdot;a^*))-\e_j(a;\lambda(\cdot;a^*)),\quad j=1,2.
\end{align*}
For $\lambda\in\Lambda$, recall $\E(\lambda)$ in \eqref{eq:PI_Elambda} and define its FA-frontier as
\begin{align}
\sF(\lambda)&\equiv\left\{\e\in\E(\lambda):\nexists \e'\in\E(\lambda) \text{ such that } \e'\succ_{FA}\e\right\}.\label{eq:sF_lambda}
\end{align}
\noindent{We} apply Theorem~\ref{thm:FequivLLMO} to $\E(\lambda(\cdot;a^*))$ and $\sF(\lambda(\cdot;a^*))$: Lemma~\ref{lem:PI-favorable-comparison} shows that $\e(a^*;\lambda)\in\sF(\lambda)$ for some $\lambda\in\Lambda$ if and only if $\e(a^*;\lambda(\cdot;a^*))\in\sF(\lambda(\cdot;a^*))$.
Lemma~\ref{lem:PI-endpoint-consequences} establishes that, under Assumption~\ref{asm:endpoint-margin}, the margin condition in Assumption~\ref{asm:margin} holds at the distribution completed by $\lambda(\cdot;a^*)$; for brevity, we refer to this as Assumption~\ref{asm:margin} holding at $\lambda(\cdot;a^*)$.

Theorem \ref{prop:PI-frontier-corner} shows how to check whether $a^*$ can generate expected losses on the FA-frontier for \emph{some} distribution of the unobserved labels, through a finite dimensional optimization problem that involves only functionals of the observed distribution $\bP$.
\begin{theorem}\label{prop:PI-frontier-corner}
\textit{Let Assumptions~\ref{asm:moments} and \ref{asm:endpoint-margin} hold. Suppose $Y^*\in\{0,1\}$, $\ell^A(d,y)=\one\{d\ne y\}$, and $\ell^F(d,y)=\one\{d=1\}$. Fix $a^*\in\mathcal A(\sX)$ with $a^*(X)\in\{0,1\}$, $\bP_X$-a.s., and use the completion $\lambda(\cdot;a^*)$ in \eqref{eq:PI-favorable-completion}. For $\q\in\R^3$ with $q_1,q_2\ge0$, define
\begin{align}
J(\q;a^*)\equiv{}&
\bE\!\left[\left\{(2a^*(X)-1)\q^\intercal
\Delta\btheta(X;\lambda(\cdot;a^*))\right\}_+\right]+2\left\{-q_3\e_3(a^*)\right\}_+,
\label{eq:PI-J-repaired}
\end{align}
where $t_+=\max\{t,0\}$. Then
\begin{equation}
\exists\lambda\in\Lambda:\ \e(a^*;\lambda)\in\sF(\lambda)
\quad\Longleftrightarrow\quad
\min_{\q\in\widetilde{\bS}^2}J(\q;a^*)=0.
\label{eq:PI-corner-main}
\end{equation}
The criterion $J(\q;a^*)$ is nonnegative and continuous in $\q$, and its minimum is attained. Whenever \eqref{eq:PI-corner-main} holds, $\e(a^*;\lambda(\cdot;a^*))\in\sF(\lambda(\cdot;a^*))$ and Assumption~\ref{asm:margin} holds at $\lambda(x;a^*)$.}

\end{theorem}
Theorem~\ref{prop:PI-frontier-corner} reduces the search over $\Lambda$ to evaluation at a single observable, $a^*$-dependent completion. 
Moreover, the equivalence in \eqref{eq:PI-corner-main} remains valid if the existential quantifier is restricted to completions satisfying Assumption~\ref{asm:margin}, because $\lambda(\cdot;a^*)$ satisfies that assumption. Hence, the characterization is sharp in the sense that it extracts all information in the distribution of the observable variables and in the maintained assumptions.

\subsubsection{An algorithm that improves uniformly over $\lambda\in\Lambda$}
When \eqref{eq:PI-corner-main} in Theorem~\ref{prop:PI-frontier-corner} fails to hold, for each $\lambda\in\Lambda$, there is an algorithm $\widetilde{a}\in\E(\lambda)$ that FA-dominates $a^*$. 
The common least-favorable characterization of $\lambda(\cdot;a^*)$ yields a much stronger conclusion: the same algorithm $\widetilde a$ can strictly reduce both accuracy losses for every $\lambda\in\Lambda$ while weakly reducing statistical disparity.

To characterize the algorithm $\widetilde a$, recall $J(\q;a^*)$ in \eqref{eq:PI-J-repaired} and define
\begin{equation}
\Jstar\equiv
\min_{\substack{\q\in\R^3:\ q_1,q_2\ge0,\ q_1+q_2=1}}J(\q;a^*),
\qquad
\widetilde\q\in
\arg\min_{\substack{\q\in\R^3:\ q_1,q_2\ge0,\ q_1+q_2=1}}J(\q;a^*).
\label{eq:PI-uniform-program}
\end{equation}
We can now define
\begin{equation}
\widetilde a(x)\equiv
\one\!\left\{\widetilde\q^\intercal
\Delta\btheta(x;\lambda(\cdot;a^*))<0\right\}.
\label{eq:PI-uniform-algorithm}
\end{equation}

\begin{prop}\label{prop:PI-uniform-LDA}
\textit{Maintain the assumptions and binary $a^*$ of Theorem~\ref{prop:PI-frontier-corner}. 
Then the minimum in \eqref{eq:PI-uniform-program} is attained, and
\begin{align}
\Jstar=0&\quad\Longleftrightarrow\quad
\min_{\q\in\widetilde{\bS}^2}J(\q;a^*)=0.
\label{eq:PI-uniform-gap}
\end{align}
For every $\lambda\in\Lambda$,
\begin{align}
\e_j(\widetilde a;\lambda)&\le\e_j(a^*;\lambda)-\Jstar,
\quad j\in\{1,2\},
&
|\e_3(\widetilde a)|&\le|\e_3(a^*)|.
\label{eq:PI-uniform-improvement}
\end{align}
Furthermore, $\e(\widetilde a;\lambda(\cdot;a^*))\in\sF(\lambda(\cdot;a^*))$, and
\begin{equation}
\Jstar=
\max_{\substack{a\in\mathcal A(\sX):\ |\e_3(a)|\le|\e_3(a^*)|}}
\ \inf_{\lambda\in\Lambda}\ \min_{j\in\{1,2\}}
\{\e_j(a^*;\lambda)-\e_j(a;\lambda)\}.
\label{eq:PI-uniform-maximin}
\end{equation}
The maximum is attained by the binary rule \eqref{eq:PI-uniform-algorithm}.}
\end{prop} 
The optimization problem in \eqref{eq:PI-uniform-program} is convex and has two free arguments: $q_1\in[0,1]$ and $q_3\in\R$. 
The function $J(\q;a^*)$ depends on the observed distribution $\bP$ and the completion $\lambda(\cdot;a^*)$. The improving algorithm in~\eqref{eq:PI-uniform-algorithm} takes as input the covariates $\X$ but not group identity $\G$. 
When the minimum in \eqref{eq:PI-corner-main} is positive, $\widetilde a$ strictly improves both accuracy coordinates by at least $\Jstar>0$ uniformly over $\lambda\in\Lambda$, while weakly improving statistical parity. 
Eq.~\eqref{eq:PI-uniform-maximin} gives that $\Jstar$ is the largest uniform accuracy improvement that can be guaranteed for both groups, among all algorithms that are at least as fair as $a^*$.
Hence, $\widetilde a$ is optimal in the maximin sense: it maximizes the worst-case accuracy gain across groups and $\lambda\in\Lambda$.

\subsubsection{The closed FA-frontier envelope}
We next characterize the \emph{FA-frontier envelope}, which we define using
\eqref{eq:sF_lambda} as
\begin{align}
\sF^\exists\equiv \bigcup_{\lambda\in\Lambda}\sF(\lambda).
\label{eq:sF_envelope}
\end{align}
This set collects the expected loss vectors $\e$ for which there exists an algorithm $a\in\mathcal{A}(\sX)$ and an admissible completion $\lambda\in\Lambda$ such that
$\e=\e(a;\lambda)\in\sF(\lambda)$.
Crucially, the same completion generates the expected losses and places them on the corresponding FA-frontier.
We provide a finite-dimensional characterization of $\cl(\sF^\exists)$,  where $\cl$ denotes the Euclidean closure.
The closure accounts for the fact that frontier membership need not be preserved in limits as the completion varies. 
Intuitively, an improvement in one coordinate may require a small worsening in another, with that worsening disappearing in the limit.

For a given $q\in\widetilde{\bS}^2$, collect the expected loss vectors generated by measurable binary algorithms that minimize the $q$-weighted losses pointwise for some $\lambda\in\Lambda$ in the set\footnote{The minimization in the definition of $\cM_q$ is a sign convention that we use because the first two coordinates are losses (we could alternatively take $\arg\max_d -\q^\intercal\btheta_d(x;\lambda)$ for $\q\in\widetilde{\bS}^2$)}
\begin{multline}
\cM_q\equiv\Big\{\bE\left[a(X)\btheta_1(X;\lambda)+(1-a(X))\btheta_0(X;\lambda)\right]:\\
a(x)\in\arg\min_{d\in\{0,1\}}\q^\intercal\btheta_d(x;\lambda),\, \bP_\X\text{-a.s.},\,\lambda\in\Lambda\Big\}.\label{eq:mathfrak_A_q}
\end{multline}
In Lemma~\ref{lem:frak_A_q}, we show that for every $q\in\widetilde{\bS}^2$, the set $\cM_q$ is nonempty, compact, and convex.
In Lemma~\ref{lem:continuous_h_frakAq} we show that its support function, $\h_{\cM_q}(\v),\v\in\bS^2$, is given by
    \begin{align*}
        \h_{\cM_\q}(\v)&=\bE\left[\max\{\kappa_0(\v,\q;X),\kappa_1(\v,\q;X)\}\right],
    \end{align*}
where $\kappa_0(\v,\q;x)$ and $\kappa_1(\v,\q;x)$, formally defined in \eqref{eq:kappa0}-\eqref{eq:kappa1}, are known functions of $A_0(x),A_1(x),B(x)$.
Lemma~\ref{lem:continuous_h_frakAq} also establishes Hausdorff continuity of $q\mapsto\cM_q$ on $\widetilde{\bS}^2$. We obtain the following characterization of $\cl(\sF^\exists)$.
\begin{theorem}\label{thm:sF_envelope}
\textit{Let Assumptions~\ref{asm:moments} and \ref{asm:endpoint-margin} hold. Suppose $Y^*\in\{0,1\}$, $\ell^A(d,y)=\one\{d\ne y\}$, and $\ell^F(d,y)=\one\{d=1\}$. Then the closure of the FA-frontier envelope is compact and
\begin{equation}
\cl(\sF^\exists)
=\bigcup_{\q\in\widetilde{\bS}^2}
\left(\cM_q\cap\{\e\in\R^3:q_3\e_3\ge0\}\right).
\label{eq:define_sF_exists}
\end{equation}
For $\tilde\e\in\R^3$, define the criterion function
\begin{equation}
T_{\mathrm{env}}(\tilde\e)
\equiv\inf_{\q\in\widetilde{\bS}^2:\,q_3\tilde\e_3\ge0}
\left[\sup_{\v\in\bS^2}
\{\v^\intercal\tilde\e-\h_{\cM_q}(\v)\}\right]_+.
\label{eq:define:T_env_representation}
\end{equation}
The infimum in \eqref{eq:define:T_env_representation} is attained, and
\begin{equation}
\tilde\e\in\cl(\sF^\exists)
\quad\Longleftrightarrow\quad T_{\mathrm{env}}(\tilde\e)=0.
\label{eq:PI-envelope-closed-criterion}
\end{equation}
For each $\tilde\e\in\cl(\sF^\exists)$, there exist a binary algorithm $a$, completion $\lambda$ with $\e(a;\lambda)=\tilde\e$, and sequence $t_n\downarrow0$ such that the completion
$\lambda_{t_n}=(1-t_n)\lambda+t_n\lambda(\cdot;a)$ satisfies Assumption~\ref{asm:margin},
\begin{equation}
\e(a;\lambda_{t_n})\in\sF(\lambda_{t_n}),
\qquad\text{and}\qquad \e(a;\lambda_{t_n})\to\tilde\e.
\label{eq:PI-envelope-approximation}
\end{equation}
The margin exponents $m$ and finite constant $C$ in Assumption~\ref{asm:margin} for these completions may depend on $n$.
Moreover, $\e_3(a;\lambda_{t_n})=\tilde\e_3$, and $\e_j(a;\lambda_{t_n})\le\tilde\e_j$, $j=1,2$.
}
\end{theorem}
The characterization in Theorem~\ref{thm:sF_envelope} is sharp pointwise, up to closure: every $\tilde\e$ such that $T_{\mathrm{env}}(\tilde\e)=0$ can be attained or approximated arbitrarily closely by points on FA-frontiers under admissible completions.
However, different points in the envelope may require different completions; hence, a surface contained in the envelope, even one with the geometry of an FA-frontier, need not equal $\sF(\lambda)$ for any single $\lambda\in\Lambda$ (this phenomenon is similar to the fact that not every cumulative distribution function (CDF) in a pointwise sharp envelope of admissible CDFs is itself admissible; see \citet[Example 2.1]{mol20}).

\begin{remark}
In work in progress, we develop a debiased machine learning estimator of the criterion in \eqref{eq:PI-J-repaired} and an inference procedure for the binary candidate null in \eqref{eq:PI-corner-main}. 
We also explore extending Theorem~\ref{prop:PI-frontier-corner} to other loss functions.
\end{remark}

\subsection{Point Identification Under Missing at Random Assumptions}\label{subsec:identification} 
While the distribution $\bP$ of $(\Y,\G,\X,\Z)$ is point identified by the observed data, absent additional assumptions, the distribution $\bP^*$ of $(\Y^*,\G,\X,\Z)$ is not.
The previous section shows that selectively observed labels can also prevent point identification of $\h_\E(q)$ and $\e(a)$ for $a\in\mathcal A(\sX)$, which depend on the ideal labels $\bLL_d$ and their conditional expectations $\btheta_d(\X)$.
To obtain point identification, we impose selection on observables and overlap restrictions, thereby disciplining the mechanism generating label observability.
\begin{asm}[Ideal labels missing at random and overlap]
\label{asm:MAR-X}
\textit{$(\Y^*,\G)\perp\Z\mid\X$ and
$\bpi(\X)\equiv\bE[\Z\mid\X]\in(0,1)$ a.s.}
\end{asm}

The conditional independence component of Assumption~\ref{asm:MAR-X} is equivalent to requiring $\Y^*\perp\Z\mid\X,\G$ together with $\G\perp\Z\mid\X$.
Our approach also accommodates the weaker condition $\Y^*\perp\Z\mid\X,\G$, provided overlap is imposed conditional on $(\X,\G)$, so that $\bP(\Z=1\mid\X,\G)\in(0,1)$ a.s.
Here we work with Assumption~\ref{asm:MAR-X} because it yields lighter notation and simpler expressions.

A sufficient mechanism for Assumption~\ref{asm:MAR-X} arises
when label observability is generated by a randomized decision rule as follows.
Suppose an individual receives the binary decision $\Z=1\{U\le a_0(\X)\}$, with $a_0\in\mathcal A(\sX)$ and $U\sim\operatorname{Unif}[0,1]$ independent of $(\Y^*,\G,\X)$.
Then $\bP^*(\Z=1\mid\X,\G,\Y^*)=\bP(\Z=1\mid\X)=a_0(\X)$, so the conditional independence assumption holds.
Overlap is a separate restriction on the data generating mechanism, requiring $a_0(\X)\in(0,1)$ a.s.
Beyond assignments implemented through independent randomization, the condition $\Y^*\perp\Z\mid\X,\G$ is restrictive.
It is plausible when recorded variables account for the information linking selection to ideal labels, and any remaining variation in assignment is unrelated to those outcomes conditional on $(\X,\G)$.
It may fail when the assignment process exploits unrecorded
``soft information'' (e.g., interviews, free text, documents, or unlogged history) that remains informative about ideal labels after conditioning on $(\X,\G)$.
Such departures from selection on observables motivate our partial identification analysis in Section~\ref{subsec:PI-CR-SP}, subject to that section's maintained assumptions.

The next lemma shows that, under Assumption~\ref{asm:MAR-X} the support function of $\E$ is point identified through a standard inverse propensity weighting identity:
\begin{lemma}\label{lemma:support-funct-identif}
\textit{Let  Assumptions~\ref{asm:moments} and \ref{asm:MAR-X} hold. Then,
\begin{align}
    \btheta_d(\X) &= \bE \left[ \left. \tfrac{\Z\bLL_d }{\bpi(\X)}  \right| \X \right]\label{eq:thetaMAR},  \\
    \h_\E(\q) &=  \bE \left[    \tfrac{\q^\intercal \bLL_0\Z }{\bpi(\X)}  + \left ( \tfrac{\q^\intercal\Delta\bLL \Z}{\bpi(\X)}\right)\one\{\q^\intercal \Delta\btheta(\X)>0\}  \right],\quad\text{for all }\q\in\bS^2,\label{eq:hMAR}\\
    e_g^\iota(a) & = \bE\left[   \tfrac{a(\X)\LL_1^{g,\iota} \Z}{\bpi(\X)} + \tfrac{(1-a(\X))\LL_0^{g,\iota} \Z}{\bpi(\X)}  \right],\quad g=r,b~\text{and}~\iota=A,F.\label{eq:e_gMAR}
\end{align} }
\end{lemma}

\section{A Debiased Machine Learning Estimator and Its Asymptotic Distribution}\label{sec:DML}
\subsection{The DML Estimator}\label{subsec:DMLestimator}
Our goal is to estimate the support function $\h_\E(\q)$ using a random sample $\{(\Y_i,\G_i,\X_i,\Z_i): 1 \le i \le n \}$ drawn from the distribution $\bP$. By Lemma \ref{lemma:support-funct-identif}, $\h_\E(\q)$ is identified by \eqref{eq:hMAR}, which involves nuisance functions $\bpi(\X)$ and $\Delta \btheta(\X)$ that need to be estimated in the first step.
As the ideal labels $\bLL$ depend on $\bmu \equiv [\mu_r,\mu_b]^\intercal$ (see \eqref{eq:defLdgA}-\eqref{eq:defLvector}), we also need to estimate $\bmu$. 
In order to build an estimator for $\h_\E(\q)$ that is locally robust to the first-step estimation error, we propose using the debiased machine learning (DML) approach.
Consider the estimand 
\begin{equation}
    \h_\E(\q) = \bE\left[ \zeta_i(\q;\bLL,\boeta) \right],\label{eq:new-estimand-support-funct}
\end{equation}
where $\boeta\equiv(\Delta \btheta,\bpi,\btheta_0)$ is a vector of nuisance functions and
\begin{align}
    \zeta_i(\q;\bLL,\boeta) &\equiv      \tfrac{\q^\intercal\bLL_{0,i}\Z_i}{\bpi(\X_i)}  + \left(\tfrac{\q^\intercal \Delta \bLL_{i} \Z_i}{\bpi(\X_i)}\right)  \one\{\q^\intercal \Delta\btheta(\X_i)>0\} + \balpha^\h(\q;\X_i) \left(1 - \tfrac{\Z_i}{\bpi(\X_i)}\right),\label{eq:zeta}\\
    \balpha^\h(\q;\X) &\equiv \q^\intercal \btheta_0(\X) + \q^\intercal \Delta \btheta (\X) \one\{\q^\intercal \Delta\btheta(\X)>0\}.\label{eq:alpha}
\end{align}
The local effect of the first-step function $\bpi$ on the original estimand $\h_\E(q)$ in \eqref{eq:new-estimand-support-funct} is accounted for by 
$ \balpha^\h(\q;\X_i) \left(1 - \tfrac{\Z_i}{\bpi(\X_i)}\right)$.
As we show in the proof of Theorem~\ref{thm:gaussian}, under Assumptions \ref{asm:nuisance-function-structure} and \ref{asm:nuisance-funct-estimators} presented below, it is not necessary to account for the local effect of the first-step function $\Delta \btheta$ on $\h_\E(q)$. 
Finally, we note that the estimand in \eqref{eq:new-estimand-support-funct} is based on $\bLL$, which by \eqref{eq:defLdgA} relies on the population means $\mu_g, g\in\{r,b\}$.

Following the DML approach \citep[e.g.,][]{ChernozhukovDML,velez2024asymptotic}, we first randomly split the indices $[n]\equiv\{1,\ldots,n\}$ into $K$ equal-sized folds $\mathcal{I}_k$, i.e., $\cup_{k=1}^K \mathcal{I}_k = [n]$. We denote by $n_k$ the size of $\mathcal{I}_k$.\footnote{When $n$ is not divisible by $K$, the number of observations in some folds will be $\lfloor n/K \rfloor$ while in others $\lfloor n/K \rfloor + 1$, where $\lfloor n/K  \rfloor$ is the greatest integer less than or equal to $n/K$.}
We then estimate $\widehat{\boeta}(X_i)$ as $\widehat{\boeta}_k(X_i)$ for $i \in \mathcal{I}_k$ and $k=1,\ldots,K$, where $\widehat{\boeta}_k(\cdot)$ is estimated using all data except the portion with indices in fold $\mathcal{I}_k$. 
Finally, the first-stage estimators for the nuisance functions are used to form a second-stage estimator for $\h_\E(\q)$ based on the expression of the estimand defined in \eqref{eq:new-estimand-support-funct}: 
\begin{equation}\label{eq:DML-estimator}
    \widehat{\h}_\E(q;\widehat\bLL,\widehat{\boeta}) 
    \equiv \tfrac{1}{K} \sum_{k=1}^K \left( \tfrac{1}{n_k} \sum_{i \in \mathcal{I}_k}  \zeta_i(\q;\widehat\bLL,\widehat{\boeta}_k)\right)
    = \tfrac{1}{n} \sum_{i=1}^n \zeta_i(\q;\widehat\bLL,\widehat{\boeta}),
\end{equation}
where $\widehat\bLL$ is defined as in \eqref{eq:defLdgA}-\eqref{eq:defLvector} with $\mu_g$ replaced by $\widehat{\mu}_g \equiv n^{-1} \sum_{i=1}^n \one\{G_i=g\}$.

We also propose a DML estimator for $\e^* = \e(a^*) \in \E$. By \eqref{eq:e_gMAR} in Lemma \ref{lemma:support-funct-identif} and Definition 3, we can write $\e^*= \bE[ \xi_i(a^*;\bLL,\boeta)]$, 
where for a given algorithm $a\in\mathcal A$,
\begin{align}
    \xi_i(a;\bLL,\boeta) &\equiv  \left( {\bLL}_{0,i} + a(\X_i) \Delta\bLL_i \right) \tfrac{\Z_i}{{\bpi}(\X_i)} + \boldsymbol{\balpha}^e(a;\X_i)  \left(1 - \tfrac{\Z_i}{{\bpi}(\X_i)} \right),\label{eq:xi}\\
      \boldsymbol{\balpha}^e(a;\X) &\equiv  {\btheta}_{0}(\X) + a(\X) \Delta {\btheta}(\X),\label{eq:alpha-e}
\end{align}
and again the local effect of the first-step function $\bpi$ on the original estimand $\e^*$ is accounted for by $\boldsymbol{\balpha}^e(\X_i)  \big(1 - \tfrac{\Z_i}{{\bpi}(\X_i)}\big)$.
 Using this notation, we propose to estimate $\e^*$ by
 \begin{equation}\label{eq:DML-e*}
     \widehat{\e}^* \equiv \tfrac{1}{K} \sum_{k=1}^K \left( \tfrac{1}{n_k} \sum_{i \in \mathcal{I}_k}  \xi_i(a^*;\widehat{\bLL},\widehat{\boeta}_k)   \right) = \tfrac{1}{n} \sum_{i=1}^n \xi_i(a^*;\widehat{\bLL},\widehat{\boeta})   ~,
 \end{equation}
 where $\widehat{\bLL}$ stacks the feasible labels and $\{\widehat{\boeta}_k\}_{k=1}^K$ are estimators for the nuisance functions using cross-fitting, as we described above.

\subsection{Asymptotic Distribution of the DML Estimator}\label{subsec:asy_distDML}
For the nuisance functions' estimation error to be asymptotically negligible for the DML estimator, the nuisance estimators must satisfy suitable regularity conditions set out next.
\begin{asm}\label{asm:nuisance-function-structure}
\textit{    There is a known partition of  $\X$,
    $ \X = (\X_1, \X_2),$
    with $\X_1 \in \sX_1 \subset \R^{p}$ and $\X_2 \in \sX_2 \subset \R^{d_X-p}$ such that:}
    \begin{itemize} 
        \item[(a)] \textit{
        For every $\delta>0$,
        $ \sup_{\q\in\bS^2}\bP\left( |\q^\intercal \Delta \btheta(\X)| \le \delta | X_2 \right) \lesssim \delta. $}

        \item[(b)] \textit{There is a constant $c>0$ such that 
        $\inf_{ \q \in \bS^2} Var[ \,|\q^\intercal \Delta \btheta(\X)| \,| \X_2] \ge c~.$}
         
        \item[(c)] \textit{There are unknown functions $\nu:\sX_1\to\R^{d_\nu}$ and $\gamma:\sX_2\to\R^{d_\gamma}$, and a known function $F: \R^{d_\nu} \times \R^{d_\gamma} \times \R^2 \to [\epsilon,1-\epsilon] \times \R^3$, such that:
        \begin{align*}
            (\bpi(\X),~\Delta \btheta(\X)) &= F(\nu(X_1), \gamma(X_2),\mu)
        \end{align*}
        and the partial derivatives of $F$ are uniformly bounded for any $\nu \in \R^{d_\nu}$, $\gamma \in \R^{d_\gamma}$, and $\mu \in \bB_\epsilon(\mu_r,\mu_b) \equiv\{ \mu \in \R^2: \|\mu - (\mu_r,\mu_b)\|<\epsilon\}$, that is
        $$\sup_{\nu \in \R^{d_\nu}} \sup_{ \gamma\in \R^{d_\gamma}} \sup_{\mu \in \bB_\epsilon(\mu_r,\mu_b) }\| D F(\nu, \gamma,\mu)\| < \infty~.$$}
    \end{itemize}
\end{asm}
A sufficient condition for Assumption \ref{asm:nuisance-function-structure} is given in \citet[Remark 4.1]{liu:mol25v2}, where a partially linear structure is assumed for $\Delta\btheta$.

Let $\mathfrak{F}_k \equiv \sigma( (Y_i,G_i,X_i,Z_i): i \notin \mathcal{I}_k)$ be the $\sigma$-algebra generated by the data used for constructing the estimator $\widehat{\boeta}_k$ for the nuisance ${\boeta} \equiv (\Delta \btheta,\bpi,\btheta_0)$. Let $\widehat{\bmu}_k$ be an estimator for $\bmu \equiv [\mu_r,\mu_b]^\intercal$ that uses only the data outside fold $\mathcal{I}_k$ (similar to how we define $\widehat{\boeta}_k$):
$$\widehat{\bmu}_k = \left(\tfrac{1}{n-n_k} \sum_{i \notin \mathcal{I}_k} \one\{G_i=r\},~\tfrac{1}{n-n_k} \sum_{i \notin \mathcal{I}_k} \one\{G_i=b\} \right)^\intercal~.$$

We assume that the estimators for $\Delta \btheta(\X)$ and $\bpi(\X)$ share the same structure as in Assumption~\ref{asm:nuisance-function-structure}. Furthermore, we treat $\nu(\X_1)$ as a low-dimensional nonparametric component that can be estimated uniformly well over $\sX_1$, whereas $\gamma(X_2)$ captures a high-dimensional component that can be estimated sufficiently well in the sense that its mean squared error converges to zero sufficiently fast. We assume that $\btheta_0$ can be estimated sufficiently well.

\begin{asm}\label{asm:nuisance-funct-estimators}
\textit{    There exist estimators $\widehat{\nu}_k$, $\widehat{\gamma}_k$, and $\widehat{\btheta}_{0,k}$ for $k=1,\ldots,K$ such that
    \begin{enumerate}
        \item[(a)] $(\widehat{\bpi}_k(\X),~\Delta \widehat{\btheta}_k(\X)) = F(\widehat{\nu}_k(X_1), \widehat{\gamma}_k(X_2),\widehat{\bmu}_k)$.
        \item[(b)] $\sup_{x_1 \in \sX_1} \|\widehat{\nu}_k(x_1) - \nu(x_1)\| = o_p(n^{-1/4})$.
        \item[(c)] $\bE \left[ \|\widehat{\gamma}_k(X_2)-{\gamma}(X_2)\|^2  | \mathfrak{F}_k \right] = o_p(n^{-1/2})$ .
        \item[(d)]  $\bE \left[ \|\widehat{\btheta}_{0,k}(\X)-{\btheta}_0(\X)\|^2  | \mathfrak{F}_k \right] = o_p(n^{-1/2})$.
    \end{enumerate}
}    
\end{asm}
Example machine learning methods that, under suitable conditions, are $n^{1/4}$-consistent in the $L^2$-norm include $\ell_1$-penalized methods, boosting, 
and neural nets \citep[see, e.g.,][and references therein, for a discussion of machine learners compatible with the rate requirement]{ChernozhukovDML}. 

Similarly to \citet[Thm. 4.1]{liu:mol25v2}, our main asymptotic result shows that $\widehat{\h}_{\E}(\q;\widehat{\bLL}, \widehat{\boeta})$ converges to a Gaussian process uniformly in $\q\in\bS^2$, where the score $\zeta_i(\q;\bLL,\boeta)$ in \eqref{eq:zeta} is the influence function that governs the part of the limit distribution of $\widehat{\h}_{\E}(\q; \widehat{\bLL},\widehat{\boeta})$ due to the uncertainty in $\widehat{\boeta}$, and the remaining part is attributed to estimating $\bmu$:
\begin{theorem}\label{thm:gaussian}
    \textit{Let Assumptions~\ref{asm:moments}-\ref{asm:no-kink} and \ref{asm:MAR-X}-\ref{asm:nuisance-funct-estimators} hold. Then, 
    $$ \sqrt{n}\left( \widehat{\h}_\E(\q; \widehat{\bLL}, \widehat{\boeta}) - \h_\E(q) \right) \Rightarrow \mathbb{G}[\zeta_i^{*}(q;\bLL,\boeta)] \quad \text{in} \quad \ell^{\infty}(\bS^2)~,$$
     where, for $\Gamma_g^h(\q)\equiv \bE\left[ \one\{G_i=g\} \left(\q^\intercal \bLL_{0,i}+ \q^\intercal\Delta\bLL_i\one\{\q^\intercal \Delta\btheta(\X)>0\} \right) \tfrac{Z_i}{\pi(X_i)} \right]$,
\begin{align}\label{eq:zeta-star}
    \zeta_i^*(q;\bLL,{\boeta}) \equiv \zeta_i(q;\bLL,{\boeta})  + \sum_{g \in \{r,b\}}  \Gamma_g^h(\q) \left(1-\tfrac{\one \{ G_i = g\}}{\mu_g}\right).
\end{align} }
\end{theorem}
\begin{remark}
    One could alternatively construct a DML estimator of $\h_\E(q)$ using $\zeta_i^{*}(q;\bLL,\boeta)$ in \eqref{eq:zeta-star} instead of $\zeta_i(q;\bLL,\boeta)$ in \eqref{eq:zeta}.
    Doing so would be valid because $\zeta_i^{*}(q;\bLL,\boeta)$ satisfies the Neyman orthogonality condition with respect to $\boeta$ and $\bmu$. In practice, when $\widehat{\boeta}$ is obtained via cross-fitting and $\bmu$ is estimated by sample means, $\widehat{\h}_\E(\q; \widehat{\bLL},  \widehat{\boeta})$ is numerically equivalent to $\frac{1}{n}\sum_{i=1}^n \zeta_i^{*}(q; \widehat{\bLL},\widehat{\boeta})$.
    Hence, our proposed estimator implicitly accounts for first-step estimation error in $\widehat\bmu$.
\end{remark}

\section{Testing Procedures}\label{sec:tests}
We provide test statistics and their asymptotic distribution to test whether (i) the FA-frontier coincides with the Pareto frontier; and (ii) a less discriminatory alternative to a given algorithm exists.

\subsection{Testing $\sPF=\sF$}\label{sec:testFA_Pareto}
Theorem~\ref{thm:Pareto_equal_FA} provides a necessary and sufficient condition for $\sPF=\sF$.
Here we put forward an equivalent characterization based on the support function process, which is amenable to testing the null hypothesis $\sPF=\sF$ against the alternative $\sPF\neq\sF$.
\begin{prop}\label{prop:testable_F=PF}
\textit{Let Assumptions~\ref{asm:moments}-\ref{asm:margin} hold. For $C$ a constant pinned down by Assumption~\ref{asm:moments}, $\E\subset \bB_C$, with $\bB_C$ the ball in $\R^3$ of radius $C$. Define
\begin{align*}
T^{\texttt{LDA}}(\e)&\equiv\left[\sup_{\q\in\mathbb S^2}\{\q^\intercal\e-\h_\E(\q)\}\right]_+ +\left[\inf_{\q\in\widetilde{S}^2}\{h_{\mathcal C(\e)}(\q)+h_\E(-\q)\}\right]_+,\\
\psi_\sPF(\e)&\equiv\left[\inf_{\alpha\in[0,1]}\{\q(\alpha)^\intercal\e+\h_\E(-\q(\alpha))\}\right]_+,
\qquad
\q(\alpha)\equiv\tfrac{[\alpha,1-\alpha,0]^\intercal}{\Vert[\alpha,1-\alpha,0]^\intercal\Vert}.
\end{align*}
Then $\sPF=\sF \iff \sup_{\e\in \bB_C:T^{\texttt{LDA}}(\e)=0}\psi_\sPF(\e)=0$.}

\textit{If Assumption~\ref{asm:no-kink} also holds, fix any $\bar s>0$ and define 
\begin{align}
\Delta_{\sPF}^{h}(\bar s)
\equiv
\sup_{\alpha\in[0,1]}
\left[
h_\E(-q(\alpha))
-
\inf_{s\in[-\bar s,\bar s]}
\sqrt{1+s^2}\,
h_\E\left(
\tfrac{-q(\alpha)+s u_3}{\sqrt{1+s^2}}
\right)
\right],\label{eq:Delta_sPF_h}
\end{align}
where $u_3=[0,0,1]^\intercal$. Then $\sPF=\sF\iff \Delta_{\sPF}^{h}(\bar s)=0$.
}
\end{prop}

Eq. \eqref{eq:Delta_sPF_h} translates the geometric condition
$\sPF\subset\{\e_3=0\}$ in Theorem \ref{thm:Pareto_equal_FA}
into a criterion involving only the support function. For each
Pareto-supporting direction $-\q(\alpha)$,  the expression inside
the supremum checks 
whether $-\q(\alpha)$ already minimizes the support value over
fairness-direction tilts
$\{-\q(\alpha)+s u_3:s\in[-\bar s,\bar s]\}$. In the proof of Proposition \ref{prop:testable_F=PF}, we show that the
derivative of this tilted support function at $s=0$ equals the
fairness coordinate of the Pareto point; by convexity,
$s=0$ minimizes the support value over $[-\bar s,\bar s]$ if and only if
this derivative---and thus the fairness coordinate---is zero. Hence, the
criterion in \eqref{eq:Delta_sPF_h} is zero exactly when \emph{all} 
Pareto points have zero fairness disparity, i.e.
$\sPF\subset\{\e_3=0\}$, which is equivalent to $\sPF=\sF$ by Theorem
\ref{thm:Pareto_equal_FA}.

\begin{remark}
    In work in progress, we derive the limit distribution of sample analogs of $\psi_{\sPF}(\e)$ and $\Delta_{\sPF}^{h}$, and valid bootstrap critical values.
\end{remark}

\subsection{Assessing the Existence of a Less Discriminatory Alternative}\label{sec:test_statLDA}

When evaluating whether there exists a less discriminatory alternative (LDA) to a given algorithm, regulators and policymakers need to confront the fact that they can only rely on finite data to do so. 
Concretely, given an algorithm $a^*\in\mathcal{A}(\sX)$, we call another algorithm $a \in\mathcal{A}$ an LDA if it yields expected losses that are at least as accurate as those associated with $a^*$ for both groups, and at least as fair, with one of these inequalities strict. 
In terms of the FA-dominance notion and FA-frontier reported in Definition \ref{def:R3_frontier}, the absence of an LDA is equivalent to the fact that $a^*$ yields expected losses with
\begin{align}
    \e^*\equiv\e(a^*)=(e^A_r(a^*),e^A_b(a^*),e^F_r(a^*)-e^F_b(a^*))\in\sF.\label{eq:null_LDA_e_in_sF}
\end{align} 
By Theorem~\ref{thm:FequivLLMO}, \eqref{eq:null_LDA_e_in_sF} and \eqref{eqn:null-LDA} below are equivalent, hence we can use the latter to test the hypothesis that no LDA exists against the alternative that it does, as follows:
\begin{align}
    \HH_0&: \min_{q \in \tilde\bS^2} \left[\h_{\cC(\e^*)}(\q) + \h_{\E}(-\q)\right] = 0,\label{eqn:null-LDA}\\
     \HH_A&:  \min_{q \in \tilde\bS^2} \left[\h_{\cC(\e^*)}(\q) + \h_{\E}(-\q)\right] > 0,\label{eqn:H1-LDA}
\end{align} 
where $\h_{\cC(e^*)}(\q)$ is defined in \eqref{eq:supportC}. 
We refer to \eqref{eqn:null-LDA} as the LDA hypothesis.

We propose to reject the LDA hypothesis for large values of  the following test statistic:
\begin{align}
T_n^{\texttt{LDA}} \equiv \sqrt{n}\left( \left[\max_{\q\in\bS^2}(\q^\intercal \widehat{\e}^{\omega}-\hn_\E^\omega(q; \widehat{\bLL}^\omega,\widehat{\boeta}))\right]_{+} +\left[\min_{q\in\tilde{\bS}^2}\left(\h_{\cC(\widehat{\e}^\omega)}(\q) + \hn_\E^\omega(-q; \widehat{\bLL}^\omega,\widehat{\boeta})\right)\right]_{+}\right),\label{eq:Tn_LDA}
\end{align} 
where $\hn_\E^\omega$ is a weighted estimator for $\h_\E$ defined in \eqref{eq:DML-estimator-weighted} in Appendix \ref{app:DML} and $\widehat{\e}^{\omega}$ is a weighted estimator for $\e^*$ defined in \eqref{eq:DML-estimator-weighted-e} in Appendix \ref{sec:weighted-estimator-e}.
The former uses weights $\omega_i^h = 1 - \Xi_i/2$ and the latter uses weights $\omega_i^\e = 1 + \Xi_i/2$, $i=1,\dots,n$, with $\Xi_i$ a Rademacher random variable taking values in $\{-1,1\}$ with uniform probability and independent of the sample data and of the first stage training split. The test statistic in \eqref{eq:Tn_LDA} was proposed in \citet[Section 6.2]{liu:mol25v2}. The reason why it resorts to using reweighted estimators for $\h_\E(\q)$ and $\e^*$ is that under the null hypothesis that $\e(a^*)\in\sF$, Corollary~\ref{corr:alg_as} yields that there exists a direction vector $\q^*\in\tilde{\bS}$ such that $a^*(\X)=\mathds{1}\{\q^{*\intercal}\Delta\btheta(\X)>0\},\,\bP^*_\X$-a.s. 
As a consequence, one can show that under Assumption~\ref{asm:no-kink}, for $\e^*$ a point at which $\h_{\cC(\e^*)}(\q^*)$ is locally differentiable, the limit distribution of a test statistic like the one in \eqref{eq:Tn_LDA} but with $\omega_i^h=\omega_i^e=1$ for all $i=1,\dots,n$, has a degenerate limit distribution.
\citet{liu:mol25v2} show that this test statistic with degenerate limit distribution asymptotically rejects a true null with probability zero, and a false null under \emph{local} ($1/\sqrt n$) alternatives motivated by practical significance criteria with probability one \citep[for a discussion of ``practical significance'' see][Sec. 6]{bla:spi22}.
The use of Rademacher weights regularizes the limit distribution of the test statistic, trading asymptotic size $\alpha$ for the null, and local asymptotic unbiasedness at local alternatives closer to the null than $1/\sqrt n$ times the practical significance threshold, for power above nominal size instead of one just above that threshold.

Our contribution is deriving the test's asymptotic distribution as well as a valid bootstrap procedure to estimate the associated critical values, in the presence of selectively observed labels.
We consider a bootstrap-based critical value $\widehat{c}_{1-\alpha}$ defined in Procedure~\ref{proc:bootstrap-weighted} below. To simplify notation, let $\bbeta \equiv ( {\h}_\E(\q),{\e}^*) \in \ell^{\infty}(\bS^2) \times \R^3$ denote the vector of population values and let $\widehat{\bbeta}^\omega  \equiv (  \widehat{\h}_\E^\omega(\q),\widehat{\e}^\omega) \in \ell^{\infty}(\bS^2) \times \R^3$ denote the vector of their estimators. We can write $T_n^{\texttt{LDA}} = \sqrt{n} \left( \phi(\widehat{\bbeta}^\omega) - \phi(\bbeta) \right)$, where $\phi$ is defined in \eqref{eq:t-stat} in Appendix \ref{appn:A}.  

\begin{proc}[Bayesian Bootstrap for the Quantiles of $T_n^{\texttt{LDA}}$]
\label{proc:bootstrap-weighted}

    \begin{enumerate}
    \textit{   \item Draw $\{W_i\}_{i=1}^n$ i.i.d. from the exponential distribution with mean $1$ independent of the sample $\{(\Y_i,\G_i,\X_i,\Z_i)\}_{i=1}^n$ and the weights $\{\Xi_i\}_{i=1}^n$. Let the bootstrap analogue of $\widehat{\bbeta}^\omega$ be $\widetilde{\bbeta}^\omega \equiv (  \widetilde{\h}_\E^\omega,\widetilde{\e}^\omega)$, where $ \widetilde{\h}^\omega_\E(q;\widetilde{\bLL}^\h,\widehat{\boeta})  \equiv \tfrac{1}{n}    \sum_{i =1}^n  \left( \tfrac{W_i \omega_i^\h}{\overline{W^\h}} \right) \zeta_i(q;\widetilde{\bLL}^\h,\widehat{\boeta}) $, $ \widetilde{\e}^\omega  \equiv \tfrac{1}{n}    \sum_{i =1}^n      \left( \tfrac{W_i \omega_i^\e}{\overline{W^\e}} \right)  \xi_i(a;\widetilde{\bLL}^\e,\widehat{\boeta})$,  $\widetilde\bLL^\h$ and $\widetilde\bLL^\e$ are defined as in \eqref{eq:defLdgA}-\eqref{eq:defLvector} with $\mu_g$ replaced by $\widetilde{\mu}_g^\h = n^{-1} \sum_{i=1}^n \left( \tfrac{W_i \omega_i^\h}{\overline{W^h}} \right)  \one\{G_i=g\}$ and $\widetilde{\mu}_g^\e = n^{-1} \sum_{i=1}^n \left( \tfrac{W_i \omega^\e}{\overline{W^\e}} \right)  \one\{G_i=g\}$, respectively, and $\overline{W^\h} = n^{-1} \sum_{i=1}^n W_i \omega_i^\h$ and $\overline{W^\e} = n^{-1} \sum_{i=1}^n W_i \omega_i^\e$.
}

\noindent\textit{ 
    \item Numerically approximate $\phi'_{\bbeta}(\cdot)$, the directional derivative of $\phi(\cdot)$ at $\bbeta$, by
    \begin{align*}
        \widehat{\phi}'_{\bbeta}(	\Ddot{\bbeta})=\tfrac{1}{s_n}\left(\phi\left(\widehat{\bbeta}^\omega+s_n(\Ddot{\bbeta})\right)-\phi\left(\widehat{\bbeta}^\omega\right)\right),
    \end{align*}
    where $\Ddot{\bbeta}\in\ell^\infty(\bS^2)\times\mathbb{R}^3$ is a candidate direction at which we evaluate $\phi'_{\bbeta}(\cdot)$ and $s_n$ is a vanishing sequence of step sizes such that $\sqrt{n}s_n\to\infty$.
    \item Obtain $\widehat{\phi}'_{\bbeta}\big(\sqrt{n}\{\widetilde{\bbeta}^\omega-\widehat{\bbeta}^\omega\}\big)$ and calculate 
    \begin{align*}  \widehat{c}_{1-\alpha}\equiv\inf\bigg\{c:\bP\bigg(\left.\widehat{\phi}'_{\bbeta}\big(\sqrt{n}\{\widetilde{\bbeta}^\omega-\widehat{\bbeta}^\omega\}\big)\leq c \,\,\right|\, \{(\Y_i,\G_i,\X_i,\Z_i,\Xi_i)\}_{i=1}^n\bigg)\geq 1-\alpha\bigg\}~.
\end{align*}
}
    \end{enumerate}
\end{proc}

Let us define the LDA test by $\varphi_n^{\texttt{LDA}} = \one \{ T_n^{\texttt{LDA}} > \widehat{c}_{1-\alpha + \kappa}+\kappa\}$   for a given significance level $\alpha \in (0,1)$, where $\kappa>0$ is an arbitrarily small positive constant. The constant $\kappa$ can be taken equal to $10^{-6}$ as in \cite{and:shi13}. We can take $\kappa = 0$ when the asymptotic distribution of $T_n^{\texttt{LDA}}$ is continuous at its $(1-\alpha)$-quantile; $\kappa>0$ appears in the critical value used in the LDA test to handle the technicalities arising from possible discontinuity points of the asymptotic distribution of $T_n^{\texttt{LDA}}$.

The next result establishes the consistency of the Bayesian bootstrap outlined in Procedure \ref{proc:bootstrap-weighted} and guarantees the asymptotic correct size of the proposed test $\varphi_n^{\texttt{LDA}}$. 

\begin{theorem}\label{thm:test-bootstrap-weights} 
    \textit{Let Assumptions \ref{asm:moments}--\ref{asm:nuisance-funct-estimators} hold. Then, the Bayesian bootstrap outlined in Procedure \ref{proc:bootstrap-weighted} is consistent. Furthermore, if the LDA hypothesis holds, the test $\varphi_n^{\texttt{LDA}}$ has asymptotically correct size,
    $\limsup_{n \to \infty}\bE[\varphi_n^{\texttt{LDA}}] \le \alpha~.$}
\end{theorem}

Given Theorem~\ref{thm:test-bootstrap-weights}, we can build an asymptotically valid confidence set for $\sF$ by test inversion, similarly to \citet[Proposition 5.1]{liu:mol25v2}.

\section{Conclusions}\label{sec:conclude}
In this paper, we consider the problem of identification and inference for the FA-frontier put forward in \citetalias{lia:lu:mu:oku24}, when labels (outcomes) are selectively observed.
We build on the methodology developed by \citet{liu:mol25v2}, who assume away the selective labels problem to focus on deriving moment-inequality representations for the frontier.
Selective labels challenge one's ability to (point) identify features of the conditional distribution of the true labels, including the FA-frontier. Naively computed accuracy metrics using only observed outcomes may misrepresent true performance, potentially overstating accuracy where it matters most and distorting conclusions on whether LDAs exist.

When the selection process is left completely unrestricted, we provide (for specific loss functions) a tractable criterion function to assess whether a status quo algorithm is on the frontier for some admissible distribution of the missing labels.
When this is not the case, we characterize an alternative algorithm that strictly improves accuracy and weakly improves fairness, uniformly over admissible distributions of the missing labels. We also characterize the sharp identification region for the FA-frontier.
In work in progress, we derive a DML estimator and its asymptotic theory (along with valid bootstrap methods) to test the hypothesis that, for some distribution of the missing labels, no LDA exists to a given algorithm.

Under the assumption that labels are missing at random conditional on a rich set of observed covariates, we obtain point identification of the FA-frontier using inverse propensity score weighting and derive the asymptotic distribution of the test statistic that we use to characterize whether an algorithm is on the frontier.
The inference method that we propose is based on a debiased machine learning estimator, which is particularly appropriate in this context, where high-dimensional administrative data enable flexible nonparametric modeling of the selection process and hence make the missing at random assumption more credible.
In this context, formal methods that control the bias of the first-step nonparametric estimation and yield valid confidence statements are especially valuable for making legally defensible determinations.

\bibliographystyle{ecta-fullname} 
\bibliography{fairness_frontier}  

\setcounter{defn}{0}
\counterwithin{defn}{section}
\renewcommand{\thedefn}{\thesection.\arabic{defn}}
\begin{appendix}
\numberwithin{equation}{section}

\section{Proofs of the Main Results}\label{appn:A}
Assume throughout that $(\sX,\mathcal B(\sX))$ is a standard Borel space.
\subsection{Proofs for Section~\ref{sec:setup}}
\begin{proof}[Proof of Theorem~\ref{thm:support_E}]
Recall $\E\equiv\left\{e(a)\equiv\big(e_r^A(a),\,e_b^A(a),\,e_r^F(a)-e_b^F(a)\big): a\in\mathcal A\right\}$ and $\h_{\E}(\q)=\sup_{a\in\mathcal A}\q^\intercal e(a)$.
Using the representation $e(a)=\bEs[\btheta_0(\X)]+\bEs\left[a(\X)\Delta\btheta(\X)\right]$, we obtain $\q^\intercal e(a)=\bEs\left[\q^\intercal \btheta_0(\X)\right]+\bEs\left[a(\X)\q^\intercal \Delta\btheta(\X)\right]$.
Hence,
\begin{align*}
\h_{\E}(\q)=
\bEs\left[\q^\intercal \btheta_0(\X)\right]
+
\sup_{a\in\mathcal A}\bEs\left[a(\X)\q^\intercal \Delta\btheta(\X)\right].
\end{align*}
Because $a(\X)\in[0,1]$ and $\q^\intercal \Delta\btheta(\X)$ is $\sigma(\X)$-measurable, the pointwise maximizer is
\begin{align*}
a^\q(\X)\in \arg\max_{t\in[0,1]} t\,\q^\intercal \Delta\btheta(\X)
=
\begin{cases}
1, & \q^\intercal \Delta\btheta(\X)>0,\\
0, & \q^\intercal \Delta\btheta(\X)<0,
\end{cases}
\end{align*}
with arbitrary tie-breaking on $\{\q^\intercal \Delta\btheta(\X)=0\}$. Therefore,
\begin{align*}
\sup_{a\in\mathcal A}\bEs\left[a(\X)\q^\intercal \Delta\btheta(\X)\right]=
\bEs\left[\q^\intercal \Delta\btheta(\X)\one\{\q^\intercal \Delta\btheta(\X)>0\}\right].
\end{align*}
Next, note that $\one\{\q^\intercal \Delta\btheta(\X)>0\}$ is $\sigma(\X)$-measurable. Therefore, by iterated expectations,
\begin{multline*}
\bEs\left[\q^\intercal \bLL_0+\q^\intercal\Delta\bLL\one\{\q^\intercal \Delta\btheta(\X)>0\}\right]
=\bEs\left[
\q^\intercal \btheta_0(\X)+\q^\intercal \Delta\btheta(\X)\one\{\q^\intercal \Delta\btheta(\X)>0\}
\right].
\end{multline*}
This proves \eqref{eq:support}. Next, \eqref{eq:ss} follows from a similar argument used in the proof of Proposition 3.2 of \citet{liu:mol25v2}, and we defer the details to Lemma \ref{lem:grad-supp-func} in the \hyperref[online_appendix]{Online Appendix}. Finally, by \citet[Corollary 1.7.3]{sch93}, $\E$ is strictly convex because by \eqref{eq:ss} its support set is a singleton in each direction $\q\in\bS^2$.
\end{proof}

\begin{proof}[Proof of Theorem \ref{thm:FequivLLMO}]
Recall $\widetilde{\bS}^2\equiv\{\q\in \bS^2:\q_1\ge 0,\ \q_2\ge 0\}$ and that under Assumptions~\ref{asm:moments}-\ref{asm:margin}, $\E$ has a nonempty interior, is compact, and is strictly convex.

\noindent
\textbf{Step 1: \ref{thm:FequivLLMO_e_in_sF}$\iff$\ref{thm:FequivLLMO_E_single_inters_C}.} By the definition of $\cC(\e^*)$ given in \eqref{eq:def_C}, we have $\E\cap \cC(\e^*)=\left\{\e\in \E:\e_1\le \e_1^*,\ \e_2\le \e_2^*,\ |\e_3|\le |\e_3^*|\right\}.$
Next, recall that $\e^*\in \sF$ if and only if there does not exist $\e\in \E$ with $\e_1\le \e_1^*$, $\e_2\le \e_2^*$, and $|\e_3|\le |\e_3^*|$, with at least one strict inequality. Since $\e^*\in \E\cap \cC(\e^*)$, this is equivalent to $\E\cap \cC(\e^*)=\{\e^*\}$.\smallskip

\noindent
\textbf{Step 2: \ref{thm:FequivLLMO_E_single_inters_C} $\imply$ \ref{thm:FequivLLMO_support}.}
Assume $\E\cap \cC(\e^*)=\{\e^*\}$.
As $\E$ is compact and convex with nonempty interior, and $\cC(\e^*)$ is closed and convex, the two sets admit a common supporting hyperplane at their unique common point. 
Hence, there exists $\q\neq 0$ and $\alpha\in\R$ such that
\begin{align*}
\q^\intercal \e\ge \alpha\ge \q^\intercal c
\qquad
\forall \e\in \E,\ \forall c\in \cC(\e^*).
\end{align*}
As $\e^*\in \E\cap \cC(\e^*)$, necessarily $\alpha=\q^\intercal \e^*$, hence $\inf_{\e\in \E}\q^\intercal \e = \sup_{c\in \cC(\e^*)}\q^\intercal c$.
Equivalently, $-\h_{\E}(-\q)=\h_{\cC(\e^*)}(\q)$.
Since $\h_{\cC(\e^*)}(\q)<+\infty$, we must have $\q_1\ge 0$ and $\q_2\ge 0$, because
$\cC(\e^*)$ is unbounded in the negative first- and second-coordinate directions. Since
$\q\neq 0$, positive homogeneity of support functions implies that, after replacing $\q$
by $\q/\|\q\|$, we may assume $\q\in \widetilde{\bS}^2$ while preserving the equality $\h_{\cC(\e^*)}(\q)+\h_{\E}(-\q)=0$.
Since $\e^*\in \E\inter \cC(\e^*)$, for every $\q\in \widetilde{\bS}^2$, $\h_{\cC(\e^*)}(\q)\ge \q^\intercal \e^*$ and $\h_{\E}(-\q)\ge -\q^\intercal \e^*$, and hence $\h_{\cC(\e^*)}(\q)+\h_{\E}(-\q)\ge 0$.
Therefore, $\min_{\q\in \widetilde{\bS}^2}
\left(\h_{\cC(\e^*)}(\q)+\h_{\E}(-\q)\right)=0$, proving \ref{thm:FequivLLMO_E_single_inters_C} $\imply$ \ref{thm:FequivLLMO_support}.\smallskip

\noindent
\textbf{Step 3: \ref{thm:FequivLLMO_support} $\imply$ \ref{thm:FequivLLMO_E_single_inters_C}.}
Assume $\min_{\q\in \widetilde{\bS}^2}\left(\h_{\cC(\e^*)}(\q)+\h_{\E}(-\q)\right)=0$.
As $\widetilde{\bS}^2$ is compact, $\E$ is compact, and $\h_{\cC(\e^*)}$ is finite and continuous on $\widetilde{\bS}^2$, the minimum is attained at some $\bar \q\in \widetilde{\bS}^2$. Thus, $\sup_{c\in \cC(\e^*)}\bar \q^\intercal c=\inf_{\e\in \E}\bar \q^\intercal \e$.
Let $\tilde \e\in \E\cap \cC(\e^*)$. Then
$\bar \q^\intercal \tilde \e
\le \sup_{c\in \cC(\e^*)}\bar \q^\intercal c=
\inf_{\e\in \E}\bar \q^\intercal \e
\le \bar \q^\intercal \tilde \e$.
Hence, every $\tilde \e\in \E\cap \cC(\e^*)$ satisfies $\bar \q^\intercal \tilde \e=\inf_{\e\in \E}\bar \q^\intercal \e$.
Therefore,
\begin{align*}
\E\cap \cC(\e^*)
\subset
\left\{
\e\in \E:\bar \q^\intercal \e=\inf_{\tilde \e\in \E}\bar \q^\intercal \tilde \e
\right\},
\end{align*}
that is, the intersection is contained in the support set of $\E$ in direction $-\bar \q$.
Since $\E$ is strictly convex, every support set of $\E$ is a singleton. Because $\e^*\in \E\cap \cC(\e^*)$, the above inclusion implies $\E\cap \cC(\e^*)=\{\e^*\}$, proving \ref{thm:FequivLLMO_support} $\imply$ \ref{thm:FequivLLMO_E_single_inters_C}.\smallskip

Combining Steps 1-3, we obtain \ref{thm:FequivLLMO_e_in_sF} $\iff$ \ref{thm:FequivLLMO_E_single_inters_C} $\iff$ \ref{thm:FequivLLMO_support}.
\end{proof}

\begin{proof}[Proof of Theorem~\ref{thm:Pareto_equal_FA}] We divide the proof of this theorem into four steps. The first two steps prove the first part of the theorem using Lemma \ref{lemma:PF-F-support-map}, while the last two steps additionally use the injectivity of the support map $\ss_\E$ implied by the no-kink condition.   

\noindent \textbf{Step 1:} Suppose that for any $\e^* \in \sF$ there exists $\q^* = \q(\alpha) \in \widetilde{\bS}^2$ (defined in the statement of the theorem) such that $\e^*= \ss_\E(-\q^*)$ and $\q^*_3 = 0$. By part 1 of Lemma \ref{lemma:PF-F-support-map}, this implies that $\e^* \in \sPF$; therefore, $\sF \subset \sPF$. Since $\sPF \subset \sF$ by Lemma \ref{lemma:PF-F-support-map}, we conclude $\sPF= \sF$.

\noindent \textbf{Step 2:} Suppose that $\sPF= \sF$. Since $\e^* \in \sF$ implies $\e^* \in \sPF$, we conclude that there exist $\q^* \in \widetilde{\bS}^2$ such that $\e^*= \ss_\E(-\q^*)$ and $\q^*_3 = 0$ by part 1 of Lemma \ref{lemma:PF-F-support-map}. We conclude the proof by setting $\q(\alpha) = \q^*$.

Recall that Assumption \ref{asm:no-kink} (no kinks) and the argument in \citet[Supplemental Appendix B.2.3]{bon:mag:mau12} imply the  injectivity of the support map $\ss_\E$.

    \noindent \textbf{Step 3:}  Suppose that $\sPF = \sF$.  Note that the injectivity of  $\ss_\E$  implies that any $\e^* \in \sF$ has a unique support direction $\q$ whose third coordinate is $0$ due to $\sPF = \sF$ and part 1 of Lemma \ref{lemma:PF-F-support-map}. Now, suppose that  $\sPF \subset \{ \e_3 = 0 \}$ is false, that is, there exist $\e^* \in \sPF$ such that $\e_3^* \neq 0$.  In what follows, we will find $\e^t \in \sF$ such that   $\e^t = \ss_\E(-\q^t)$ and $\q^t_3 \neq 0$, which will be a contradiction to $\sPF = \sF$  and the injectivity of $\ss_\E$.

    By part 1 of Lemma \ref{lemma:PF-F-support-map}, there exist $\q = (\q_1,\q_2,0) \in \widetilde{\bS}^2 $ such that $\e^* = \ss_\E(-\q)$.  Let $s\equiv\mathrm{sgn}(\e_3^*)\in\{-1,1\}$ and for $t>0$, define $\q^t\equiv\tfrac{(\q_1,\q_2,st)}
{\sqrt{\q_1^2+ \q_2^2+t^2}}$.
Then $\q^t\to q$ as $t\downarrow 0$.  By continuity of $\ss_\E$,
$\e^t=\ss_\E(-\q^t)\to \ss_\E(-\q) = \e^*$.
Hence, for sufficiently small $t>0$, $\mathrm{sgn}(\e^t_3)=s=\mathrm{sgn}(\q^t_3)$, so $\q^t_3\e^t_3>0$. 
By part 2 of Lemma \ref{lemma:PF-F-support-map}, this implies $\e^t\in\sF$ such that  $\e^t = \ss_\E(-\q^t)$ and $\q^t_3 \neq 0$, which is a contradiction.

 \noindent \textbf{Step 4:} Suppose that $\sPF \subset \{ \e_3 = 0 \}$. Note that $\sPF \subset \{ \e_3 = 0 \}$ implies that the third coordinate of both $R=\ss_\E(-u_1)$ and $B=\ss_\E(-u_2)$ is zero, where $u_1=(1,0,0)^\intercal$ and $u_2 = (0,1,0)^\intercal$; therefore, $u_3^\intercal R = 0$ and $u_3^\intercal B = 0$ where $u_3 = (0,0,1)$. Since $R \neq B$, it follows that $\min_{\e \in \E} u_3^\intercal \e <0$ due to strict convexity of $\E$. This implies that $\ss_\E(-u_3) \notin \sF$ by part 2 of Lemma \ref{lemma:PF-F-support-map}. Similarly, we can conclude $\ss_\E(u_3) \notin \sF$. 
 
 Since $q \notin \{u_3,-u_3\}$, it follows that $q_1 + q_2 >0$. 
 Let $\e^* \in \sF$. By part 2 of Lemma \ref{lemma:PF-F-support-map}, there exist $\q \in \widetilde{\bS}^2 $ such that $\e^* = \ss_\E(-q)$ and $\e^*_3 \q_3 \ge 0$. Let $\tilde{\q} = (\q_1, \q_2,0)/\sqrt{q_1^2 + q_2^2}$ and $\rho = \ss_\E(-\tilde{\q} )$. Consider the following derivations:
 \begin{align*}
     \q^\intercal \rho &\overset{(1)}{=} \q_1\rho_1+\q_2\rho_2
     \overset{(2)}{\le}
\q_1\e_1^*+\q_2\e_2^* 
\overset{(3)}{\le}
\q_1\e_1^*+\q_2\e_2^*+\q_3\e_3^*  
\overset{(4)}{\le} 
\q^\intercal \rho,
 \end{align*}
 where (1) holds since $\rho \in \sPF$ by part 1 of Lemma \ref{lemma:PF-F-support-map} and $\rho_3 = 0$ due to  $\sPF \subset \{ \e_3 = 0 \}$, (2) holds since  $\rho$ minimizes $\q_1\e_1+\q_2\e_2$ over $\E$, (3) holds since $\e^* \in \sF$ and $\e^*_3 \q_3 \ge 0$ by part 2 of Lemma \ref{lemma:PF-F-support-map}, and (4) holds since $\e^*$ minimizes $\q_1\e_1+\q_2\e_2  + \q_3 \e_3$ over $\E$. Therefore, $\rho$ also minimize $\q_1\e_1+\q_2\e_2  + \q_3 \e_3$ over $\E$. As a result, due to strict convexity of $\E$, we have that $\rho = \e^*$, which implies that $\e^* \in \sPF$.  Therefore, $\sF \subset \sPF$ which implies $\sF = \sPF$ since $\sPF \subset \sF$ due to Lemma \ref{lemma:PF-F-support-map}. 
\end{proof}

\subsection{Proofs for Section~\ref{sec:identif_and_estim}}
\subsubsection{Proofs for Section~\ref{subsec:PI-CR-SP}}\label{app:PI-classification-SP}
For $\lambda\in\Lambda$, let $\btheta_d(x;\lambda)$ denote the completion-dependent analogue of $\btheta_d(x)$.
Throughout Section \ref{app:PI-classification-SP}, we use classification loss to measure accuracy and statistical parity loss for fairness, under which
\begin{align*}
    L_0^{g,A} &= \tfrac{Y_i^* \one\{G=g\}}{\mu_g}, \hspace{2cm}L_1^{g,A} = \tfrac{(1-Y_i^*) \one\{G=g\}}{\mu_g}, \\
    L_0^{r,F} - L_0^{b,F}  &= 0,
    \hspace{2.2cm}L_1^{r,F} - L_1^{b,F}  = \tfrac{\one\{G=r\}}{\mu_r} - \tfrac{\one\{G=b\}}{\mu_b}.
\end{align*}

\begin{lemma}\label{lem:PI-endpoint-consequences}
\textit{Under Assumptions~\ref{asm:moments} and \ref{asm:endpoint-margin}, the necessary support conditions in \eqref{eq:PI-endpoint-support} hold and $A_{1,3}(X)\ne0$ a.s. For every binary rule $a^*\in\{0,1\}$, the completion \eqref{eq:PI-favorable-completion} satisfies
\begin{equation}
\sup_{\q\in\bS^2}
\bP_X\!\left(\left|\q^\intercal\Delta\btheta(X;\lambda(\cdot;a^*))\right|\le\delta\right)
\lesssim\delta^m,
\label{eq:PI-favorable-margin}
\end{equation}
where the $\lesssim$ in \eqref{eq:PI-favorable-margin} is associated with a constant independent of $a^*$. Thus $\E(\lambda(\cdot;a^*))$ is compact, convex, has nonempty interior, and has singleton support sets.}
\end{lemma}
\begin{proof}
Letting $\delta\downarrow0$ in Assumption \ref{asm:endpoint-margin} yields $\bP_\X(|\q^\intercal\Delta\btheta(\X;\lambda_j)|=0)=0$ for $j\in\{0,1\}$ and $\q\in\bS^2$. If $\bP(G=r\mid X)=0$ on a set of positive $\bP_\X$-probability, the first coordinate of $\Delta\btheta(\X;\lambda_j)$ is zero for $j\in\{0,1\}$; hence choosing $\q=(1,0,0)^\intercal$ yields $\bP_\X(|\q^\intercal\Delta\btheta(\X;\lambda_j)|=0)>0$, violating the implication of Assumption \ref{asm:endpoint-margin}. A similar argument applies if $\bP(G=b|\X)=0$. If $\bP(Z=1,Y=1\mid X)=0$ (resp., $\bP(Z=1,Y=0\mid X)=0$) on a set of positive $\bP_\X$-probability, then the vector $\Delta\btheta(X;\lambda_0)$ (resp., $\Delta\btheta(X;\lambda_1)$) is orthogonal to $(1,-1,-1)^\intercal/\sqrt{3}$ (resp., $(1,-1,1)^\intercal/\sqrt{3}$), violating the implication of Assumption \ref{asm:endpoint-margin}. This proves \eqref{eq:PI-endpoint-support}. The direction $(0,0,1)^\intercal$ similarly implies $A_{1,3}(X)\neq0$ a.s. 

As $a^*$ is binary, $\Delta\btheta(X;\lambda(\cdot;a^*))
=(1-a^*(X))\Delta\btheta(X;\lambda_0)
+a^*(X)\Delta\btheta(X;\lambda_1)$; hence,
$$\bP(|\q^\intercal\Delta\btheta(\X;\lambda(\cdot;a^*))|\leq\delta)\leq\bP(|\q^\intercal\Delta\btheta(\X;\lambda_0)|\leq\delta)+\bP(|\q^\intercal\Delta\btheta(\X;\lambda_1)|\leq\delta),$$ for any $\q\in\bS^2$, which proves \eqref{eq:PI-favorable-margin}. The asserted properties of the feasible set follow from Lemma~\ref{prop:compact_convex} and Theorem~\ref{thm:support_E} applied at $\lambda(\cdot;a^*)$.
\end{proof}

\begin{lemma}\label{lem:PI-favorable-comparison}
\textit{For every binary rule $a^*\in\{0,1\}$, $a\in\mathcal A(\sX)$, and $\lambda\in\Lambda$, we have
\begin{align}
\e_j(a^*;\lambda)-\e_j(a;\lambda)
={}&\e_j(a^*;\lambda(\cdot;a^*))-\e_j(a;\lambda(\cdot;a^*))\nonumber\\
&+2\bE\!\left[B_{jj}(X)|a(X)-a^*(X)|\,|\lambda_g(X)-\lambda_g(X;a^*)|\right],
\label{eq:PI-favorable-gain}
\end{align}
for $(j,g)=(1,r)$ or $(2,b)$; in particular, $\e_j(a^*;\lambda)-\e_j(a;\lambda)
\geq{}\e_j(a^*;\lambda(\cdot;a^*))-\e_j(a;\lambda(\cdot;a^*))$, i.e., the signed reduction in expected accuracy loss from replacing
$a^*$ with $a$ is minimized at the least-favorable $\lambda(\cdot;a^*)$.
Then,
\begin{equation}
\exists\lambda\in\Lambda:\ \e(a^*;\lambda)\in\sF(\lambda)
\quad\Longleftrightarrow\quad
\e(a^*;\lambda(\cdot;a^*))\in\sF(\lambda(\cdot;a^*)).
\label{eq:PI-favorable-equivalence}
\end{equation}}
\end{lemma}
\begin{proof}
Lemma~\ref{lem:PI-support-functions} gives
\begin{align}
\Delta\btheta(X;\lambda)
=\Delta\btheta(X;\lambda(\cdot;a^*))
+2B(X)\{a^*(X)\lambda_1-\lambda(X)\}.\label{eq:decompose_Delta}
\end{align}
By $\e_j(a^*;\lambda)-\e_j(a;\lambda)=\mathbb E[(a^*(\X)-a(\X))\Delta\btheta_j(\X;\lambda)]$, substituting \eqref{eq:decompose_Delta} and using the structure of $B(X)$ in \eqref{eq:PI-B-main} yields \eqref{eq:PI-favorable-gain}, since $a^*\in\{0,1\}$ implies $(a^*-a)(a^*-\lambda_g)=|a-a^*||a^*-\lambda_g|\geq0$. Since  $B_{jj}(\X)\geq0$ and the fairness coordinate is independent of $\lambda$, \eqref{eq:PI-favorable-gain} implies that any algorithm $a$ FA-dominating $a^*$ at the least-favorable completion $\lambda(\cdot;a^*)$ also FA-dominates it under every other $\lambda\in\Lambda$.
The reverse implication in \eqref{eq:PI-favorable-equivalence} uses $\lambda(\cdot;a^*)$ itself as the completion under which $\e(a^*;\lambda(\cdot;a^*))\in\sF(\lambda(\cdot;a^*))$.
\end{proof}

\begin{lemma}\label{lem:PI-scalar-regret}
\textit{Under Assumption~\ref{asm:moments}, for any $\lambda\in\Lambda$, any $a^*\in\mathcal A(\sX)$, and any $\q\in\R^3$,
\begin{align}
\q^\intercal\e(a^*;\lambda)+\h_{\E(\lambda)}(-\q)
={}&\bE\!\left[a^*(X)\{\q^\intercal\Delta\btheta(X;\lambda)\}_+\right]\nonumber\\
&+\bE\!\left[(1-a^*(X))\{-\q^\intercal\Delta\btheta(X;\lambda)\}_+\right].
\label{eq:PI-scalar-regret}
\end{align}
For binary $a^*$, the right-hand side reduces to $\bE\!\left[\left\{(2a^*(X)-1)\q^\intercal\Delta\btheta(X;\lambda)\right\}_+\right]$.}
\end{lemma}
\begin{proof}
By the pointwise support-function argument in the proof of Theorem~ \ref{thm:support_E},
\[
-\h_{\E(\lambda)}(-\q)
=\min_{a\in\mathcal A(\sX)}\q^\intercal\e(a;\lambda)
=\bE[\q^\intercal\btheta_0(X;\lambda)+\min\{0,\q^\intercal\Delta\btheta(X;\lambda)\}].
\]
Since $\q^\intercal\e(a^*;\lambda)=\bE[\q^\intercal\btheta_0(X;\lambda)+a^*(\X)\q^\intercal\Delta\btheta(X;\lambda)]$, we have $$\q^\intercal\e(a^*;\lambda)+\h_{\E(\lambda)}(-\q)=\bE\!\left[a^*(X)\q^\intercal\Delta\btheta(X;\lambda)-\min\{0,\q^\intercal\Delta\btheta(X;\lambda)\}\right].$$
Plugging $z=\q^\intercal\Delta\btheta(X;\lambda)$ and $t=a^*$ into the identity
$tz-\min\{0,z\}=tz_++(1-t)(-z)_+$ for $t\in[0,1]$ and $z\in\mathbb{R}$ yields \eqref{eq:PI-scalar-regret}; when $t\in\{0,1\}$ for the case of binary $a^*$, $tz-\min\{0,z\}=\{(2t-1)z\}_+$.
\end{proof}

\begin{proof}[Proof of Theorem~\ref{prop:PI-frontier-corner}]
By \eqref{eq:supportC}, for $q_1,q_2\ge0$,
\begin{align*}
\h_{\cC(\e(a^*;\lambda(\cdot;a^*)))}(\q)
={}&\q^\intercal\e(a^*;\lambda(\cdot;a^*))
+|q_3|\,|\e_3(a^*)|-q_3\e_3(a^*).
\end{align*}
Combining this identity with  Lemma~\ref{lem:PI-scalar-regret} for binary $a^*$ and that $|q_3|\,|\e_3(a^*)|-q_3\e_3(a^*)=2\{-\q_3\e_3(a^*)\}_+$ yields that $\h_{\cC(\e(a^*;\lambda(\cdot;a^*)))}(\q) + \h_{\E}(-q;\lambda(\cdot;a^*))$ equals \eqref{eq:PI-J-repaired}. By Lemma~\ref{lem:PI-favorable-comparison}, existence of an $\lambda\in\Lambda$ placing $a^*$ on $\sF(\lambda)$ is equivalent to $\e(a^*;\lambda(\cdot;a^*))\in\sF(\lambda(\cdot;a^*))$. Lemma~\ref{lem:PI-endpoint-consequences} shows Assumption~\ref{asm:margin} holds at $\lambda(\cdot;a^*)$; together with Assumption \ref{asm:moments}, we can apply Theorem~\ref{thm:FequivLLMO} to establish
\eqref{eq:PI-corner-main}. Both terms in \eqref{eq:PI-J-repaired} are nonnegative. Since the loss functions are bounded and $\{\cdot\}_+$ is Lipschitz, $J(\q;a^*)$ is continuous in $\q$; compactness of $\widetilde{\bS}^2$ then gives attainment of the minimum. 
Lemma \ref{lem:PI-favorable-comparison} then implies whenever \eqref{eq:PI-corner-main} holds, $\e(a^*;\lambda(\cdot;a^*))\in\sF(\lambda(\cdot;a^*))$ and the distribution associated with $\lambda(\cdot;a^*)$ satisfies Assumption~\ref{asm:margin} by Lemma~\ref{lem:PI-endpoint-consequences}.
\end{proof}

\begin{proof}[Proof of Proposition~\ref{prop:PI-uniform-LDA}]
By Lemma~\ref{lem:supp_C}, on $q_1,q_2\ge0$,
\begin{align}
J(\q;a^*)
={}&q_1\e_1(a^*;\lambda(\cdot;a^*))
+q_2\e_2(a^*;\lambda(\cdot;a^*))
+|q_3|\,|\e_3(a^*)|
+\h_\E(-\q;\lambda(\cdot;a^*)).
\label{eq:PI-J-support-identity}
\end{align}
We establish attainment of the minimum in \eqref{eq:PI-uniform-program} by showing that $J$ is coercive over $\{\q\in\R^3:\q_1,\q_2\ge0,\ \q_1+\q_2=1\}$. By the definition of $A_{1,3}$, $\bE[A_{1,3}(X)]=0$; by Lemma~\ref{lem:PI-endpoint-consequences}, $A_{1,3}(X)\ne0$ a.s. Therefore
\[
\bE[(A_{1,3}(X))_+]=\bE[(-A_{1,3}(X))_+]>0.
\]
The rules $\one\{A_{1,3}>0\}$ and $\one\{A_{1,3}<0\}$ attain, respectively, the disparity $\bE[(A_{1,3}(X))_+]$ and $-\bE[(-A_{1,3}(X))_+]$. Evaluating the support function in \eqref{eq:PI-J-support-identity} at $a\in\big\{\one\{A_{1,3}>0\}, \one\{A_{1,3}<0\}\big\}$ such that $-\q_3\e_3(a)=|\q_3|\bE[(A_{1,3}(X))_+]$, and using that classification loss lies in $[0,1]$, gives, for $q_1+q_2=1$, $\h_\E(-\q;\lambda(\cdot;a^*))\geq|\q_3|\bE[(A_{1,3}(X))_+]-1$, implying
$J(\q;a^*)\ge
\left\{|\e_3(a^*)|+\bE[(A_{1,3}(X))_+]\right\}|q_3|-1.$
The objective is thus coercive over $\{\q\in\R^3:\q_1,\q_2\ge0,\ \q_1+\q_2=1\}$, since $q_1,q_2\in[0,1]$ and hence $|\q|\to\infty$ if and only if $|\q_3|\to\infty$. Also $J((1,0,0)^\intercal;a^*)\le1$. Thus the minimization may be restricted, without changing its value, to the compact set given by $q_1\in[0,1]$, $q_2=1-q_1$, and
\[
|q_3|\le
\frac{2}{|\e_3(a^*)|+\bE[(A_{1,3}(X))_+]}.
\]
Since $J$ is continuous, the minimum is attained.

Choose any minimizer $\widetilde\q$ and its rule $\widetilde a$ in \eqref{eq:PI-uniform-algorithm}. For this proof write $\e^*=\e(a^*;\lambda(\cdot;a^*))$ and $\widetilde\e=\e(\widetilde a;\lambda(\cdot;a^*))$. 
By Lemma \ref{lem:PI-endpoint-consequences}, Assumption \ref{asm:margin} holds at $\lambda(\cdot;a^*)$. Theorem \ref{thm:support_E} and its proof therefore imply uniqueness of the minimizing algorithm and differentiability of $\q\mapsto h_\E(-\q;\lambda(\cdot;a^*))$ at $\widetilde \q$. Parameterize the first two coordinates by $(q_1,1-q_1)$. The derivative of $J(\q;a^*)$ with respect to $q_1$ is $(\e_1^*-\widetilde\e_1)-(\e_2^*-\widetilde\e_2)$. At an interior minimum, $(\e_1^*-\widetilde\e_1)=(\e_2^*-\widetilde\e_2)$. At $\widetilde q_1=0$, $(\e_1^*-\widetilde\e_1)\ge(\e_2^*-\widetilde\e_2)$, while at $\widetilde q_1=1$, $(\e_1^*-\widetilde\e_1)\le(\e_2^*-\widetilde\e_2)$. Hence
\begin{equation}
\widetilde q_1(\e_1^*-\widetilde\e_1)
+\widetilde q_2(\e_2^*-\widetilde\e_2)
=\min_{j\in\{1,2\}}(\e_j^*-\widetilde\e_j).
\label{eq:PI-uniform-accuracy-KKT}
\end{equation}
The subgradient condition for $q_3$ is $\widetilde\e_3\in|\e_3^*|\,\partial|\widetilde q_3|$, which is equivalent to
\begin{equation}
|\widetilde\e_3|\le|\e_3^*|,
\qquad
\widetilde q_3\widetilde\e_3=|\widetilde q_3|\,|\e_3^*|.
\label{eq:PI-uniform-fairness-KKT}
\end{equation}
For any $a\in\mathcal A(\sX)$ with $|\e_3(a)|\le|\e_3^*|$, \eqref{eq:PI-J-support-identity} implies, for any feasible $\q$,
\begin{align*}
J(\q;a^*)
&\ge\sum_{j=1}^2 q_j\{\e_j^*-\e_j(a;\lambda(\cdot;a^*))\}
+|q_3|\,|\e_3^*|-q_3\e_3(a)\ge\min_{j\in\{1,2\}}\{\e_j^*-\e_j(a;\lambda(\cdot;a^*))\}.
\end{align*}
At $(\widetilde\q,\widetilde a)$, both inequalities are equalities by \eqref{eq:PI-uniform-accuracy-KKT}-\eqref{eq:PI-uniform-fairness-KKT}. Therefore
\[
\Jstar=\min_{j\in\{1,2\}}(\e_j^*-\widetilde\e_j)
=\max_{\substack{a\in\mathcal A(\sX):\ |\e_3(a)|\le|\e_3^*|}}
\min_{j\in\{1,2\}}\{\e_j^*-\e_j(a;\lambda(\cdot;a^*))\}.
\]
By Lemma~\ref{lem:PI-favorable-comparison}, $\inf_{\lambda\in\Lambda}\min_{j\in\{1,2\}}\{\e_j(a^*;\lambda)-\e_j(a;\lambda)\}=\min_{j\in\{1,2\}}\{\e_j^*-\e_j(a;\lambda(\cdot;a^*))\}$. This proves \eqref{eq:PI-uniform-maximin}. The same lemma and \eqref{eq:PI-uniform-fairness-KKT} yield \eqref{eq:PI-uniform-improvement}.

We next establish the equivalence in \eqref{eq:PI-uniform-gap}, connecting the optimal uniform gain $\Jstar$ to the criterion \eqref{eq:PI-corner-main}. 
The function $J$ is positively homogeneous in $\q$. For any feasible $\q$ in \eqref{eq:PI-uniform-program}, $\|\q\|\ge1/\sqrt2$ and $\q/\|\q\|\in\widetilde{\bS}^2$. Therefore
\[
J(\q;a^*)=\|\q\|J(\q/\|\q\|;a^*)
\ge\frac1{\sqrt2}\min_{\v\in\widetilde{\bS}^2}J(\v;a^*).
\]
If $\min_{\v\in\widetilde{\bS}^2}J(\v;a^*)=0$, Theorem~\ref{prop:PI-frontier-corner} and Lemma~\ref{lem:PI-favorable-comparison} imply that $\e^*\in\sF(\lambda(\cdot;a^*))$. The maximum in \eqref{eq:PI-uniform-maximin} is at most zero and is at least zero because $a=a^*$ is feasible; hence $\Jstar=0$. The converse follows from $\Jstar\ge\frac1{\sqrt2}\min_{\v\in\widetilde{\bS}^2}J(\v;a^*)\ge0$.

Finally, suppose some $\e\in\sF(\lambda(\cdot;a^*))$  FA-dominates $\widetilde{\e}$. Then $\widetilde{\q}_1\e_1\leq\widetilde{\q}_1\widetilde\e_1$, $\widetilde{\q}_2\e_2\leq\widetilde{\q}_2\widetilde\e_2$, and $\widetilde{\q}_3\e_3\leq|\widetilde{\q}_3||\e_3|\leq|\widetilde{\q}_3||\widetilde\e_3|=\widetilde{\q}_3\widetilde\e_3$ since \eqref{eq:PI-uniform-fairness-KKT} implies $\widetilde q_3\widetilde\e_3\ge0$. These inequalities imply $\widetilde{\q}^\intercal\e\leq\widetilde{\q}^\intercal\widetilde\e$, contradicting the singleton support set at $\lambda(\cdot;a^*)$. Thus $\widetilde\e\in\sF(\lambda(\cdot;a^*))$.
\end{proof}

\begin{proof}[Proof of Theorem~\ref{thm:sF_envelope}]
Let $\widetilde{\cM} = \bigcup_{\q\in\widetilde{\bS}^2}
\left(\cM_q\cap\{\e\in\R^3:q_3\e_3\ge0\}\right)$. Using this notation, we divide the proof of \eqref{eq:define_sF_exists} in three steps. We first show that $\sF^\exists \subseteq \widetilde{\cM}$. We then show that $\widetilde{\cM} \subseteq \cl(\sF^\exists)$. Finally, we prove that $\widetilde{\cM}$ is compact.

\noindent \textbf{Step 1:} $\sF^\exists \subseteq \widetilde{\cM}$. By definition, $\sF^\exists=\bigcup_{\lambda\in\Lambda}\sF(\lambda)$.
Take $\e\in\sF^\exists$. Then $\e\in\sF(\lambda)$ for some $\lambda\in\Lambda$. By Lemma~\ref{lem:E_lambda}, there exists $\q\in\widetilde{\bS}^2$ such that $\e \in \ss_{\E(\lambda)}(-q)$, $q_3\e_3\ge0$, and $\e \in \cM_q$. Hence,  $\e\in\cM_q\cap\{\eta\in\R^3:q_3\eta_3\ge0\} \subset \widetilde{\cM}.$ 

\noindent \textbf{Step 2:} $\widetilde{\cM} \subseteq \cl(\sF^\exists)$. Suppose $\e \in \widetilde{\cM} $, then  $\e\in\cM_q$ and $q_3\e_3\ge0$ for some $q\in\widetilde{\bS}^2$. Let $(a,\lambda)$ be a pair representing $\e$ in the definition of $\cM_\q$. Fix a sequence $t_n \in (0,1)$ such that $t_n\downarrow0$. 
By Lemma~\ref{lem:lambda_t}, there is a sequence of $\lambda_{t_n} = (1-t_n) \lambda + t_n \lambda(\cdot;a) \in \Lambda$ such that $\e(a;\lambda_{t_n}) \in \sF(\lambda_{t_n}) \subset \sF^\exists$. Note that $\e(a;\lambda_{t_n})= (1-t_n) \e(a;\lambda) + t_n \e(a;\lambda(\cdot;a))$; therefore, $\e(a;\lambda_{t_n}) \to \e(a;\lambda) = \e$.  This implies that $\e \in \cl(\sF^\exists)$. 

\noindent \textbf{Step 3:} $\widetilde{\cM}$ is a compact set. 
Under classification loss and statistical parity, every feasible expected loss vector belongs to $[0,1]^2\times[-1,1]$, regardless of the completion. Hence, $\cM_q$ is uniformly bounded and $\widetilde{\cM}$ is bounded.
To show $\widetilde{\cM}$ is closed, take a sequence $\e_n \in \cM_{\q_n}$ for $\q_n \in \widetilde{\bS}^2$ such that $\e_n \to \e$, $q_{n_3}\e_{n_3}\geq0$, and, by compactness of $\widetilde{\bS}^2$, $\q_n \to \q \in \widetilde{\bS}^2$ and $q_3\e_3\geq0$. By Lemma~\ref{lem:frak_A_q},  $\cM_{\q_n}$ is compact and convex. Then, for every $v \in \bS^2$, $v^\intercal \e_n \le \h_{\cM_{q_n}}(v)$. By Lemma~\ref{lem:continuous_h_frakAq}, $v^\intercal \e \le \h_{\cM_{q}}(v)$. Hence, support-function characterization of a compact convex set implies $\e\in\cM_q \subset \widetilde{\cM}$, and thus $\widetilde{\cM}$ is closed.

We next show \eqref{eq:PI-envelope-closed-criterion}. Using the result in~\eqref{eq:define_sF_exists}, suppose $\e\in \cl(\sF^\exists)$. 
Then there exists $\q\in\widetilde{\bS}^2$ such that $q_3\e_3\ge0$ and $\e\in\cM_q$.
By Lemma~\ref{lem:frak_A_q}, $\cM_q$ is compact and convex, hence 
\begin{align*}
 \rho_q(\e)\equiv\sup_{\v\in\bS^2}
\{\v^\intercal\e-h_{\cM_q}(\v)\}\le0\quad\implies\quad[\rho_q(\e)]_+=0. 
\end{align*}
Since $q$ is feasible in the infimum defining $T_{\mathrm{env}}(\e)$ in~\eqref{eq:define:T_env_representation}, $0\le T_{\mathrm{env}}(\e)\le[\rho_q(\e)]_+=0$.
Therefore $T_{\mathrm{env}}(\e)=0$. Conversely, suppose $T_{\mathrm{env}}(\e)=0$.
By Lemma~\ref{lem:T_env_attains}, the infimum is attained. Hence there exists $q^*\in \{q\in\widetilde{\bS}^2:q_3\e_3\ge0\}$ such that $[\rho_{q^*}(\e)]_+=0$.
Thus, $\rho_{q^*}(\e)\le0$, implying $\e\in\cM_{q^*}$.
Since $q_3^*\e_3\ge0$, $ \e\in\cM_{q^*}\cap\{\eta\in\R^3:q_3^*\eta_3\ge0\}\subseteq \cl(\sF^\exists),$
where the last inclusion follows from~\eqref{eq:define_sF_exists}. This proves $\e\in\cl(\sF^\exists)\iff T_{\mathrm{env}}(\e)=0$.

Lastly, we establish \eqref{eq:PI-envelope-approximation}, which follows from using the sequence defined in Step 2 above and Lemma~\ref{lem:PI-segment-margin}. Moreover, the completion distribution $\lambda(\cdot;a)$ weakly reduces both accuracy losses while leaving the fairness coordinate unchanged; affinity in $\lambda$ therefore implies $\e_3(a;\lambda_{t_n})=\e_3$ and $\e_j(a;\lambda_{t_n})\le\e_j$, for $j=1,2$.
\end{proof}

\subsubsection{Proofs for Section~\ref{subsec:identification}}
\begin{proof}[Proof of Lemma~\ref{lemma:support-funct-identif}]
Since $\bpi(\X)\in(0,1)$ a.s., the ratio $\Z\bLL_d/\bpi(\X)$ is well-defined. 
Applying the law of iterated expectations,
\begin{align*}
\bE\left[\left.\tfrac{\Z\bLL_d}{\bpi(\X)}\right|\X\right]&=
\bE\left[\left.
\bE\left[\left.\tfrac{\Z\bLL_d}{\bpi(\X)}\right|\X,\Y^*,\G\right]\right|\X\right]=
\bE\left[\left.\tfrac{\bLL_d}{\bpi(\X)}\bE[\Z|\X,\Y^*,\G]\right| \X\right].
\end{align*}
By Assumption~\ref{asm:MAR-X}, $\bE[\Z| \X,\Y^*,\G]=\bpi(\X)$ a.s., and hence $\tfrac{\bLL_d}{\bpi(\X)}\bE[\Z|\X,\Y^*,\G]=\bLL_d$, a.s., and the characterization of $\btheta_d(\X)$ in \eqref{eq:thetaMAR} follows. 
Repeating the same argument and using \eqref{eq:support} yields the characterizations of $\h_\E(\q)$ and $e_g^\iota(a)$ in \eqref{eq:hMAR}-\eqref{eq:e_gMAR}.
\end{proof}

\subsection{Proofs for Section~\ref{sec:DML}}\label{app:DML}

Let $\{\omega_i\}_{i=1}^n$ be random positive weights independent of the data. Define  
\begin{align}\label{eq:DML-estimator-weighted}
    \widehat{\h}_\E^\omega(q; \widehat{\bLL}^\omega, \widehat{\boeta}) \equiv \tfrac{1}{n}    \sum_{i =1}^n  \left( \tfrac{\omega_i}{\overline{\omega}} \right)  \zeta_i(q; \widehat{\bLL}^\omega,\widehat{\boeta})~,
\end{align}
where  $\overline{\omega} = n^{-1} \sum_{i=1}^n \omega_i$ and  $\widehat{\bLL}^\omega$ is defined as in \eqref{eq:defLdgA}--\eqref{eq:defLvector} with $\mu_g$ replaced by $\widehat{\mu}_g^\omega = n^{-1} \sum_{i=1}^n \left( \tfrac{\omega_i}{\overline{\omega}} \right)  \one\{G_i=g\}$. Note that $\widehat{\h}_\E(q; \widehat{\bLL}, \widehat{\boeta}) = \widehat{\h}_\E^\omega(q; \widehat{\bLL}^\omega, \widehat{\boeta})$ by taking $\omega_i = 1 $, $\forall ~i$.

The next result presents a general version of Theorem \ref{thm:gaussian}. It presents the limiting distribution for the class of weighted DML estimator $\widehat{\h}_\E^\omega$ of the support function $\h_\E$. 

\begin{theorem}\label{thm:gaussian-h}
    \textit{Let Assumptions \ref{asm:moments}--\ref{asm:nuisance-funct-estimators} hold and let $\{\omega_i\}_{i=1}^n$ be random positive weights independent of the data such that $\bE[\omega_i] = 1$ and $\bE[\omega_i^2] < C_\omega$. Then, 
    $$ \sqrt{n}\left( \widehat{\h}_\E^\omega(\q; \widehat{\bLL}^\omega, \widehat{\boeta}) - \h_\E(q) \right) \Rightarrow \mathbb{G}[\omega_i \{\zeta_i^{*}(q;\bLL,\boeta)-\h_\E(q)\}]  \quad \text{in} \quad \ell^{\infty}(\bS^2)~,$$
    where $\zeta_i^{*}(q;\bLL,\boeta)$ and $\widehat{\h}_\E^\omega(\q; \widehat{\bLL}^\omega, \widehat{\boeta}) $ are defined in \eqref{eq:zeta-star} and \eqref{eq:DML-estimator-weighted}, respectively. 
}
\end{theorem}

 \begin{proof}
The proof has two parts. We first show in part 1 that
\begin{equation}\label{eq:aux1-thm-gaussian}
     \sup_{\q \in \bS^2 }  \left| \sqrt{n}(\widehat{\h}_\E^\omega(\q,\widehat{\bLL}^\omega; \widehat{\boeta}) - \h_\E(\q) ) -  \mathbb{G}_n[ \omega_i \{\zeta_i^{*}(q;\bLL,\boeta)-\h_\E(q)\}] \right| = o_p(1)~.
\end{equation}
We then use \citet[Thm. 2.5.2]{van:wel13} to conclude in part 2 that
\begin{equation}\label{eq:aux2-thm-gaussian}
    \mathbb{G}_n[\omega_i \{\zeta_i^{*}(q;\bLL,\boeta)-h_\E(q)\}] \Rightarrow  \mathbb{G}[\omega_i \{\zeta_i^{*}(q;\bLL,\boeta)-\h_\E(q)\}] \quad \text{in} \quad \ell^\infty(\bS^2)~.
\end{equation}

\noindent \textbf{Part 1:} Our goal is to establish \eqref{eq:aux1-thm-gaussian}. Define 
$$ \widehat{\h}_\E^{\omega,*}(\q; {\bLL}, {\boeta}) \equiv \tfrac{1}{n}    \sum_{i =1}^n  \left( \tfrac{\omega_i}{\overline{\omega}}  \right)  \zeta_i^*(q;\bLL,{\boeta})~.$$

By Lemma \ref{lemma:asympt-equivalence-weights}, we have
$\sup_{\q \in \bS^2 }  | \sqrt{n}(\widehat{\h}_\E^\omega(\q,\widehat{\bLL}^\omega; \widehat{\boeta}) -   \widehat{\h}_\E^{\omega,*}(\q; {\bLL}, {\boeta})  )   | = o_p(1)~. $
This implies that  \eqref{eq:aux1-thm-gaussian} holds if the equation below holds, 
\begin{equation}
        \sup_{\q \in \bS^2 }  | \sqrt{n}(\widehat{\h}_\E^{\omega,*}(\q;{\bLL}, {\boeta}) - \h_\E(\q) ) -  \mathbb{G}_n[ \omega_i \{\zeta_i^{*}(q;\bLL,\boeta)-\h_\E(q)\}] | = o_p(1)~.\label{eq:aux3-thm-gaussian}
\end{equation} 
Now, we rewrite the left-hand size of \eqref{eq:aux3-thm-gaussian} using the identity $ \mathbb{G}_n[ \omega_i \{\zeta_i^{*}(q;\bLL,\boeta)-h_\E(q)\}]  =  \overline{\omega} \sqrt{n}(\widehat{\h}_\E^{\omega,*}(\q,{\bLL}; {\boeta}) - \h_\E(\q) ) $, which gives us 
$ \eqref{eq:aux3-thm-gaussian}=|1 - \tfrac{1}{\overline{\omega}}|  \sup_{\q \in \bS^2 }  |   \mathbb{G}_n[ \omega_i \{\zeta_i^{*}(q;\bLL,\boeta)-\h_\E(q)\}] |. $
Since $ |1 - \tfrac{1}{\overline{\omega}}| = o_p(1)$  by the definition of the weights $\omega_i$ and the law of large numbers. Therefore,   \eqref{eq:aux3-thm-gaussian} holds if  $  \sup_{\q \in \bS^2 }  |   \mathbb{G}_n[ \omega_i \{\zeta_i^{*}(q;\bLL,\boeta)-\h_\E(q)\}] | = O_p(1)$, which holds by \eqref{eq:aux2-thm-gaussian} and the Continuous Mapping Theorem. Note that we can use \eqref{eq:aux2-thm-gaussian} since the proof of part 2 is independent of the proof of part 1. This completes the proof of part 1. 

\noindent \textbf{Part 2}: Our goal is to establish \eqref{eq:aux2-thm-gaussian}. Consider the random vector 
\begin{align*}
   W_i = (\bLL_{0,i}, \Delta \bLL_{i}, \Z_i, \bpi(\X_i),\Delta \btheta(\X_i), \btheta_0(\X_i),\omega_i) \in \mathcal{W} 
\end{align*}
and the class of functions $\mathcal{Q} \equiv \left\{ f: \mathcal{W} \to \R : f(W_i) = \omega_i \{\zeta_i^{*}(q;\bLL,\boeta)-\h_\E(q)\},\q \in \bS^2 \right\}$. Here, we are using the true value of the nuisance function $\boeta$---which is fixed---and treating $\boeta(\X_i)$ as a random vector. 
In order to use \citet[Thm. 2.5.2]{van:wel13} to complete the proof of part 2, we need to verify that (i) an envelope function of the class $\mathcal{Q}$ has finite second-moment, and (ii) the uniform entropy bound in \citet[(2.5.1)]{van:wel13} holds. To verify (i), we note that the minimum envelope $\sup_{\q \in \bS^2} | \omega_i \{\zeta_i^{*}(q;\bLL,\boeta)-\h_\E(q)\}|$ has finite second-moment due to Assumptions \ref{asm:moments} and \ref{asm:nuisance-function-structure}. To verify (ii), we rely on \citet[Thms.  2 and 3]{and94} that states that the class of functions of type I and their mixing (by addition and product operation) verify the uniform entropy bound. We conclude by noting that $\mathcal{Q}$ is of type I, since any function in $\mathcal{Q}$ can be written as the sum of (i) linear functions in $\q$, (ii) indicators of linear functions in $\q$, and (iii) product of (i) and (ii).     
 \end{proof}

\begin{proof}[Proof of Theorem \ref{thm:gaussian}] It follows from Theorem \ref{thm:gaussian-h} by taking $\omega_i=1,~ \forall ~i$ .
\end{proof}

\subsection{Proofs for Section~\ref{sec:testFA_Pareto}}
\begin{proof}[Proof of Proposition~\ref{prop:testable_F=PF}]
Theorem~\ref{thm:FequivLLMO} yields $T^{\texttt{LDA}}(\e)=0 \iff \e\in \sF$.
Next, we show that for $\e\in\E$, $\psi_\sPF(\e)=0 \iff \e\in\sPF$.
Indeed, for every $\alpha\in[0,1]$,
\begin{align*}
\q(\alpha)^\intercal\e+\h_\E(-\q(\alpha))=\q(\alpha)^\intercal\e-\inf_{\eta\in\E}\q(\alpha)^\intercal \eta
\ge 0.
\end{align*}
So, $\psi_\sPF(\e)=0$ if and only if there exists $\alpha\in[0,1]$ with $\q(\alpha)^\intercal\e=\inf_{\eta\in\E}\q(\alpha)^\intercal \eta$, i.e., $\e=\arg\min_{\eta\in\E}\q(\alpha)^\intercal \eta$.
By definition, this is equivalent to $\e\in\sPF$.
Since $\{\e\in \bB_C:T^{\texttt{LDA}}(\e)=0\}=\sF$, we have $\sup_{\e\in \bB_C:T^{\texttt{LDA}}(\e)=0}\psi_\sPF(\e)=\sup_{\e\in\sF}\psi_\sPF(\e)$.
Because $\psi_\sPF(\e)\ge 0$ for every $e\in\E$, the right hand side equals $0$ if and only if $\psi_\sPF(\e)=0$ for all $\e\in\sF$.

For $v\neq 0$, by positive homogeneity of $\h_\E(v)$,
$\h_\E(v)=\|v\|\h_\E\left(\tfrac{v}{\|v\|}\right)$,
so that, for every $\alpha\in[0,1]$ and $s\in[-\bar{s},\bar{s}]$, $\sqrt{1+s^2}\,
\h_\E\left(\tfrac{-q(\alpha)+s u_3}{\sqrt{1+s^2}}
\right)=\h_\E(-q(\alpha)+s u_3)$ for $u_3=[0,0,1]^\intercal$.
Indeed, $\| -q(\alpha)+s u_3\|=\sqrt{1+s^2}$, since
$q(\alpha)_3=0$ and $\|q(\alpha)\|=1$.

Fix $\alpha\in[0,1]$, set $q=q(\alpha)$, and consider the scalar function
\begin{align*}
g_\alpha(s)\equiv \h_\E(-q+s u_3)
=
\sup_{\eta\in\E}(-q+s u_3)^\intercal\eta .
\end{align*}
The function $g_\alpha$ is convex because it is the supremum of affine functions
of $s$. By Assumptions~\ref{asm:moments}-\ref{asm:margin}, $\E$ is compact and strictly convex, so each
support set is a singleton. Hence the support point $\ss_\E(-q)=\arg\max_{\eta\in\E}(-q)^\intercal\eta$ is well-defined.

We next show that $g_\alpha$ is differentiable at zero and that $g_\alpha'(0)=u_3^\intercal\ss_\E(-q)$.
Let $\e^*=\ss_\E(-q)$. For $s\neq0$, define the (unique) support point $\e^s
\equiv
\arg\max_{\eta\in\E}(-q+s u_3)^\intercal\eta$.
For $s>0$, optimality of $\e^*$ at
$s=0$ and optimality of $\e^s$ at $s$ give
\begin{align*}
(\e^*)_3
\leq
\tfrac{g_\alpha(s)-g_\alpha(0)}{s}
\leq
(\e^s)_3 .
\end{align*}
For $s<0$, the same inequalities hold with the order reversed:
\begin{align*}
(\e^s)_3
\leq
\tfrac{g_\alpha(s)-g_\alpha(0)}{s}
\leq
(\e^*)_3 .
\end{align*}
The support map is continuous by Berge's maximum theorem, so
$\e^s\to\e^*$ as $s\to0$. Therefore, $g_\alpha'(0)=(\e^*)_3=u_3^\intercal\ss_\E(-q)$.
Since $g_\alpha$ is convex and differentiable at zero, and since zero is an
interior point of $[-\bar{s},\bar{s}]$,
\begin{align*}
g_\alpha(0)
=
\inf_{s\in[-\bar{s},\bar{s}]}g_\alpha(s)
\quad\Longleftrightarrow\quad
g_\alpha'(0)=0.
\end{align*}
Combining this with the derivative calculation yields
\begin{align*}
\h_\E(-q(\alpha))
=
\inf_{s\in[-\bar{s},\bar{s}]}
\sqrt{1+s^2}\,
\h_\E\left(
\tfrac{-q(\alpha)+s u_3}{\sqrt{1+s^2}}
\right)
\quad\Longleftrightarrow\quad
u_3^\intercal\ss_\E(-q(\alpha))=0.
\end{align*}
The expression inside the supremum in \eqref{eq:Delta_sPF_h} is nonnegative, as $s=0$ is in the interval $[-\bar{s},\bar{s}]$. Hence, $\Delta_{\sPF}^{h}(\bar{s})=0$ if and only if $u_3^\intercal\ss_\E(-q(\alpha))=0$ for every $\alpha\in[0,1]$.
By Lemma~\ref{lemma:PF-F-support-map}, $\sPF=\{\ss_\E(-q(\alpha)):\alpha\in[0,1]\}$, so $\Delta_{\sPF}^{h}(\bar{s})=0$ if and only if $\sPF\subset\{\e\in\E:\e_3=0\}$.
Under Assumption~\ref{asm:no-kink}, Theorem~\ref{thm:Pareto_equal_FA} gives $\sPF=\sF$ if and only if $\sPF\subset\{\e\in\E:\e_3=0\}$.
Combining the last two equivalences proves $\sPF=\sF\iff\Delta_{\sPF}^{h}(\bar{s})=0$.
\end{proof}

\subsection{Proofs for Section~\ref{sec:test_statLDA}}
    Let $[\,\cdot\,]_+\equiv\max\{\cdot,0\}$. Define $\phi: \ell^\infty(\bS^2) \times \R^3 \to \R$ as
    \begin{equation}\label{eq:t-stat}
       \phi (f,\e) \equiv  \left[ \max_{\q \in \bS^2} (\q^\intercal \e -f(\q) )\right]_+ + \left[ \min_{\q \in \widetilde{\bS}^2} (\q_1 \e_1+\q_2\e_2 +|\q_3\e_3| + f(-\q) )\right]_+,
    \end{equation}
    where $f \in \ell^\infty(\bS^2) $ and $\e \in  \R^3 $. Note that $T_n^{\text{LDA}} = \sqrt{n} \left( \phi(\widehat{\bbeta}^\omega) - \phi(\bbeta) \right)$.
    
\begin{proof}[Proof of Theorem \ref{thm:test-bootstrap-weights}]
 
The proof has two parts. In the first part, we use \citet[Thms.  2.1 and 3.2]{fan:san19} to show that the Bayesian bootstrap is consistent. In the second part, we use \citet[Lemma A.1]{romano2012uniform} to prove that the test $\varphi_n^{\text{LDA}}$ has asymptotically correct size.

\noindent \textbf{Part 1:} By \citet[ Lemma D.3]{kai16} and standard arguments, it follows that $\phi$ defined in \eqref{eq:t-stat} is Hadamard directionally differentiable at $\bbeta$ tangentially to $\mathbb{D}_0 = C(\bS^2) \times \R^3 \subset \mathbb{D} = \ell^\infty(\bS^2) \times \R^3 $, where $C(\bS^2)$ is the class of continuous functions on $\bS^2$. This guarantees Assumption 1 in \cite{fan:san19}. By Cramer-Wold and Theorems \ref{thm:gaussian-h} and \ref{thm:gaussian-e}, we have $\sqrt{n}(\widehat{\bbeta}^\omega - \bbeta) \Rightarrow \mathbb{G}_{\bbeta}$. Let $K(\q,\q')$ be the covariance matrix of $\mathbb{G}_{\bbeta}$ for $\q,\q' \in \bS^2$. It can be show that $K(\q,\q')$ is continuous in both $\q,\q' \in \bS^2$, which implies that $\mathbb{G}_{\bbeta} \in \mathbb{D}_0$ with probability 1 by \citet[Lemma 1.3.1]{adler2007random}. This guarantees Assumption 2 in \cite{fan:san19}. By \citet[Thm. 2.1]{fan:san19}, we conclude $\sqrt{n} \left( \phi(\widehat{\bbeta}^\omega) - \phi(\bbeta)\right) \Rightarrow \phi_{\bbeta}'( \mathbb{G}_{\bbeta})$.

 To prove the consistency of the bootstrap, it is sufficient to show that
\begin{align}
\sup_{f\in\mathcal{BL}_1}\left|\bE\big[f\big(\widehat{\phi'}_{\bbeta}\big(\sqrt{n}\{\widetilde{\bbeta}^\omega-\widehat{\bbeta}^\omega\}\big)\big)\,\big|\,\{(\Y_i,\G_i,\X_i,\Z_i,\Xi_i)\}_{i=1}^n\big]-\bE\big[f\big(\phi_{\bbeta}'(\mathbb{G}_{\bbeta})\big)\big]\right|=o_p(1),\label{eq:aux1_proof_thm_test}
\end{align}
where $\mathcal{BL}_1$ is the set of $1$-Lipschitz functions $f:\mathbb{R}\to\mathbb{R}$ such that $|f|_\infty\leq1$. 

It is sufficient to verify Assumptions 3 and 4 in \cite{fan:san19} to conclude \eqref{eq:aux1_proof_thm_test} by   \citet[Thm. 3.2]{fan:san19}. Note that Assumption 3 part (i), (iii), and (iv) hold by construction. By \citet[Thm. 3.1]{HONG2018379}, we conclude $\widehat{\phi'}_{\bbeta}$ defined in Procedure 1 verifies Assumption 4 in \cite{fan:san19}. Therefore, it is sufficient to verify Assumption 3 part (ii), i.e., the Bayesian bootstrap works for the estimator $\widehat{\bbeta}^\omega $:
\begin{equation}\label{eq:aux1-proof-test-bootstrap}
    \sup_{f\in\mathcal{BL}_1}\left|\bE\left[f\left(  \sqrt{n} \left(  \widetilde{\bbeta}^\omega  -  \widehat{\bbeta}^\omega \right)  \right)\,\big|\,\{(\Y_i,\G_i,\X_i,\Z_i,\Xi_i)\}_{i=1}^n\right] -\bE\big[f\big( \mathbb{G}_{\bbeta})\big]\right|=o_p(1)~.
\end{equation}

    Define
    \begin{align}
        \widehat{\bbeta}^{\omega,*} &\equiv \left( \tfrac{1}{n}    \sum_{i =1}^n   \left( \tfrac{\omega_i^\h}{\overline{\omega^\h}} \right) \zeta_i^*(q;{\bLL},{\boeta}) , \tfrac{1}{n}    \sum_{i =1}^n      \left( \tfrac{\omega_i^\e}{\overline{\omega^\e}} \right)  \xi_i^*(a^*;{\bLL},{\boeta})\right)\label{eq:beta-hat-star}\\
         \widetilde{\bbeta}^{\omega,*} &\equiv  \left( \tfrac{1}{n}    \sum_{i =1}^n    \left( \tfrac{W_i \omega_i^\h}{\overline{W^\h}} \right) \zeta_i^*(q;{\bLL},{\boeta}) , \tfrac{1}{n}    \sum_{i =1}^n      \left( \tfrac{W_i \omega_i^\e}{\overline{W^\e}} \right)  \xi_i^*(a^*;{\bLL},{\boeta})\right) \label{eq:beta-tilde-star}
    \end{align}
    
By Lemma \ref{lemma:bootstrap-equivalence-weights}, we have $\left\|\sqrt{n} \left(  \widetilde{\bbeta}^\omega  -  \widehat{\bbeta}^\omega \right) - \sqrt{n} \left(  \widetilde{\bbeta}^{\omega,*}  -  \widehat{\bbeta}^{\omega,*} \right)   \right\| = o_p(1)$; therefore, we conclude that \eqref{eq:aux1-proof-test-bootstrap} holds if the following equation holds 
$$ \sup_{f\in\mathcal{BL}_1}\left|\bE\left[f\left(  \sqrt{n} \left(  \widetilde{\bbeta}^{\omega,*}  -  \widehat{\bbeta}^{\omega,*} \right)  \right)\,\big|\,\{(\Y_i,\G_i,\X_i,\Z_i,\Xi_i)\}_{i=1}^n\right] -\bE\big[f\big( \mathbb{G}_{\bbeta})\big]\right|=o(1)~. $$
Finally, the previous equation holds by definition of $ \widetilde{\bbeta}^{\omega,*}$, $ \widehat{\bbeta}^{\omega,*}$, by \citet[Thm. 3.6.13]{van:wel13}, and because $\sqrt{n}\left( \widehat{\bbeta}^{\omega,*}  - \bbeta \right) \Rightarrow \mathbb{G}_{\bbeta}$ since $E[\omega^\h] = E[\omega^\e] = 1$. This completes the proof of the bootstrap consistency.

\noindent \textbf{Part 2:} For an arbitrarily small $\kappa>0$,  $\bP\left(\sup_{x \in \R}  \widetilde{F}_n(x)-F_\infty(x) > \kappa \right)= o(1)$, where
\begin{align*}
     \widetilde{F}_n(x) &=  \bP\left(  \widehat{\phi'}_{\bbeta}\big(\sqrt{n}\{\widetilde{\bbeta}^\omega-\widehat{\bbeta}^\omega\}\big) \le x   | \{(\Y_i,\G_i,\X_i,\Z_i,\Xi_i)\}_{i=1}^n \right), ~~F_\infty(x)  =  \bP\left( \phi_{\bbeta}'( \mathbb{G}_{\bbeta}) \le x \right)~.
\end{align*}
Let $c_{1-\alpha}$ be $(1-\alpha)$-quantile of $\phi_{\bbeta}'( \mathbb{G}_{\bbeta})$, i.e., $F_\infty(c_{1-\alpha}) \ge 1-\alpha$. Recall that $ \widehat{c}_{1-\alpha + \kappa}$ is the $(1-\alpha+\kappa)-$quantile of the distribution of $\widehat{\phi'}_{\bbeta}\big(\sqrt{n}\{\widetilde{\bbeta}^\omega-\widehat{\bbeta}^\omega\}\big)$.  The following derivation completes the proof of part 2:
\begin{align*}
    \bP( T_n^{\text{LDA}} > \widehat{c}_{1-\alpha + \kappa}  + \kappa) &\overset{(1)}{\le} \bP( T_n^{\text{LDA}} > c_{1-\alpha}  + \kappa) + o(1) \overset{(2)}{\le}  \bP( T_n^{\text{LDA}} \ge c_{1-\alpha} + \kappa/2 ) + o(1) \\
    & \hspace{4cm}\overset{(3)}{\le}  \bP( \phi_{\bbeta}'( \mathbb{G}_{\bbeta}) \ge c_{1-\alpha} + \kappa/2 )  + o(1)  \\   
    &\hspace{4cm}\overset{(4)}{\le}   \bP( \phi_{\bbeta}'( \mathbb{G}_{\bbeta}) > c_{1-\alpha}  ) + o(1)~\overset{(5)}{\le} \alpha + o(1)~,
\end{align*}
where (1) holds since  $\bP( \widehat{c}_{1-\alpha + \kappa} \ge c_{1-\alpha}) = 1-o(1)$ by \citet[Lemma A.1]{romano2012uniform}, (2) and (4) hold since $\kappa >0$, (3) holds since $T_n^{\text{LDA}} \overset{d}{\rightarrow} \phi_{\bbeta}'( \mathbb{G}_{\bbeta})$, and (5) holds by the definition of $c_{1-\alpha}$. 
\end{proof}

\section{Auxiliary Results }\label{app:auxiliary}
\subsection{Lemmas and Auxiliary Results for Section~\ref{sec:setup}}\label{app:auxiliary:1}

\begin{lemma}[Support function of $\cC(\e^*)$]\label{lem:supp_C}
\textit{For $\e^*=(\e_1^*,\e_2^*,\e_3^*)\in\mathcal E$ and $\q\in\R^3$, the support function of the closed convex set $\cC(\e^*)$ in \eqref{eq:def_C} is as given in \eqref{eq:supportC}.}
\end{lemma}
\begin{proof}
We have
$\h_{\cC(\e^*)}(\q) = \sup_{\e \in \cC(\e^*)} \q^\intercal{\e}= \sup_{\substack{\e_1 \leq \e_1^*, \e_2 \leq \e_2^*\\ |\e_3| \leq |\e_3^*|}} (\q_1 \e_1 + \q_2 \e_2 + \q_3 \e_3).$
If $\q_1 < 0$, then $\sup_{\e_1 \leq \e_1^*} \q_1 \e_1 = +\infty$ (taking $\e_1 \to -\infty$). Similarly if $\q_2 < 0$.
For $\q_1, \q_2 \geq 0$, the supremum over $\e_j \leq \e_j^*$ gives $\q_j \e_j^*$, $j=1,2$.
For the third coordinate:
\begin{align*}
    \sup_{|\e_3| \leq |\e_3^*|} \q_3 \e_3 = \begin{cases}
\q_3 |\e_3^*| & \text{if } \q_3 \geq 0,\\
-\q_3 |\e_3^*| = |\q_3| |\e_3^*| & \text{if } \q_3 < 0.
\end{cases}
\end{align*}
Thus $\h_{\cC(\e^*)}(\q) = \q_1 \e_1^* + \q_2 \e_2^* + |\q_3| |\e_3^*|$ for $\q_1, \q_2 \geq 0$.
\end{proof}

\begin{lemma}\label{prop:compact_convex}
\textit{Let $(\Omega,\mathfrak F,\bP)$ be the probability space on which $(\Y,\G,\X)$ are defined. 
Then: (i) $\E$ is convex. 
(ii) Under Assumption \ref{asm:moments}, $\E$ is compact. 
(iii) Under Assumptions \ref{asm:moments}-\ref{asm:margin}, $\E$ has a non-empty interior in $\R^3$, and hence $\dim(\E)=3$.}
\end{lemma}

The proof of Lemma \ref{prop:compact_convex} is deferred to the \hyperref[online_appendix]{Online Appendix}.

\subsubsection{ Characterization of the Geometry of $\sF$ using the Support-Set of $\E$}\label{appn:B:geometry-sF}

Let Assumptions~\ref{asm:moments}-\ref{asm:margin} hold and let $\ss_\E(\q)$ be the support-set of $\E \subset \R^3$,
\begin{align*}
\ss_\E(\q) =\arg\max_{\e\in\E}\q^\intercal \e,
\qquad \q\in\bS^2.
\end{align*}
Since $\E$ is compact and strictly convex, $\ss_\E(\q)$ is a well-defined continuous function. 

\begin{lemma}\label{lemma:PF-F-support-map}
    \textit{Let $\sPF$ and $\sF$ be the Pareto and FA-frontiers defined in \eqref{eq:support_frontier} and \eqref{eq:sPF}, respectively, and let $\widetilde{\bS}^2\equiv\{q\in \bS^2:q_1\ge 0,\ q_2\ge 0\}$. Under Assumptions~\ref{asm:moments}-\ref{asm:margin}, we have
    \begin{enumerate}
        \item  $\sPF = \{ \e = \ss_\E(-\q): \q \in \widetilde{\bS}^2,\ \q_3 = 0  \}$
        \item $\sF = \{ \e = \ss_\E(-\q):  \q \in \widetilde{\bS}^2,\ \q_3 \e_3 \ge 0\} $
    \end{enumerate} 
    In particular, $\sPF \subseteq \sF$.}
\end{lemma}

\begin{proof}
The proof has two parts. We first characterize the Pareto and FA-frontiers using support functions. We then use the characterization to conclude the lemma.

\noindent \textbf{Part 1:}  For $\e^*=(\e_1^*,\e_2^*,\e_3^*)\in \E$, we define
\begin{align}
    \cC^P(\e^*) \equiv \{\e=(\e_1,\e_2,\e_3)\in\R^3: \e_1\le \e_1^*,\ \e_2\le \e_2^*\},\label{eq:def_CPv2}
\end{align}
the set of group-wise expected accuracy losses that are weakly more accurate than $\e^*$. It can be verified that the support function of $\cC^P(\e^*)$ equals:
\begin{align}
    \h_{\cC^P(\e^*)}(\q)\equiv\sup_{\e\in\cC^P(\e^*)}\q^\intercal\e=
\begin{cases}
\q_1 \e_1^* + \q_2 \e_2^*  & \text{if } \q_1\ge 0,\ \q_2\ge 0, \q_3 = 0\\
+\infty & \text{otherwise.}
\end{cases}\label{eq:supportCPv2}
\end{align} 

Let $\e^*=(\e_1^*,\e_2^*,\e_3^*)\in \E$, then $\e^* \in \sPF$ if and only if 
\begin{equation}\label{eq:aux1_thm23v2}
  \min_{\q\in \widetilde{\bS}^2,\ \q_3 = 0}
\left( \h_{\cC^P(\e^*)}(q)+\h_{\E}(-q) \right)=0~,  
\end{equation} 
and $\e^* \in \sF$ if and only if 
\begin{equation}\label{eq:aux2_thm23v2}
  \min_{\q\in \widetilde{\bS}^2}
\left( \h_{\cC(\e^*)}(q)+\h_{\E}(-q) \right)=0~,  
\end{equation} 
where \eqref{eq:aux1_thm23v2} follows by similar arguments to those in the proof of Theorem \ref{thm:FequivLLMO} \ref{thm:FequivLLMO_e_in_sF}-\ref{thm:FequivLLMO_support}, and \eqref{eq:aux2_thm23v2} is due to Theorem \ref{thm:FequivLLMO} \ref{thm:FequivLLMO_support}. 

\noindent \textbf{Part 2:} We use \eqref{eq:supportCPv2} and the support set $\ss_\E$ to note that  \eqref{eq:aux1_thm23v2}  can be written as follows
$$  \min_{\q\in \widetilde{\bS}^2,\ \q_3 = 0}
\left( \q_1 \e_1^* + \q_2 \e_2^* -  \q^\intercal \ss_\E(-\q)  \right)=0 ~.$$
Therefore, any $\e^* = \ss_\E(-\q)$ for $\q\in \widetilde{\bS}^2$ and $\q_3 =0 $ satisfies the previous equation, implying that $ \{ \e = \ss_\E(-\q) : \q \in \widetilde{\bS}^2,\ \q_3 = 0  \} \subset \sPF$. Similarly, for any $\e^*$ that satisfies the previous equation there exists $\q  \in \widetilde{\bS}^2$ such that $\q_3 = 0$ and $ \q_1 \e_1^* + \q_2 \e_2^* -  \q^\intercal \ss_\E(-\q) = 0$, which implies $\e^* = \ss_\E(-\q) $ and $\sPF \subset  \{ \e = \ss_\E(-\q): \q \in \widetilde{\bS}^2,\ \q_3 = 0  \}$. This proves Lemma \ref{lemma:PF-F-support-map}(1).

We derive a convenient characterization of $\sF$ to prove Lemma \ref{lemma:PF-F-support-map}(2). Fix $\e^*\in\E$ and $\q\in\tilde{\bS}^2$. Since $h_{\cC(\e^*)}(\q)=\q_1\e_1^*+\q_2\e_2^*+|\q_3||\e_3^*|$,
and $\h_{\E}(-\q)=-\min_{\e\in\E}\q^\intercal \e$,
we obtain
\begin{align*}
\h_{\cC(\e^*)}(\q)+\h_{\E}(-\q)
=
\q_1\e_1^*+\q_2\e_2^*+|\q_3||\e_3^*|-\min_{\e\in\E}\q^\intercal \e.
\end{align*}
Since $\min_{\e\in\E}\q^\intercal \e\le \q^\intercal \e^*
= \q_1\e_1^*+\q_2\e_2^*+\q_3\e_3^*$, it follows that
\begin{align*}
\h_{\cC(\e^*)}(\q)+\h_{\E}(-\q)
\ge |\q_3||\e_3^*|-\q_3\e_3^*\ge 0.
\end{align*}
Hence $\e^*\in\sF
\iff
\exists\,\q\in\tilde{\bS}^2:\,\,
\e^*=\ss_\E(-\q)
\text{ and }
\q_3\e_3^*\ge 0$, proving Lemma \ref{lemma:PF-F-support-map}(2).
\end{proof}

Lemma \ref{lemma:PF-F-support-map} has two immediate implications. First, $\sPF \subseteq \sF$. Second, the Pareto frontier $\sPF$ is a one-dimensional object, while the FA-frontier can be a two-dimensional object.  

Let $\widetilde{\bS}^{2}_{+}\equiv\{\q\in\widetilde{\bS}^{2}:\q_3\geq0\}$.
Define the coordinate-edge image 
\begin{align*}
\sIM_{ij}
\equiv
\left\{
\arg\min_{\e\in\E}\{\alpha\e_i+(1-\alpha)\e_j\}:
\alpha\in[0,1]
\right\},\quad 1\leq i<j\leq3.
\end{align*}

\begin{corr}
\label{corr:one-sided-geometry}
\textit{Suppose Assumptions~\ref{asm:moments}-\ref{asm:margin} hold and $\E\subset\{\e\in\R^3:\e_3\geq0\}$. Then
\begin{align}
\sF
=
\{\ss_\E(-\q):\q\in\widetilde{\bS}^{2}_{+}\}.\label{eq:sF_representation_tildeS+}
\end{align}
The images under $\q\mapsto\ss_\E(-\q)$ of the three boundary edges of $\widetilde{\bS}^{2}_{+}$ are
\begin{align*}
\sIM_{12}
=
\{\ss_\E(-\q):\q\in\widetilde{\bS}^{2}_{+},\ \q_3=0\}
=
\sPF,
\end{align*}
\begin{align*}
\sIM_{13}
=
\{\ss_\E(-\q):\q\in\widetilde{\bS}^{2}_{+},\ \q_2=0\},
\qquad
\sIM_{23}
=
\{\ss_\E(-\q):\q\in\widetilde{\bS}^{2}_{+},\ \q_1=0\}.
\end{align*}
If Assumption~\ref{asm:no-kink} also holds, then $\sIM_{12},\sIM_{13},\sIM_{23}$ are the relative boundary edge curves of $\sF$ and $\sIM_{12}=\sPF\subsetneq\sF$.
The case $\E\subset\{\e_3\leq0\}$ follows by the reflection
$(\e_1,\e_2,\e_3)\mapsto(\e_1,\e_2,-\e_3)$.}
\end{corr}

The proof of Corollary \ref{corr:one-sided-geometry} is deferred to the \hyperref[online_appendix]{Online Appendix}.

\begin{corr}
\label{corr:crossing-one-sided-pareto}
\textit{Let Assumptions~\ref{asm:moments}-\ref{asm:margin} hold. 
Suppose that $\E\cap\{\e\in\R^3:\e_3>0\}\neq\emptyset$, $\E\cap\{\e\in\R^3:\e_3<0\}\neq\emptyset$, and that either $\sPF\subset\{\e\in\R^3:\e_3>0\}$ or $\sPF\subset\{\e\in\R^3:\e_3<0\}$.
Recall $\E^A\equiv\pi(\E)$, and let
\begin{align*}
\sPF^A
\equiv
\left\{
c\in\E^A:
\not\exists c'\in\E^A
\text{ such that } c'_1\leq c_1,\ c'_2\leq c_2,
\text{ with at least one strict inequality}
\right\}.
\end{align*}
Then $\sPF^A\subsetneq\sF^A$ and $\sPF\subsetneq\sF$.}
\end{corr}

The proof of Corollary \ref{corr:crossing-one-sided-pareto} is deferred to the \hyperref[online_appendix]{Online Appendix}.

\begin{remark}[Implications for Figure~\ref{fig:F}]
Lemma~\ref{lemma:PF-F-support-map} identifies $\sPF$ with support directions in $\tilde{\bS}^2$ satisfying $\q_3=0$, and identifies $\sF$ with support directions in $\tilde{\bS}^2$ satisfying $\q_3\e_3\geq0$. 
Combining this characterization with Theorem~\ref{thm:Pareto_equal_FA} gives the classification used in Figure~\ref{fig:F}. 
Under Assumption~\ref{asm:no-kink}, $\sPF=\sF\iff\sPF\subset\{\e\in\E:\e_3=0\}$.
Thus, under Assumption~\ref{asm:no-kink}, whenever $\sPF\not\subset\{\e_3=0\}$, one has $\sPF\subsetneq\sF$, while if $\sPF\subset\{\e_3=0\}$, the two frontiers coincide. 
In particular, a single intersection of $\sPF$ with the plane
$\{\e_3=0\}$, or the fact that only the endpoints $R$ and $B$ lie in that plane, is not sufficient for $\sPF=\sF$.
Corollary~\ref{corr:crossing-one-sided-pareto} further shows that in the crossing case with $\sPF$ strictly on one side of $\{\e_3=0\}$, the strict inclusions $\sPF^A\subsetneq\sF^A$ and $\sPF\subsetneq\sF$ hold already under Assumptions~\ref{asm:moments}-\ref{asm:margin} alone.
\end{remark}

\subsubsection{Connection to \citetalias{lia:lu:mu:oku24}'s FA-Frontier} \label{app:auxiliary:LLMO-frontier}
\citetalias{lia:lu:mu:oku24} (Appendix O.1) work with a feasible set in $\R^2$ defined by the \emph{accuracy} loss:
\begin{align}
    \E^A=\{(e_r^A(a), e_b^A(a)) : a \in \mathcal{A}\} \subset \R^2.\label{eq:E_A_def}
\end{align}
They measure the \emph{unfairness} of an algorithm $a$ through $|f(a)|$, with $f:\mathcal A\to\R$ a linear functional.
This includes our fairness measure as a special case, where $f(a) = e_r^F(a)-e_b^F(a)$, the disparity in expected fairness loss between groups.
For each pair of expected accuracy losses $c = (c_r, c_b) \in \E^A$, they define the minimal achievable unfairness as
\begin{equation}\label{eq:min_disparity}
d(c)\equiv \min\{|f(a)|:(e_r^A(a),e_b^A(a))=c,\ a\in\mathcal A\}.
\end{equation}
\begin{defn}[\citetalias{lia:lu:mu:oku24}'s FA-dominance and FA-frontier ]\label{def:liang_frontier}
For $c,c' \in \E^A$, we say $c' \succ_{FA^{\text{\tiny LLMO}}} c$ if $c_r' \leq c_r,~c_b' \leq c_b,~d(c') \leq d(c)$, with at least one inequality strict.
Then the FA-frontier relative to the accuracy loss is:
\begin{align}
    \sF^A\equiv\{c \in \E^A: \nexists c' \in \E^A \text{ such that } c' \succ_{FA^{\text{\tiny LLMO}}} c\}\subset\R^2.\label{eq:sF_A:LLMO}
\end{align}
\end{defn}

\noindent We instead propose to work with the feasible set for $\left\{\left(e^A_r(a),e^A_b(a),e^F_r(a)-e^F_b(a)\right):a\in\mathcal A\right\}$ as in Definition \ref{def:feasible}, where we explicitly include a coordinate measuring the difference in expected fairness loss for the two groups.
Doing so allows us to characterize the properties of the FA-frontier, extending \citetalias{lia:lu:mu:oku24}'s results, and to provide a representation of it that is amenable to estimation and inference, even in the presence of selectively observed data.

In what follows, we establish that our FA-frontier $\sF$ in Definition \ref{def:R3_frontier} and the FA-frontier of \citetalias{lia:lu:mu:oku24} $\sF^A$ in Definition \ref{def:liang_frontier} coincide in what we argue is the only way that matters: given an algorithm $a^*\in\mathcal A$ and $\e^*\equiv(e^A_r(a^*),e^A_b(a^*),e^F_r(a^*)-e^F_b(a^*))$, $\e^*\in\sF$ if and only if $(\e_1^*,\e_2^*)\in\sF^A$ with $\e_1^*=e^A_r(a^*)$ and $\e_2^*=e^A_b(a^*)$, and $|\e_3^*|$ equals the minimum absolute disparity in \eqref{eq:min_disparity} associated with $(\e_1^*,\e_2^*)$. 
To show this, we define the projection $\pi: \R^3 \to \R^2$ by $\pi(x_1, x_2, x_3) = (x_1, x_2)$, and let $\pi(\E)=\{\pi(\e):\e\in\E\} = \E^A$.

\begin{prop}\label{prop:FequivLLMO_proj}
\textit{Let the assumptions of Theorem \ref{thm:FequivLLMO}. Then for $\e^*=(\e_1^*,\e_2^*,\e_3^*)\in \E$, 
\begin{align}
    \pi(\e^*)=(\e_1^*,\e_2^*)\in \sF^A \text{ and } |\e_3^*|=d\big((\e_1^*,\e_2^*)\big)\label{thm:FequivLLMO_proj}
\end{align}
is equivalent to the three statements in Theorem \ref{thm:FequivLLMO}.}
\end{prop}
\noindent An immediate consequence of Proposition~\ref{prop:FequivLLMO_proj} is that $\sF$ captures as many features of the tradeoff between accuracy and fairness as $\sF^A$ does in \citetalias{lia:lu:mu:oku24}'s original contribution, while preserving convexity of the feasible set and yielding a simple characterization of $\sF$.\footnote{In fact, $\sF$ yields a finer characterization than $\sF^A$. Consider algorithms $a,a'\in\mathcal A$ such that $\e:=\e(a)=(0.5,0.5,1)$ and $\e':=\e(a')=(0.5,0.5,0)$. Then $\e'\succ_{FA}\e$ but $c'\not\succ_{FA^{\text{\tiny LLMO}}} c$ for $c=(\e_1,\e_2)$ and $c'=(\e_1',\e_2')$.}

\begin{proof}[Proof of Proposition \ref{prop:FequivLLMO_proj}]
It suffices to show the equivalence between \eqref{thm:FequivLLMO_proj} and statement \ref{thm:FequivLLMO_E_single_inters_C} of Theorem \ref{thm:FequivLLMO}. 
We begin by showing \ref{thm:FequivLLMO_E_single_inters_C} $\imply \eqref{thm:FequivLLMO_proj}$.
Assume $\E\cap \cC(\e^*)=\{\e^*\}$.
We prove first that $|\e_3^*|=d((\e_1^*,\e_2^*))$.
If not, then $d((\e_1^*,\e_2^*))<|\e_3^*|$, and by definition of $d$, there exists $\tilde \e_3\in\R$ such that $(\e_1^*,\e_2^*,\tilde \e_3)\in \E$ and $|\tilde \e_3|<|\e_3^*|$.
But then $(\e_1^*,\e_2^*,\tilde \e_3)\in \E\cap \cC(\e^*)$ and this point is distinct from $\e^*$, contradicting $\E\cap \cC(\e^*)=\{\e^*\}$.
Hence, $|\e_3^*|=d((\e_1^*,\e_2^*))$.
We next show that $\pi(\e^*)=(\e_1^*,\e_2^*)\in \sF^A$.
Suppose not. 
Then there exists $(c_r,c_b)\in \E^A$ such that $c_r\le \e_1^*$, $c_b\le \e_2^*$, and $d((c_r,c_b))\le d((\e_1^*,\e_2^*))$, with at least one strict inequality. 
As $(c_r,c_b)\in \E^A$, by definition of $d$ there exists $\tilde \e_3\in\R$ such that $(c_r,c_b,\tilde \e_3)\in \E$ and $|\tilde \e_3|=d((c_r,c_b))$.
Using the equality already proved, $d((\e_1^*,\e_2^*))=|\e_3^*|$, we obtain $c_r\le \e_1^*$, $c_b\le \e_2^*$, and $|\tilde \e_3|\le |\e_3^*|$, with at least one strict inequality. 
Thus $(c_r,c_b,\tilde \e_3)\in \E\cap \cC(\e^*)$ with $(c_r,c_b,\tilde \e_3)\neq \e^*$, contradicting $\E\cap \cC(\e^*)=\{\e^*\}$.
Hence, $\pi(\e^*)\in \sF^A$, proving \ref{thm:FequivLLMO_E_single_inters_C} $\imply \eqref{thm:FequivLLMO_proj}$.

We now prove $\eqref{thm:FequivLLMO_proj}\imply$ \ref{thm:FequivLLMO_E_single_inters_C}.
Assume $\pi(\e^*)=(\e_1^*,\e_2^*)\in \sF^A$ and $|\e_3^*|=d((\e_1^*,\e_2^*))$.
Let $\e=(\e_1,\e_2,\e_3)\in \E\cap \cC(\e^*)$. 
Then $\e_1\le \e_1^*$, $\e_2\le \e_2^*$, and $|\e_3|\le |\e_3^*|$.
Since $\e\in \E$, we have
\begin{align}
d((\e_1,\e_2))\le |\e_3|\le |\e_3^*|=d((\e_1^*,\e_2^*)).\label{eq:dominance_chain_d}
\end{align}
Thus, $(\e_1,\e_2)\in \E^A$ weakly improves on $(\e_1^*,\e_2^*)$ in both group-accuracy losses and in the fairness index $d$. 
As $(\e_1^*,\e_2^*)\in \sF^A$, it follows that $(\e_1,\e_2)=(\e_1^*,\e_2^*)$ (else we would contradict $(\e_1^*,\e_2^*)\in \sF^A$).
Substituting back $(\e_1^*,\e_2^*)$ for $(\e_1,\e_2)$ in \eqref{eq:dominance_chain_d} yields $|\e_3|=|\e_3^*|$.
If $\e_3\neq \e_3^*$, then necessarily $\e_3=-\e_3^*$. If $\e_3^*=0$, this is impossible, so $\e_3=\e_3^*$. If $\e_3^*\neq 0$, by convexity of $\E$ the midpoint $\bar \e\equiv \tfrac{\e+\e^*}{2}=(\e_1^*,\e_2^*,0)$ belongs to $\E$. Therefore, $d((\e_1^*,\e_2^*))\le 0$, which implies $d((\e_1^*,\e_2^*))=0$.
But then $|\e_3^*|=d((\e_1^*,\e_2^*))=0$, contradicting $\e_3^*\neq 0$. Hence $\e_3=\e_3^*$, so $\e=\e^*$. We conclude that $\E\cap \cC(\e^*)=\{\e^*\}$, proving $\eqref{thm:FequivLLMO_proj}\imply$ \ref{thm:FequivLLMO_E_single_inters_C}.
\end{proof}

\subsection{Lemmas, Auxiliary Results, and Notation for Section~\ref{sec:identif_and_estim}}\label{app:auxiliary:2}

The next lemma gives the $\lambda$-dependent representations used in Section~\ref{subsec:PI-CR-SP}.

\begin{lemma}\label{lem:PI-support-functions}
\textit{For every $\lambda\in\Lambda$ and every $x\in\sX$,
\begin{align}
\btheta_0(x;\lambda)&=A_0(x)+B(x)\lambda(x),
&
\btheta_1(x;\lambda)&=A_1(x)-B(x)\lambda(x).\label{eq:PI-theta-app}
\end{align}
Consequently,
\begin{align}
\h_{\E(\lambda)}(\q)
&=
\bE\!\left[
\max\left\{
\q^\intercal\!\left(A_0(\X)+B(\X)\lambda(\X)\right),
\q^\intercal\!\left(A_1(\X)-B(\X)\lambda(\X)\right)
\right\}
\right].
\label{eq:PI-hE-app}  
\end{align} 
}
\end{lemma}

\begin{proof}
Fix $g\in\{r,b\}$. Using $\Y=\Z\Y^*$ and the definition of $\lambda_g(x)$,
\begin{align*}
\theta_0^{g,A}(x;\lambda)
&= \bEs\!\left[\left.\tfrac{\Y^*\one\{\G=g\}}{\mu_g}\right|\X=x\right] \\
&= \bE\!\left[\left.\tfrac{\Z\Y\one\{\G=g\}}{\mu_g}\right|\X=x\right]
+ \bE\!\left[\left.\tfrac{(1-\Z)\Y^*\one\{\G=g\}}{\mu_g}\right|\X=x\right] \\
&= \bE\!\left[\left.\tfrac{\Z\Y\one\{\G=g\}}{\mu_g}\right|\X=x\right]
+ \lambda_g(x)\,\bE\!\left[\left.\tfrac{(1-\Z)\one\{\G=g\}}{\mu_g}\right|\X=x\right],
\end{align*}
which is the $g$th accuracy coordinate of $A_0(x)+B(x)\lambda(x)$. Likewise,
\begin{align*}
\theta_1^{g,A}(x;\lambda)
&= \bEs\!\left[\left.\tfrac{(1-\Y^*)\one\{\G=g\}}{\mu_g}\right|\X=x\right] \\
&= \bE\!\left[\left.\tfrac{(\Z(1-\Y)+(1-\Z))\one\{\G=g\}}{\mu_g}\right|\X=x\right]
- \lambda_g(x)\,\bE\!\left[\left.\tfrac{(1-\Z)\one\{\G=g\}}{\mu_g}\right|\X=x\right],
\end{align*}
which is the $g$th accuracy coordinate of $A_1(x)-B(x)\lambda(x)$. For the fairness coordinate,
\[
\theta_0^F(x;\lambda)=0,
\qquad
\theta_1^F(x;\lambda)=\bE\!\left[\left.\tfrac{\one\{\G=r\}}{\mu_r}-\tfrac{\one\{\G=b\}}{\mu_b}\right|\X=x\right],
\]
which does not depend on $\lambda$. This proves \eqref{eq:PI-theta-app}. Equation~\eqref{eq:PI-hE-app} follows by pointwise maximization over $a(X)\in[0,1]$.  
\end{proof}

Let $\Delta A(x)\equiv A_1(x)-A_0(x)$, so that $\Delta\btheta(x;\lambda)\equiv\btheta_1(x;\lambda)-\btheta_0(x;\lambda)=\Delta A(x)-2B(x)\lambda(x)$, and $
\tau_q(x)\equiv\frac12 q^\intercal\Delta A(x)$. 
For each $x\in\sX$ and for $q\in\bS^2$, define 
\begin{align*}
\sU(x) & \equiv\{B(x)\s:\s\in[0,1]^2\},\qquad K_q^0(x) \equiv\{A_0(x)+u:u\in\sU(x),\ q^\intercal u\le\tau_q(x)\},~\\
K_q^1(x)&\equiv\{A_1(x)-u:u\in\sU(x), q^\intercal u\ge\tau_q(x)\}, \qquad K_q(x)\equiv K_q^0(x)\cup K_q^1(x).
\end{align*}

Define the finite constants $M_0\equiv\ess\sup_x\sup_{u\in\sU(x)}\|\Delta A(x)-2u\|<\infty$ and
\begin{align*}
H_0&\equiv\ess\sup_x\sup_{u\in\sU(x)}\max\{\|A_0(x)+u\|,\|A_1(x)-u\|\}<\infty,
\end{align*}
where the bounds follow from the binary-loss setup and the lower bound on $\mu_g,\,g\in\{r,b\}$.
With the convention that $\sup\{\emptyset\}=-\infty$, define
\begin{align}
\kappa_0(\v,q;x)&\equiv\sup_{u\in\sU(x):\, q^\intercal u\le\tau_q(x)}
\v^\intercal(A_0(x)+u),\label{eq:kappa0}\\
\kappa_1(\v,q;x)&\equiv\sup_{u\in\sU(x):\, q^\intercal u\ge\tau_q(x)}
\v^\intercal(A_1(x)-u).\label{eq:kappa1}
\end{align}
Recall the fixed-completion feasible set $\E(\lambda)$ in \eqref{eq:PI_Elambda} and its associated support set,
\begin{align*}
\ss_{\E(\lambda)}(v)&=\arg\max_{\eta\in \E(\lambda)}v^\intercal\eta.
\end{align*}
Thus $\ss_{\E(\lambda)}(-q)$ is the set of minimizers of $q^\intercal\eta$ over $\E(\lambda)$.

\begin{lemma}\label{lem:E_lambda}
\textit{Let Assumptions~\ref{asm:moments} and \ref{asm:endpoint-margin} hold in the binary-loss setup of Section~\ref{subsec:PI-CR-SP}. Then, $\E(\lambda)$ is nonempty, compact, and convex for every $\lambda\in\Lambda$. Furthermore, $\bP_X$ is atomless and for any $\e\in\sF(\lambda)$, there is $q\in\widetilde{\bS}^2$ such that $
\e\in\ss_{\E(\lambda)}(-q)$, $q_3\e_3\ge 0$, and $\e \in \cM_\q$.}
\end{lemma}
\begin{proof}
    Assumption~\ref{asm:moments}  and  Lemma~\ref{prop:compact_convex}(i)-(ii) imply that  $\E(\lambda)$ is nonempty, compact, and convex for every $\lambda\in\Lambda$. Assumption~\ref{asm:endpoint-margin} implies that $\bP_X$ is atomless. 

    In the proof of Theorem~\ref{thm:FequivLLMO}, Assumption~\ref{asm:moments} is used to show that (i) $\e^* \in \sF$ and (ii) $\E \cap \cC(\e^*) = \{\e^*\}$ are equivalent and both imply (iii) $\min_{\q \in \Tilde{S}^2} \left( \h_{\cC(\e^*)}(\q)+\h_{\E}(-\q) \right) = 0$. Since Assumption~\ref{asm:moments} holds for the distribution of $(Y^*,G,X,Z)$ implied by each $\lambda \in \Lambda$, it follows that for any $\e\in\sF(\lambda)$, there is $\q\in\widetilde{\bS}^2$ such that $\h_{\cC(\e)}(\q)+\h_{\E(\lambda)}(-\q) = 0$, which implies that $\e\in\ss_{\E(\lambda)}(-\q)$ and $\q_3\e_3\ge 0$. 

    To show $\e\in\sF(\lambda)$ implies $\e \in \cM_\q$, let $a \in \mathcal A(\sX)$ be such that $\e = \bE[ a(\X) \Delta \btheta(\X;\lambda) + \btheta_0(\X;\lambda)]$. Since $\q_3\e_3\ge 0$, it follows that $\h_{\cC(\e)}(\q)=\q^\intercal \e$.  
    By Lemma \ref{lem:PI-scalar-regret},
 \begin{align*}
\h_{\cC(\e)}(\q)+\h_{\E(\lambda)}(-\q)
={}&\bE\!\left[a(X)\{\q^\intercal\Delta\btheta(\X;\lambda)\}_++(1-a(X))\{-\q^\intercal\Delta\btheta(X;\lambda)\}_+\right]=0.
\end{align*}
Since the term inside of the expectation above is nonnegative for any $a(\X) \in [0,1]$, it has to be zero a.s. Therefore, $a(\X)=1 $ when $q^\intercal \Delta \btheta(\X;\lambda) < 0$; $a(\X) =0 $ when $q^\intercal \Delta \btheta(\X;\lambda) > 0$, and $a(\X) \in [0,1]$ when $q^\intercal \Delta \btheta(\X;\lambda) = 0$. 
Since $\bP_\X$ is nonatomic, by \citet[Thms.  2.1.26 and 2.1.38]{mol17}, there is a measurable binary algorithm that also generates the same expected loss vector $\varepsilon$ and is pointwise $\q$-optimal, proving 
    $\e \in \cM_\q$.
\end{proof}

\begin{lemma}\label{lem:frak_A_q}
Let Assumptions~\ref{asm:moments} and \ref{asm:endpoint-margin} hold in the binary-loss setup of Section~\ref{subsec:PI-CR-SP}. Then, for every $q\in\widetilde{\bS}^2$, the set $\cM_q$ is nonempty, compact, and convex.
\end{lemma}
\begin{proof}
    Since $u(x)\in\sU(x)=B(x)[0,1]^2$, for $j=1,2$, $u_j(x)=0$ whenever
$B_{jj}(x)=0$, and $0\le u_j(x)\le B_{jj}(x)$ whenever $B_{jj}(x)>0$.
Define
\begin{align*}
\lambda_j(x)=\begin{cases}
u_j(x)/B_{jj}(x), & B_{jj}(x)>0,\\
0, & B_{jj}(x)=0,
\end{cases}
\qquad j=1,2.
\end{align*}
Then $\lambda_j(x)\in[0,1]$.
Hence, for $u(x)\in\sU(x)$, $\bP_\X$-a.s., with $u$ measurable, there exists a measurable $\lambda\in\Lambda$, such that
\begin{align}
B(x)\lambda(x)=u(x)\qquad\bP_\X\text{-a.s.}\label{eq:selection_representation}
\end{align}
By measurability of $A_0,A_1,B$, $\Delta A$ and $\tau_q(x)=\frac12 q^\intercal\Delta A(x)$ are measurable.  
Since $B(x)$ has diagonal form, with nonnegative entries $B_{11}(x)$ and $B_{22}(x)$, the graph
\begin{align*}
\mathrm{Gr}(\sU)
=
\{(x,u)\in\mathcal X\times\mathbb R^3:
u_3=0,\ 0\le u_1\le B_{11}(x),\ 0\le u_2\le B_{22}(x)\}
\end{align*}
is measurable.  
Hence, for fixed $q$,
\begin{align*}
\mathrm{Gr}(K_q^0)
=
\{(x,y):(x,y-A_0(x))\in\mathrm{Gr}(\sU),\
q^\intercal(y-A_0(x))\le\tau_q(x)\}
\end{align*}
is measurable, and, similarly $\mathrm{Gr}(K_q^1)$ is measurable.  
Therefore $K_q=K_q^0\cup K_q^1$ has a measurable graph.
For each $x$, $\sU(x)\equiv B(x)[0,1]^2$ is compact.  
Intersecting $\sU(x)$ with a closed half-space preserves compactness, and translating or reflecting preserves compactness.  
Hence $K_q^0(x)$ and $K_q^1(x)$ are compact, possibly empty sets, and $K_q(x)\equiv K_q^0(x)\cup K_q^1(x)$ is compact.  
Moreover, $K_q(x)$ is nonempty because $\sU(x)\neq\emptyset$, and every
$u\in\sU(x)$ satisfies either $q^\intercal u\le\tau_q(x)$ or
$q^\intercal u\ge\tau_q(x)$.  
By \citet[Thm. 1.3.3]{mol17}, $K_q$ is a non-empty random compact set.
For each fixed $\v\in\mathbb S^2$, the support function of $K_q$ in direction $\v$, $\sup_{y\in K_q(x)} \v^\intercal y$ is measurable \citep[Prop. 1.3.8]{mol17} and admits a measurable maximizer.
The binary-loss bounds imply that $K_q$ is integrably bounded.
Consequently, the Aumann integral, $\left\{\mathbb E[k(X)]:k(x)\in K_q(x)\ \mathbb P_X\text{-a.s.}\right\}$, is well defined.

We next show that for every $q\in\widetilde{\bS}^2$, $\cM_q$ has the following Aumann representation \citep[Def. 3.1]{mol:mol18}:
\begin{align}
\cM_q=
\left\{\bE[k(X)]:k(x)\in K_q(x)\ \bP_\X\text{-a.s.}\right\}.\label{eq:frak_Aq_Aumann}
\end{align}
To see this, take an element of $\cM_q$. Then for some $\lambda\in\Lambda$ and some $q$-optimal rule $a$, let
\begin{align*}
k(x)=a(x)\btheta_1(x;\lambda)+(1-a(x))\btheta_0(x;\lambda).
\end{align*}
Let $u(x)=B(x)\lambda(x)$. If $a(x)=0$, the $q$-optimality condition gives
\begin{align*}
q^\intercal u(x)\le \frac12 q^\intercal\Delta A(x)=\tau_q(x).
\end{align*}
Thus $k(x)=A_0(x)+u(x)\in K_q^0(x)$.
If $a(x)=1$, the same argument gives $q^\intercal u(x)\ge \tau_q(x)$, so $k(x)=A_1(x)-u(x)\in K_q^1(x)$.
Hence $k(x)\in K_q(x)$ a.s., and the first inclusion follows.

Conversely, let $k$ be a measurable selector of $K_q$. Since $K_q^0$ and $K_q^1$ have measurable graphs and compact values, a measurable branch-selection argument gives measurable $a:\sX\to\{0,1\}$ and measurable $u(x)\in\sU(x)$ such that
\begin{align*}
k(x)=(1-a(x))(A_0(x)+u(x))+a(x)(A_1(x)-u(x)),
\end{align*}
with $q^\intercal u(x)\le\tau_q(x)$ on $\{a(x)=0\}$, and $q^\intercal u(x)\ge\tau_q(x)$ on $\{a(x)=1\}$.
By~\eqref{eq:selection_representation}, there exists $\lambda\in\Lambda$ with $B(x)\lambda(x)=u(x)$ a.s., and $a(x)\in\arg\min_{d\in\{0,1\}}q^\intercal\btheta_d(x;\lambda)$, so that $\bE[k(X)]\in\cM_q$.
It follows that $\cM_q$ is the Aumann integral of $K_q$. 
The correspondence $K_q$ is nonempty, measurable, compact-valued, and integrably bounded because $\sU(x)$ is compact and because $A_0,A_1,B$ are bounded in the binary-loss setup. 
Hence, for each $\q\in\bS^2$, $\cM_q$ is convex and compact by \citet[Thms.  2.1.26 and 2.1.38]{mol17}, where to apply Thm. 2.1.26 we used that $\bP_\X$ is nonatomic by Lemma~\ref{lem:E_lambda}.
\end{proof}

\begin{lemma}\label{lem:T_env_attains}
\textit{ Let Assumptions~\ref{asm:moments} and \ref{asm:endpoint-margin} hold in the binary-loss setup of Section~\ref{subsec:PI-CR-SP}. Then, for every fixed $\e\in\R^3$ and $\rho_q(\e)\equiv\sup_{\v\in\bS^2}\{\v^\intercal\e-h_{\cM_q}(\v)\}$, the mapping $q\mapsto \rho_q(\e)$ is continuous on $\widetilde{\bS}^2$. 
    Moreover, the infimum defining $T_{\mathrm{env}}(\e)$ in~\eqref{eq:define:T_env_representation} is attained.}
\end{lemma}
\begin{proof}
By Lemma~\ref{lem:continuous_h_frakAq}, for all $p,q\in\widetilde{\bS}^2$, 
\begin{align*}
 \sup_{\v\in\bS^2}|h_{\cM_p}(\v)-h_{\cM_q}(\v)|\lesssim\|p-q\|^{m/(m+1)}. 
\end{align*}
Therefore,
$|\rho_p(\e)-\rho_q(\e)|
\le
\sup_{\v\in\bS^2}
|h_{\cM_p}(\v)-h_{\cM_q}(\v)|\lesssim\|p-q\|^{m/(m+1)}$.
Hence $q\mapsto\rho_q(\e)$ is continuous. Now define $Q(\e)
=
\{q\in\widetilde{\bS}^2:q_3\e_3\ge0\}$.
This set is nonempty because $(1,0,0)^\intercal\in Q(\e)$, and it is compact because it is a closed subset of $\bS^2$. 
Since $q\mapsto[\rho_q(\e)]_+$ is continuous on $Q(\e)$, it attains its minimum. 
Thus, $T_{\mathrm{env}}(\e)=\min_{q\in Q(\e)}[\rho_q(\e)]_+$.
\end{proof}

\begin{lemma}\label{lem:continuous_h_frakAq}
\textit{Let Assumptions~\ref{asm:moments} and \ref{asm:endpoint-margin} hold in the binary-loss setup of Section~\ref{subsec:PI-CR-SP}. Then, for every $q\in\widetilde{\bS}^2$ and $v\in\bS^2$,
\begin{equation}
\h_{\cM_q}(v)
=\bE\!\left[\max\{\kappa_0(v,q;X),\kappa_1(v,q;X)\}\right].
\label{eq:sup_fun_mathfrakAq}
\end{equation}
Moreover, $q_n\to q$ implies $h_{\cM_{q_n}}(\v)\to h_{\cM_q}(\v)$, and for all $p,q\in\bS^2$, 
\begin{equation}
\sup_{v\in\bS^2}|\h_{\cM_p}(v)-\h_{\cM_q}(v)|
\lesssim\|p-q\|^{m/(m+1)}.
\label{eq:cont:h_frakAq}
\end{equation}
Furthermore, $q\mapsto\cM_q$ is continuous in Hausdorff distance on $\widetilde{\bS}^2$.
}
\end{lemma}

The proof of Lemma \ref{lem:continuous_h_frakAq} is deferred to the \hyperref[online_appendix]{Online Appendix}.

\begin{lemma}\label{lem:lambda_t}
    \textit{Let Assumptions~\ref{asm:moments} and \ref{asm:endpoint-margin} hold in the binary-loss setup of Section~\ref{subsec:PI-CR-SP}. For every fixed $\q\in\widetilde \bS^2$ and $\e \in \cM_q$, let $(a,\lambda)$ be any pair representing $\e$ in the definition of $\cM_\q$. For $t\in(0,1)$, define $\lambda_t=(1-t)\lambda + t\lambda(\cdot;a)$, where $\lambda(\cdot;a)$ is as in \eqref{eq:PI-favorable-completion}.  If $q_3\e_3\ge 0$, then $\e(a;\lambda_t) \in \sF(\lambda_t)$.  Furthermore, for Lebesgue-almost every $t \in (0,1)$ it follows that the completion distribution $\lambda_t$ satisfies a margin conditions similar to Assumption~\ref{asm:margin}.}
\end{lemma}

The proof of Lemma \ref{lem:lambda_t} is deferred to the \hyperref[online_appendix]{Online Appendix}.

\begin{lemma}[Margin regularization along a segment]\label{lem:PI-segment-margin}
\textit{Let $(V,W)$ be a pair of bounded random vectors in $\mathbb R^3$, and suppose that, for some $m\in(0,1]$ and $C<\infty$,
\begin{equation}
    \sup_{v\in \bS^2} \bP(|v^\intercal W|\le \delta )\le C \delta^m,
\qquad \delta>0.\label{eq:condition_PI-segment-margin}
\end{equation} 
Set $Z_t=(1-t)V+tW$. For Lebesgue-almost every $t\in(0,1)$,
there is $C_t<\infty$ such that
\begin{equation}
\sup_{q\in\bS^2}\bP(|v^\intercal Z_t|\le \delta) \le C_t \delta^{m/24}, \qquad\delta>0.\label{eq:PI-segment-margin}
\end{equation}}
\end{lemma}

The proof of Lemma \ref{lem:PI-segment-margin} is deferred to the \hyperref[online_appendix]{Online Appendix}.

\subsection{Lemmas and Auxiliary Results for Sections~\ref{sec:DML}-\ref{sec:tests}}
Recall that $\mathfrak{F}_k $ denotes the $\sigma$-algebra generated by all the data except the ones with indices in $\mathcal{I}_k$.  

\begin{lemma}\label{lemma:delta-hat}
     Let Assumptions \ref{asm:nuisance-function-structure} and \ref{asm:nuisance-funct-estimators} hold. Then, there exist $\widehat{\delta}_k: \sX_2 \to \R^+$ and an event $\widehat{E}_k \in \mathfrak{F}_k$ such that (i) $\widehat{\delta}_k$ conditional on $\mathfrak{F}_k$ is nonstochastic, (ii) $\bP(\widehat{E}_k^c) = o(1)$, (iii) conditional on $\widehat{E}_k$ we have that $\widehat{\bmu}_k \in \bB_\epsilon(\mu_r,\mu_b)$,    
     $\|\Delta \widehat{\btheta}_k(\X_i)-\Delta {\btheta}(\X_i)\| \le \widehat{\delta}_k(\X_{2,i})$ and $\|\widehat{\bpi}_k(\X_i)-{\bpi}(\X_i)\| \le \widehat{\delta}_k(\X_{2,i})$, and (iv) $\bE[ \widehat{\delta}_k(\X_{2,i})^2 | \mathfrak{F}_k] = o_p(n^{-1/2})$. 
\end{lemma}
\begin{proof}
    Let $C_F = \sup_{\nu \in \R^{d_\nu}} \sup_{ \gamma\in \R^{d_\gamma}} \sup_{\mu \in \bB_\epsilon(\mu_r,\mu_b) }\| D F(\nu, \gamma,\mu)\|$. Consider the event $\widehat{E}_k = \{ \| \widehat{\bmu}_k - \bmu\|< \epsilon\}$, which holds with probability approaching one by definition of $\widehat{\bmu}_k$. 
    By Assumptions \ref{asm:nuisance-function-structure} and \ref{asm:nuisance-funct-estimators}, the mean value inequality, and conditional on $\widehat{E}_k$, we have
    $$\left\| \left( \widehat{\bpi}(\X_i)-{\bpi}(\X_i)~,~\Delta \widehat{\btheta}(\X_i)-\Delta {\btheta}(\X_i) \right) \right\|$$
    is lower than or equal to 
    $C_F \left( \| \widehat{\nu}_k(\X_{1,i}) - \nu(\X_{1,i})\| +\| \widehat{\gamma}_k(\X_{2,i}) - \gamma(\X_{2,i})\| + \| \widehat{\bmu}_k  - \bmu \|\right)~.$
    By Assumption \ref{asm:nuisance-funct-estimators}, the previous expression is lower than or equal to 
    $$  \widehat{\delta}_k(\X_{2,i}) \equiv  C_F \left( \sup_{x_1 \in \sX_1} \|\widehat{\nu}_k(x_1) - \nu(x_1)\| +\| \widehat{\gamma}_k(\X_{2,i}) - \gamma(\X_{2,i})\| + \| \widehat{\bmu}_k  - \bmu \|\right)~.$$ %
    By Assumption \ref{asm:nuisance-funct-estimators} and Jensen's inequality, we have $\bE[\widehat{\delta}_k(\X_{2,i})^2  | \mathfrak{F}_k] =  o_p(n^{-1/2})$.    
\end{proof}

\begin{lemma}\label{lemma:L2-alpha}
    Let Assumption \ref{asm:nuisance-function-structure} and \ref{asm:nuisance-funct-estimators} hold. Then,
    $ n_k^{-1/2} \sum_{i \in \mathcal{I}_k } | \widehat{\balpha}_k^\h(\q; \X_i) - \balpha^\h(\q;\X_i)|^2 = o_p(1)$  %
    uniformly on $\q \in \bS^2$.
\end{lemma}

\begin{proof}

We can write
$n_k^{-1/2} \sum_{i \in \mathcal{I}_k } | \widehat{\balpha}_k^\h(\q; \X_i) - \balpha^\h(\q;\X_i)|^2 \le \tilde{A}_k + B_k + C_k$,
where
\begin{align*}
    \tilde{A}_k &= n_k^{-1/2} \sum_{i \in \mathcal{I}_k } |\q^\intercal \Delta \btheta (\X_i) |   \one\{ |\q^\intercal \Delta\btheta(\X_i)| \le |\q^\intercal \Delta \widehat{\btheta}_k(\X_i) - \q^\intercal \Delta\btheta(\X_i)| \}  \\
    B_k &= n_k^{-1/2} \sum_{i \in \mathcal{I}_k } | \q^\intercal \widehat{\btheta}_{0,k}(\X_i)  - \q^\intercal \btheta_0(\X_i) |^2,
    C_k = n_k^{-1/2} \sum_{i \in \mathcal{I}_k } | \q^\intercal \Delta \widehat{\btheta}_k (\X_i) - \q^\intercal \Delta \btheta (\X_i) |^2. 
\end{align*}
In the proof of Lemma \ref{lemma:condition-A-weights}, we conclude  $\tilde{A}_k = o_p(1)$ uniformly on $\q \in \bS^2$, with $\tilde{A}_k$ defined in \eqref{eq:aux1-condition-A-weights}. Finally, note that $B_k$ and $C_k$ both are $o_p(1)$ uniformly on $\q \in \bS^2$ by Cauchy-Schwartz and Assumption \ref{asm:nuisance-funct-estimators}.  
\end{proof}

\subsubsection{Auxiliary Results for Weighted Estimators for the Support Function}\label{sec:weighted-estimator-h}

Let $\{ \omega_i: 1 \le i \le n\}$ be a random sample of positive weights that are  independent of the data such that $ \bE[\omega_i^2] \le C_\omega$.  Define  
\begin{equation}\label{eq:oracle-weighted-h}
   \widehat{\h}_\E^{\omega,*}(q; \bLL, \boeta) \equiv \tfrac{1}{n}    \sum_{i =1}^n  \left( \tfrac{\omega_i}{\overline{\omega}}  \right)  \zeta_i^*(q;\bLL,{\boeta})~, 
\end{equation}
where $\overline{\omega} = n^{-1} \sum_{i=1}^n \omega_i$  and $ \zeta_i^*(q;\bLL,{\boeta}) $ is defined in \eqref{eq:zeta-star}. 

\begin{lemma}\label{lemma:asympt-equivalence-weights}
    \textit{Let Assumptions \ref{asm:moments}--\ref{asm:nuisance-funct-estimators} hold. Then, 
        \begin{equation}\label{eq:lemma-asympt-equivalence-weights} 
        \sup_{\q \in \bS^2} n^{1/2} \left| \widehat{\h}_\E^\omega(q; \widehat{\bLL}^\omega, \widehat{\boeta})  - \widehat{\h}_\E^{\omega,*}(q; \bLL, \boeta) \right| = o_p(1)~,
    \end{equation} 
    where $\widehat{\h}_\E^\omega(q; \widehat{\bLL}^\omega, \widehat{\boeta})$ and $\widehat{\h}_\E^{\omega,*}(q; \bLL, \boeta) $ are defined in \eqref{eq:DML-estimator-weighted} and \eqref{eq:oracle-weighted-h}, respectively.
    }
\end{lemma}

\begin{proof} 

\noindent Let $\widehat{\balpha}_k^\h$ be the estimator of $\balpha^\h$ defined in \eqref{eq:alpha} using $\widehat{\boeta}_k$. Recall $n_k \equiv n/K$ and  define  $I_{j,n}(\q) \equiv \tfrac{1}{K} \sum_{k=1}^K \tfrac{1}{n_k}\sum_{i \in \mathcal{I}_k}   \left(\tfrac{\omega_i}{\overline{\omega}} \right) \mathfrak{Z}_{k,i}^{(j)}(\q) $ for $j \in \{0,1,2,3\}$, where
\begin{align*}
   \mathfrak{Z}_{k,i}^{(0)}(\q) &\equiv     \tfrac{\q^\intercal \widehat{\bLL}_{0,i}^\omega\Z_i}{\widehat{\bpi}_k (\X_i)}  + \left(\tfrac{\q^\intercal \Delta \widehat{\bLL}_{i}^\omega \Z_i}{\widehat{\bpi}_k (\X_i)}\right)  \one\{\q^\intercal \Delta\widehat{\btheta}_k(\X_i)>0\} + \widehat{\balpha}_k^\h(\q;\X_i) \left(1 - \tfrac{\Z_i}{\widehat{\bpi}_k (\X_i)}\right)  \\
    \mathfrak{Z}_{k,i}^{(1)}(\q) &\equiv    \tfrac{\q^\intercal \widehat{\bLL}_{0,i}^\omega\Z_i}{\widehat{\bpi}_k (\X_i)}  + \left(\tfrac{\q^\intercal \Delta \widehat{\bLL}_{i}^\omega \Z_i}{\widehat{\bpi}_k (\X_i)}\right)  \one\{\q^\intercal \Delta\btheta(\X_i)>0\} + \widehat{\balpha}_k^\h(\q;\X_i) \left(1 - \tfrac{\Z_i}{\widehat{\bpi}_k (\X_i)}\right) \\
    \mathfrak{Z}_{k,i}^{(2)}(\q)  &\equiv    \tfrac{\q^\intercal {\bLL}_{0,i} \Z_i}{\widehat{\bpi}_k (\X_i)}  + \left(\tfrac{\q^\intercal \Delta {\bLL}_{i}  \Z_i}{\widehat{\bpi}_k (\X_i)}\right)  \one\{\q^\intercal \Delta\btheta(\X_i)>0\}
     + \widehat{\balpha}_k^\h(\q;\X_i) \left(1 - \tfrac{\Z_i}{\widehat{\bpi}_k (\X_i)}\right)  \\ 
    &\quad + \sum_{g \in \{r,b\}}  \Gamma_g^h(\q) \left(1-\tfrac{\one \{ G_i = g\}}{\mu_g}\right) \\
    \mathfrak{Z}_{k,i}^{(3)}(\q)  &\equiv  \zeta_i^*(\q;\bLL, \boeta)~.
\end{align*}
Our goal is to establish \eqref{eq:lemma-asympt-equivalence-weights}. We divide the proof in three claims. 

\noindent \underline{Claim 1}: $\sup_{\q \in \bS^2 } \sqrt{n} | I_{0,n}(\q) - I_{1,n}(\q)| = o_p(1)$. 

We write 
$I_{0,n}(\q) - I_{1,n}(\q) = \tfrac{1}{K} \sum_{k=1}^K (A_k + B_k)~,$
where
\begin{align}
    \hspace{-.4cm}A_k &= \tfrac{1}{n_k} \sum_{i \in \mathcal{I}_k} \left( \tfrac{\omega_i}{\overline{\omega}}  \right)\left(\tfrac{\q^\intercal \Delta \bLL_{i} \Z_i}{\bpi(\X_i)}\right) \left( \one\{\q^\intercal \Delta \widehat{\btheta}_k(\X_i)>0\} -   \one\{\q^\intercal \Delta\btheta(\X_i)>0\}\right) \label{eq:lemma-condition-A}\\
    \hspace{-.4cm}B_k &= \tfrac{1}{n_k} \sum_{i \in \mathcal{I}_k} \left( \tfrac{\omega_i}{\overline{\omega}}  \right)\left( \tfrac{\q^\intercal \Delta \widehat{\bLL}_{i}^\omega \Z_i}{\widehat{\bpi}_k(\X_i)} - \tfrac{\q^\intercal \Delta \bLL_{i} \Z_i}{\bpi(\X_i)}\right) \left( \one\{\q^\intercal \Delta \widehat{\btheta}_k(\X_i)>0\} -   \one\{\q^\intercal \Delta\btheta(\X_i)>0\}\right) \label{eq:lemma-condition-B}
\end{align} 
Lemma \ref{lemma:condition-A-weights} shows that $A_k = o_p(n^{-1/2})$ uniformly on $q \in \bS^2$ by relying on Assumptions \ref{asm:nuisance-function-structure} and \ref{asm:nuisance-funct-estimators} and the approach used in the proof of Theorem 4.1 in \cite{liu:mol25v2}. Using similar ideas, Lemma \ref{lemma:condition-B-weights} shows that $B_k = o_p(n^{-1/2})$ uniformly on $q \in \bS^2$. Since $K$ is fixed as $n \to \infty$, we conclude the claim of step 1.

\noindent \underline{Claim 2}: $\sup_{\q \in \bS^2 } \sqrt{n} |I_{1,n}(\q) - I_{2,n}(\q)| = o_p(1)$. 

Recall $\widehat{\mu}_g^\omega \equiv n^{-1} \sum_{i=1}^n \left( \tfrac{\omega_i}{\overline{\omega}} \right) \one\{G_i=g\}$.   
Note that $n^{-1} \sum_{i=1}^n  \left( \tfrac{\omega_i}{\overline{\omega}} \right) \left(1 - \tfrac{\one\{G_i = g\} }{ \widehat{\mu}_g^\omega} \right)=0$, which implies
\begin{align*}
    I_{1,n}(\q) &= \tfrac{1}{n} \sum_{i=1}^n  \left( \tfrac{\omega_i}{\overline{\omega}} \right) \left[\zeta_i(\q;\widehat{\bLL}^\omega, (\Delta {\btheta},\widehat{\bpi},\widehat{\balpha}^h)) + \sum_{g \in \{r,b\}} \Gamma_g^h(\q) \left(1 - \tfrac{\one\{G_i = g\} }{ \widehat{\mu}_g^\omega} \right)\right].
\end{align*} 

Since 
\begin{align*}
    I_{2,n}(\q) &= \tfrac{1}{n} \sum_{i=1}^n  \left( \tfrac{\omega_i}{\overline{\omega}} \right) \left[\zeta_i(\q;{\bLL}, (\Delta{\btheta},\widehat{\bpi},\widehat{\balpha}^h)) +  \sum_{g \in \{r,b\}} \Gamma_g^h(\q) \left(1 - \tfrac{\one\{G_i = g\} }{ {\mu}_g} \right)\right]~,
\end{align*}
we can write
\begin{align*}
    I_{1,n}(\q) - I_{2,n}(\q) =     \sum_{g \in \{r,b\}} \left( \tfrac{1}{\widehat{\mu}_g^\omega} -  \tfrac{1}{\mu_g} \right)  \left( \tfrac{1}{K} \sum_{k=1}^K A_k^g + B_k^g \right)~,
\end{align*}
where
\begin{align*}
    A_k^g &=  \tfrac{1}{n_k} \sum_{i \in \mathcal{I}_k} \left( \tfrac{\omega_i}{\overline{\omega}} \right)  \left\{ \q^\intercal \left( \bLL_{0,i}    + \Delta\bLL_i \one\{\q^\intercal \Delta\btheta(\X_i)>0\} \right) \tfrac{\Z_i \mu_g}{\bpi(\X_i)}  - \Gamma_g^h(\q) \right\} \one\{G_i=g\}  \\
    B_k^g &= \tfrac{1}{n_k} \sum_{i \in \mathcal{I}_k}  \left( \tfrac{\omega_i}{\overline{\omega}} \right) \left(  \tfrac{\Z_i \mu_g}{\widehat{\bpi}_k(\X_i)} -\tfrac{\Z_i \mu_g}{\bpi(\X_i)} \right) \q^\intercal \left( \bLL_{0,i}    + \Delta\bLL_i   \one\{\q^\intercal \Delta\btheta(\X_i)>0\} \right)  \one\{G_i=g\} ~.
\end{align*}
We conclude that $A_k^g = O_p(n^{-1/2})$ uniformly on $\q \in \bS^2$ by using that the weights $\omega_i$ are independent of the data, $\overline{\omega} = \bE[\omega_i] + O_p(n^{-1/2})$, Assumptions \ref{asm:moments} and \ref{asm:MAR-X}, the definition of $\Gamma_g(q)$, and the Central Limit Theorem. We also conclude that $B_k^g = o_p(1)$ uniformly on $\q \in \bS^2$ by Cauchy-Schwartz, and Assumptions \ref{asm:moments}, \ref{asm:nuisance-function-structure}, and \ref{asm:nuisance-funct-estimators}. Finally, we conclude $  \tfrac{1}{\mu_g} - \tfrac{1}{\widehat{\mu}_g^\omega} = O_p(n^{-1/2})$ uniformly on $\q \in \bS^2$ by the Central Limit Theorem, Assumption \ref{asm:moments}, and the definition of the weights $\omega_i$.

\noindent \underline{Claim 3}: $\sup_{\q \in \bS^2 } \sqrt{n} \left| I_{2,n}(\q) - I_{3,n}(\q)\right| = o_p(1)$.

We  write
${\sqrt{n}} \left( I_{2,n}(\q) - I_{3,n}(\q)\right) = \tfrac{1}{\sqrt K} \sum_{k=1}^K (A_k + B_k+ C_k),$
where
\begin{align*}
    A_k &= \tfrac{1}{\sqrt{n_k}} \sum_{i \in \mathcal{I}_k} \left( \tfrac{\omega_i}{\overline{\omega}} \right)  \left( \tfrac{\Z_i}{\widehat{\bpi}_k(\X_i)} - \tfrac{\Z_i}{\bpi(\X_i)}\right) \left( \q^\intercal\bLL_{0,i}  + \left(\q^\intercal \Delta \bLL_{i} \right)  \one\{\q^\intercal \Delta\btheta(\X_i)>0\} - \balpha^\h(\q;\X_i)   \right), \\
    B_k &= \tfrac{1}{\sqrt{n_k}} \sum_{i \in \mathcal{I}_k} \left( \tfrac{\omega_i}{\overline{\omega}} \right)  \left( \widehat{\balpha}_k^\h(\q;\X) -\balpha^\h(\q;\X)  \right) \left( 1 - \tfrac{\Z_i}{\bpi(\X_i)} \right), \\
    C_k &= \tfrac{1}{\sqrt{n_k}} \sum_{i \in \mathcal{I}_k} \left( \tfrac{\omega_i}{\overline{\omega}} \right)  \left( \widehat{\balpha}_k^\h(\q;\X) -\balpha^\h(\q;\X)  \right) \left( \tfrac{\Z_i}{{\bpi}(\X_i)} - \tfrac{\Z_i}{\widehat{\bpi}_k(\X_i)}\right).
\end{align*}
Recall that $ (\Delta \widehat{\btheta}_k, \widehat{\bpi}_k, \widehat{\btheta}_{0,k})$ was estimated using all the data except the ones with indices in $\mathcal{I}_k$, and $\mathfrak{F}_k $ is the $\sigma$-algebra generated by all the data except the ones with indices on $\mathcal{I}_k$. 
The definition of the weights $\omega_i$, Assumptions \ref{asm:nuisance-function-structure} and \ref{asm:nuisance-funct-estimators}, the definition of $\balpha^h(q;X)$ in \eqref{eq:alpha}, and algebra show $\bE[ A_k | \mathfrak{F}_k] = 0$ and $\sup_{q \in \bS^2} \bE[A_k^2 | \mathfrak{F}_k] = o(1)$. Similarly, we can show $\bE[ B_k | \mathfrak{F}_k] = 0$ and $\sup_{q \in \bS^2} \bE[B_k^2 | \mathfrak{F}_k] = o(1)$. By \citet[Lemma 6.1]{ChernozhukovDML}, we conclude $A_k$ and $B_k$ both are $o_p(1)$ uniformly on $\q \in \bS^2$. Finally, we conclude that $C_k = o_p(1)$ uniformly on $\q \in \bS^2$ by Cauchy-Schwartz, Assumptions \ref{asm:nuisance-function-structure} and \ref{asm:nuisance-funct-estimators}, Lemma \ref{lemma:L2-alpha}, and the definition of the weights $\omega_i$.
 
The triangle inequality and Claims 1, 2, and 3, complete the proof of the lemma.    
\end{proof}

\begin{lemma}\label{lemma:condition-A-weights}
    Let Assumptions \ref{asm:moments}--\ref{asm:nuisance-funct-estimators} hold. Then, $A_k = o_p(n^{-1/2})$ uniformly on $q \in \bS^2$, with $A_k$ defined in \eqref{eq:lemma-condition-A}.  
\end{lemma}
\begin{proof}
It is sufficient to prove that $\bE[ A_k | \mathfrak{F}_k,  (\X_i)_{i \in \mathcal{I}_k} ] = o_p(n^{-1/2})$ to conclude that $A_k$ is $o_p(n^{-1/2})$, by \citet[Lemma 6.1]{ChernozhukovDML}. Define
\begin{equation}\label{eq:aux1-condition-A-weights}
  \widetilde{A}_k \equiv  \tfrac{1}{n_k} \sum_{i \in \mathcal{I}_k} \left| \q^\intercal \Delta\btheta(\X_i)\right|  \one\{ |\q^\intercal \Delta\btheta(\X_i)| \le |\q^\intercal \Delta \widehat{\btheta}_k(\X_i) - \q^\intercal \Delta\btheta(\X_i)| \}  ~.  
\end{equation}
Note that Law of Iterated Expectations, Assumption \ref{asm:moments}, and the definitions of the weights $\omega_i$ imply $|\bE[ A_k | \mathfrak{F}_k,  (\X_i)_{i \in \mathcal{I}_k} ]| \le  \widetilde{A}_k$; therefore, it is sufficient to show that $\widetilde{A}_k$ is $o_p(n^{-1/2})$ uniformly on $\q \in \bS^2$. By Cauchy-Schwartz and Lemma \ref{lemma:delta-hat}, we have $|\q^\intercal \Delta \widehat{\btheta}(\X_i) - \q^\intercal \Delta\btheta(\X_i)| \le \widehat{\delta}_k(\X_{2,i})$ conditional on $\widehat{E}_k$, where $\widehat{\delta}_k $ and $\widehat{E}_k$ are as in Lemma \ref{lemma:delta-hat}. Since $\bP(\widehat{E}_k^c) = o(1)$, it is sufficient to show that $\bE[ \widetilde{A}_k \one\{ \widehat{E}_k\} | \mathfrak{F}_k ] = o_p(n^{-1/2})$ to conclude $\widetilde{A}_k$ is $o_p(n^{-1/2})$ . We verify this next:
\begin{align*}
    \bE[ \widetilde{A}_k  \one\{ \widehat{E}_k\} | \mathfrak{F}_k ] 
    &= \bE \left [ \left| \q^\intercal \Delta\btheta(\X_i)\right|  \one\{ |\q^\intercal \Delta\btheta(\X_i)| \le |\q^\intercal \Delta \widehat{\btheta}_k(\X_i) - \q^\intercal \Delta\btheta(\X_i)| \}  \one\{ \widehat{E}_k\} | \mathfrak{F}_k  \right] \\
    &\overset{(1)}{\le}  \bE \left [  \left| \q^\intercal \Delta\btheta(\X_i)\right|  \one\{ |\q^\intercal \Delta\btheta(\X_i)| \le   \widehat{\delta}_k(\X_{2,i})  \} | \mathfrak{F}_k  \right]  \one\{ \widehat{E}_k\} \\
    &\overset{(2)}{\le} \bE \left [ \widehat{\delta}_k(\X_{2,i}) \bE \left [  \one\{ |\q^\intercal \Delta\btheta(\X_i)| \le   \widehat{\delta}_k(\X_{2,i})  \} | \mathfrak{F}_k, \X_{2,i}  \right]  | \mathfrak{F}_k  \right]  \one\{ \widehat{E}_k\} \\
    &\overset{(3)}\lesssim \bE \left [  \widehat{\delta}_k(\X_{2,i})^2 | \mathfrak{F}_k  \right] \overset{(4)}{=} o_p(n^{-1/2})
\end{align*}
where (1) holds by Lemma \ref{lemma:delta-hat} and since $\one\{ \widehat{E}_k\} \in \mathfrak{F}_k$, (2) holds since  $ \widehat{\delta}_k(\X_{2,i})$ is nonstochastic conditional on $\mathfrak{F}_k$ and $\X_{2,i}$, (3) holds by Assumption \ref{asm:nuisance-function-structure} and since $ \one\{ \widehat{E}_k\} \le 1$, and (4) holds by Lemma   \ref{lemma:delta-hat}. Finally, note that $\bP(\widehat{E}_k^c) =o(1)$ and (4) hold uniformly on $\q \in \bS^2$ since these objects do not depend on $\q$. 
\end{proof}

\begin{lemma}\label{lemma:condition-B-weights}
    Let Assumptions \ref{asm:moments}--\ref{asm:nuisance-funct-estimators} hold. Then,  $B_k = o_p(n^{-1/2})$ uniformly on $q \in \bS^2$, with $B_k$ defined in \eqref{eq:lemma-condition-B}.  
\end{lemma}

\begin{proof} We can write $B_k =   B_{k,1} + B_{k,2} $ , where
\begin{align*}
    B_{k,1} &= \tfrac{1}{n_k} \sum_{i \in \mathcal{I}_k} \left( \tfrac{\omega_i}{\overline{\omega} } \right) \q^\intercal \Delta \bLL_{i} \Z_i\left( \tfrac{1}{\widehat{\bpi}_k(\X_i)} - \tfrac{1}{\bpi(\X_i)}\right) \left( \one\{\q^\intercal \Delta \widehat{\btheta}(\X_i)>0\} -   \one\{\q^\intercal \Delta\btheta(\X_i)>0\}\right) \\
    B_{k,2} &= \tfrac{1}{n_k} \sum_{i \in \mathcal{I}_k} \left( \tfrac{\omega_i}{\overline{\omega} } \right) \left( \tfrac{\q^\intercal \Delta \widehat{\bLL}_{i}^\omega \Z_i - \q^\intercal \Delta \bLL_{i} \Z_i}{\widehat{\bpi}_k(\X_i)}  \right) \left( \one\{\q^\intercal \Delta \widehat{\btheta}(\X_i)>0\} -   \one\{\q^\intercal \Delta\btheta(\X_i)>0\}\right) 
\end{align*}
As in the proof of Lemma \ref{lemma:condition-A-weights}, it is sufficient to show that $\bE[B_{k,1} | \mathfrak{F}_k, (\X_i)_{i\in \mathcal{I}_k}]$ and $B_{k,2}$ are both $o_p(n^{-1/2})$ uniformly on $\q \in \bS^2$. 
Define 
\begin{align*}
     \widetilde{B}_{k,1} &= \tfrac{1}{n_k} \sum_{i \in \mathcal{I}_k} |\widehat{\bpi}_k(\X_i) - \bpi(\X_i)| \one\{ |\q^\intercal \Delta\btheta(\X_i)| \le |\q^\intercal \Delta \widehat{\btheta}_k(\X_i) - \q^\intercal \Delta\btheta(\X_i)| \}
\end{align*}
Note that Law of Iterated Expectations, Assumptions \ref{asm:moments} and \ref{asm:nuisance-function-structure}, and the definitions of the weights $\omega_i$ imply that $|\bE[B_{k,1} | \mathfrak{F}_k, (\X_i)_{i\in \mathcal{I}_k}]| \le  \left( \tfrac{1}{ \epsilon} \right) \widetilde{B}_{k,1}$; therefore, it is sufficient to show that $\widetilde{B}_{k,1}$ is $o_p(n^{-1/2})$ uniformly on $\q \in \bS^2$ to conclude the same for $\bE[B_{k,1} | \mathfrak{F}_k, (\X_i)_{i\in \mathcal{I}_k}]$. By similar arguments used in the proof of Lemma \ref{lemma:condition-A-weights} and by relying on Lemma \ref{lemma:delta-hat},
$$ \bE[ \widetilde{B}_{k,1} \one\{\widehat{E}_k \}| \mathfrak{F}_k] \lesssim \bE \left [  \widehat{\delta}_k(\X_{2,i})^2 | \mathfrak{F}_k  \right] = o_p(n^{-1/2})~,$$
where $ \widehat{E}_k$ and $\widehat{\delta}_k$ are as in  Lemma \ref{lemma:delta-hat}. By the next identity
$$ \q^\intercal \Delta \widehat{\bLL}_{i}^\omega  - \q^\intercal \Delta \bLL_{i}  =  \sum_{g \in \{r, b\}} \mu_g \left( \tfrac{1}{\widehat{\mu}_g^\omega} - \tfrac{1}{{\mu}_g} \right) \left( q^\intercal \Delta \bLL_i\right)  \one \{G_i = g \}~\implies B_{k,2} = \sum_{g \in \{r,b\}} \mu_g \left( \tfrac{1}{\widehat{\mu}_g^\omega} - \tfrac{1}{{\mu}_g} \right) B_{k,2}^g,$$
where
\begin{align*}
        B_{k,2}^g &=  \tfrac{1}{n_k} \sum_{i \in \mathcal{I}_k} \left( \tfrac{\omega_i}{\overline{\omega} } \right) \left( \tfrac{   \Delta \bLL_i   \one \{G_i = g \}  \Z_i }{\widehat{\bpi}_k(\X_i)}  \right) \left( \one\{\q^\intercal \Delta \widehat{\btheta}_k(\X_i)>0\} -   \one\{\q^\intercal \Delta\btheta(\X_i)>0\}\right)~,\quad g \in \{r, b\}   
\end{align*}

We conclude that $ \mu_g \left( \tfrac{1}{\widehat{\mu}_g^\omega} - \tfrac{1}{{\mu}_g} \right)  =O_p(n^{-1/2})$ uniformly on $\q \in \bS^2$ by the Central Limit Theorem, Assumption \ref{asm:moments}, and the definition of the weights $\omega_i$. Therefore, it is sufficient to show that $B_{k,2}^g = o_p(1)$ uniformly on $\q \in \bS^2$ to conclude that $B_{k,2} = o_p(n^{-1/2})$ uniformly on $\q \in \bS^2$. 
Define $\widetilde{B}_{k,2}^g \equiv    \tfrac{1}{n_k} \sum_{i \in \mathcal{I}_k}    \one\{ |\q^\intercal \Delta\btheta(\X_i)| \le |\q^\intercal \Delta \widehat{\btheta}_k(\X_i) - \q^\intercal \Delta\btheta(\X_i)| \}$.
Note that Law of Iterated Expectations, Assumptions \ref{asm:moments} and \ref{asm:nuisance-function-structure}, and the definitions of the weights imply $|\bE[  B_{k,2}^g | \mathfrak{F}_k]| \le C  \widetilde{B}_{k,2}^g$, where $C$ is a constant based only on $c_1$, $c_2$, and $\epsilon$ (constants from Assumptions \ref{asm:moments} and \ref{asm:nuisance-function-structure}).  By similar arguments used in the proof of Lemma \ref{lemma:condition-A-weights} and by relying on Lemma \ref{lemma:delta-hat}, we conclude 
$$ \bE[ \widetilde{B}_{k,2}^g \one\{\widehat{E}_k \}| \mathfrak{F}_k] \lesssim\bE \left [  \widehat{\delta}_k(\X_{2,i}) | \mathfrak{F}_k  \right] = o_p(n^{-1/4})~,$$
which is sufficient to complete the proof of the lemma.
\end{proof}

\subsubsection{Auxiliary Results for Weighted Estimator of $\e^*$}\label{sec:weighted-estimator-e}

Let $\{ \omega_i: 1 \le i \le n\}$ be a random sample of positive weights that are  independent of the data such that $ \bE[\omega_i^2] \le C_\omega$. Define
\begin{align}
    \widehat{\e}^\omega(a^*;\widehat{\bLL}^\omega,\widehat{\boeta}) &\equiv \tfrac{1}{n} \sum_{i=1}^n \left( \tfrac{\omega_i}{\overline{\omega}}  \right)\xi_i(a^*;\widehat{\bLL}^\omega,\widehat{\boeta}) \label{eq:DML-estimator-weighted-e}\\
    \widehat{\e}^{\omega,*}(a^*; \bLL, \boeta) &\equiv \tfrac{1}{n}    \sum_{i =1}^n  \left( \tfrac{\omega_i}{\overline{\omega}}  \right) \xi_i^*(a^*;{\bLL},{\boeta}) ~,
\end{align}
where $\overline{\omega} = n^{-1} \sum_{i=1}^n \omega_i$, $\xi_i(a^*;\widehat{\bLL}^\omega,\widehat{\boeta})$,   $\widehat{\bLL}^\omega$ is defined as in \eqref{eq:defLdgA}-\eqref{eq:defLvector} with $\mu_g$ replaced by $\widehat{\mu}_g^\omega = n^{-1} \sum_{i=1}^n \left( \tfrac{\omega_i}{\overline{\omega}} \right)  \one\{G_i=g\}$, and  
$$\xi_i^*(a^*;\bLL,\boeta) = \xi_i(a^*;\bLL,\boeta) + \sum_{g \in \{r,b\}} \Gamma_g^{e}\left(1 -  \tfrac{\one\{G_i=g\}}{\mu_g} \right)$$
and $ \Gamma_g^{e} \equiv \bE\big[ \one\{G_i=g\} \left( \bLL_{0,i} + a(\X_i)\Delta\bLL_i\right)\tfrac{\Z_i}{\bpi(\X_i)} \big]$.

\begin{lemma}\label{lemma:asympt-equivalence-weights-e}
  \textit{Let Assumptions \ref{asm:moments}--\ref{asm:nuisance-funct-estimators} hold. Then, 
        \begin{equation}\label{eq:lemma-asympt-equivalence-weights-e} 
        \sup_{\q \in \bS^2} n^{1/2} \left| \widehat{\e}^\omega(a^*; \widehat{\bLL}^\omega, \widehat{\boeta})  - \widehat{\e}^{\omega,*}(a^*; \bLL, \boeta) \right| = o_p(1)
    \end{equation} 
    }
\end{lemma}

\begin{proof}  
    Let $\widehat{\boldsymbol{\balpha}}_k^e$ be the estimator of $\boldsymbol{\balpha}^e$ defined in \eqref{eq:alpha-e} using $\widehat{\boeta}_k$. 
    Define
    $$\widehat{\e}^{\omega,*}(a^*; \bLL, \widehat{\boeta}) \equiv \tfrac{1}{n}    \sum_{i =1}^n  \left( \tfrac{\omega_i}{\overline{\omega}}  \right) \xi_i^*(a^*;{\bLL},\widehat{\boeta}) ~.$$
    Our goal is to establish \eqref{eq:lemma-asympt-equivalence-weights-e}. We divide the proofs in two claims. 
    
    \noindent \underline{Claim 1}: $\sup_{\q \in \bS^2 } \sqrt{n} | \widehat{\e}^\omega(a^*; \widehat{\bLL}^\omega, \widehat{\boeta})  - \widehat{\e}^{\omega,*}(a^*; {\bLL}, \widehat{\boeta}) | = o_p(1)$. This claim is similar to Claim 2 in the proof of Lemma \ref{lemma:asympt-equivalence-weights} and also its proof; therefore, it is omitted.

\noindent \underline{Claim 2}: $\sup_{\q \in \bS^2 } \sqrt{n} \left| \widehat{\e}^{\omega,*}(a^*; {\bLL}, \widehat{\boeta}) -  \widehat{\e}^{\omega,*}(a^*; {\bLL},{\boeta}) \right| = o_p(1)$.

We  write ${\sqrt{n}} \left( \widehat{\e}^{\omega,*}(a^*; {\bLL}, \widehat{\boeta}) -  \widehat{\e}^{\omega,*}(a^*; {\bLL},{\boeta}) \right) = \tfrac{1}{\sqrt K} \sum_{k=1}^K (A_k + B_k+ C_k)$, where
\begin{align*}
    A_k &= \tfrac{1}{\sqrt{n_k}} \sum_{i \in \mathcal{I}_k} \left( \tfrac{\omega_i}{\overline{\omega}} \right)  \left( \tfrac{\Z_i}{\widehat{\bpi}_k(\X_i)} - \tfrac{\Z_i}{\bpi(\X_i)}\right) \left(  \bLL_{0,i}  +   a^*(\X_i) \Delta \bLL_{i}  - \boldsymbol{\balpha}^e(a^*;\X_i)   \right), \\
    B_k &= \tfrac{1}{\sqrt{n_k}} \sum_{i \in \mathcal{I}_k} \left( \tfrac{\omega_i}{\overline{\omega}} \right)  \left( \widehat{\boldsymbol{\balpha}}_k^e(a^*;\X) -\boldsymbol{\balpha}^e(a^*;\X)  \right) \left( 1 - \tfrac{\Z_i}{\bpi(\X_i)} \right), \\
    C_k &= \tfrac{1}{\sqrt{n_k}} \sum_{i \in \mathcal{I}_k} \left( \tfrac{\omega_i}{\overline{\omega}} \right)  \left( \widehat{\boldsymbol{\balpha}}_k^e(a^*;\X) -\boldsymbol{\balpha}^e(a^*;\X)  \right) \left( \tfrac{\Z_i}{{\bpi}(\X_i)} - \tfrac{\Z_i}{\widehat{\bpi}_k(\X_i)}\right).
\end{align*}
Recall that $\widehat{\boeta}_k\equiv(\Delta \widehat{\btheta}_k, \widehat{\bpi}_k, \widehat{\btheta}_{0,k})$ was estimated using all the data except the ones with indices in $\mathcal{I}_k$, and $\mathfrak{F}_k $ is the $\sigma$-algebra generated by all the data except the ones with indices on $\mathcal{I}_k$. The definition of the weights $\omega_i$, Assumptions \ref{asm:nuisance-function-structure} and \ref{asm:nuisance-funct-estimators}, the definition of $\boldsymbol{\balpha}^e(a^*;X)$ in \eqref{eq:alpha-e}, and algebra show $\bE[ A_k | \mathfrak{F}_k] = 0$ and $\sup_{q \in \bS^2} \bE[A_k^2 | \mathfrak{F}_k] = o(1)$. Similarly, we can show $\bE[ B_k | \mathfrak{F}_k] = 0$ and $\sup_{q \in \bS^2} \bE[B_k^2 | \mathfrak{F}_k] = o(1)$. By \citet[Lemma 6.1]{ChernozhukovDML}, we conclude $A_k$ and $B_k$ both are $o_p(1)$ uniformly on $\q \in \bS^2$. Finally, we conclude that $C_k = o_p(1)$ uniformly on $\q \in \bS^2$ by Cauchy-Schwartz, Assumptions \ref{asm:nuisance-function-structure} and \ref{asm:nuisance-funct-estimators}, and the definition of the weights $\omega_i$. 

The triangle inequality and Claims 1 and 2 complete the proof of the lemma.
\end{proof}

\begin{theorem}\label{thm:gaussian-e}
    \textit{Let Assumptions \ref{asm:moments}--\ref{asm:nuisance-funct-estimators} hold and let $\{\omega_i\}_{i=1}^n$ be random positive weights independent of the data such that $E[\omega_i] = 1$. Then, 
    $$ \sqrt{n} \left(  \widehat{\e}^\omega(a^*;\widehat{\bLL}^\omega,\widehat{\boeta}) - \e^* \right) = \mathbb{G}[\omega_i \{ \xi_i^*(a^*;\bLL,\boeta) - \e^*\}] + o_p(1)~,$$
    where  
    $\xi_i^*(a^*;\bLL,\boeta) = \xi_i(a^*;\bLL,\boeta) + \sum_{g \in \{r,b\}} \Gamma_g^{e}\left(1 -  \tfrac{\one\{G_i=g\}}{\mu_g} \right)$
and $ \Gamma_g^{e} = \bE\big[ \one\{G_i=g\} \left( \bLL_{0,i} + a(\X_i)\Delta\bLL_i\right)\tfrac{\Z_i}{\bpi(\X_i)} \big]$.
    } 
\end{theorem}

\begin{proof}
    We use Lemma \ref{lemma:asympt-equivalence-weights-e}, the Central Limit Theorem, and similar arguments as in the proof of Theorem \ref{thm:gaussian-h}; therefore, it is omitted.
\end{proof}

\begin{lemma}\label{lemma:bootstrap-equivalence-weights}
    Let Assumptions \ref{asm:moments}--\ref{asm:nuisance-funct-estimators} hold. Then,
    $$\left\|\sqrt{n} \left(  \widetilde{\bbeta}^\omega  -  \widehat{\bbeta}^\omega \right) - \sqrt{n} \left(  \widetilde{\bbeta}^{\omega,*}  -  \widehat{\bbeta}^{\omega,*} \right)   \right\|  = o_p(1)~,$$
    where $ \widehat{\bbeta}^{\omega,*}$  and $\widetilde{\bbeta}^{\omega,*}$ are defined in \eqref{eq:beta-hat-star}--\eqref{eq:beta-tilde-star}, respectively.
\end{lemma}
\begin{proof}
The norm $\| \cdot \|$ on $\ell^{\infty}(\bS^2) \times \mathbb{R}^3$ is the product norm, i.e., $\| (f,\e) \| = \|f\|_\infty + \|\e\|_2$, where $\| \cdot\|_2$ is the euclidean norm on $\mathbb{R}^3$. 
By the triangular inequality $\left\|\sqrt{n} \left(  \widetilde{\bbeta}^\omega  -  \widehat{\bbeta}^\omega \right) - \sqrt{n} \left(  \widetilde{\bbeta}^{\omega,*}  -  \widehat{\bbeta}^{\omega,*} \right)   \right\|$ 
is lower or equal than
\begin{equation}
       \left\|\sqrt{n} \left(  \widehat{\bbeta}^\omega - \widehat{\bbeta}^{\omega,*}      \right)   \right\| + \left\|\sqrt{n} \left(  \widetilde{\bbeta}^\omega -  \widetilde{\bbeta}^{\omega,*}    \right)   \right\|~. \label{eq:aux1-lemma-e5}
\end{equation} 
Therefore, it is sufficient to show that each term above is $o_p(1)$.

We use Lemmas \ref{lemma:asympt-equivalence-weights} and \ref{lemma:asympt-equivalence-weights-e} to conclude $ \left\|\sqrt{n} \left(  \widehat{\bbeta}^{\omega,*}  -  \widehat{\bbeta}^\omega  \right)   \right\|  = o_p(1)$. Similarly, we can conclude $ \left\|\sqrt{n} \left(  \widetilde{\bbeta}^{\omega,*}  -  \widetilde{\bbeta}^\omega   \right)   \right\|  = o_p(1)~$ by using the weights $W_i \omega_i^\h$ and~$W_i \omega_i^\e$ in Lemmas \ref{lemma:asympt-equivalence-weights} and  \ref{lemma:asympt-equivalence-weights-e}, respectively. This completes the proof.
\end{proof}

\clearpage
\section*{Online Appendix}
\label{online_appendix}
\setcounter{section}{0}
\setcounter{equation}{0}
\setcounter{lemma}{0}
\renewcommand{\thesection}{OA.\arabic{section}}
\renewcommand{\thesubsection}{OA.\arabic{section}.\arabic{subsection}}
\renewcommand{\theequation}{OA.\arabic{equation}}

\section{Additional Results for Section \ref{sec:setup}}
\begin{lemma}\label{lem:grad-supp-func}
\textit{Let the assumptions of Theorem \ref{thm:support_E} hold. Then $$\nabla_\q\h_{\E}(\q)=
    \bEs\bigl[\bLL_0+\Delta\bLL\one\{\q^\intercal \Delta\btheta(\X)>0\}\bigr].$$}
\end{lemma}
\begin{proof}[Proof of Lemma \ref{lem:grad-supp-func}]
Fix $\q\in \bS^2$ and let $v\in\R^3$ be arbitrary. For $t\in\R$ such that $\q+t v\neq 0$ and $u_+ \equiv \max\{u,0\}$, by the support-function expression in \eqref{eq:support},
\begin{align*}
\tfrac{\h_{\E}(\q+t v)-\h_{\E}(\q)}{t}=\bEs\left[v^\intercal \btheta_0(\X)+\tfrac{
\left((\q+t v)^\intercal \Delta\btheta(\X)\right)_+-
\left(\q^\intercal \Delta\btheta(\X)\right)_+}{t}
\right].
\end{align*}
Hence it is enough to study the difference quotient of the map $u\mapsto u_+$.

For each realization $\X=x$ such that $\q^\intercal \Delta\btheta(x)\neq 0$, the function $t\mapsto \left((\q+t v)^\intercal \Delta\btheta(x)\right)_+$ is differentiable at $t=0$, with derivative $v^\intercal \Delta\btheta(x)\one\{\q^\intercal \Delta\btheta(x)>0\}$.
Therefore, pointwise on the event $\{\q^\intercal \Delta\btheta(\X)\neq 0\}$, as $t\ra 0$,
\begin{align}
    \tfrac{\left((\q+t v)^\intercal \Delta\btheta(\X)\right)_+-\left(\q^\intercal \Delta\btheta(\X)\right)_+}{t}\ra v^\intercal \Delta\btheta(\X)\one\{\q^\intercal \Delta\btheta(\X)>0\}.\label{eq:derive_ss_pointwise_conv}
\end{align}
Assumption~\ref{asm:margin} implies $\bP^*\bigl(\q^\intercal \Delta\btheta(\X)=0\bigr)=0$ for every $\q\in \bS^2$, because for every $\delta>0$, 
\begin{align*}
\bP^*\left(\q^\intercal \Delta\btheta(\X)=0\right)
\le
\sup_{\tilde \q\in \bS^2}\bP^*\bigl(|\tilde \q^\intercal \Delta\btheta(\X)|\le \delta\bigr)
\lesssim\delta^m
\end{align*}
and letting $\delta\downarrow 0$ yields the claim. Thus, the  convergence in \eqref{eq:derive_ss_pointwise_conv} holds $\bP^*$-almost surely. Next we obtain an integrable dominating function. Since $u\mapsto u_+$ is $1$-Lipschitz,
\begin{align*}
\left|\tfrac{\left((\q+t v)^\intercal \Delta\btheta(\X)\right)_+ - \left(\q^\intercal \Delta\btheta(\X)\right)_+}{t}\right|
\le |v^\intercal \Delta\btheta(\X)|
\end{align*}
for all $t\neq 0$. Therefore,
\begin{align*}
\left|v^\intercal \btheta_0(\X)+\tfrac{\left((\q+t v)^\intercal \Delta\btheta(\X)\right)_+ - \left(\q^\intercal \Delta\btheta(\X)\right)_+}{t}\right|
\le |v^\intercal \btheta_0(\X)| + |v^\intercal \Delta\btheta(\X)|
\end{align*}
By Assumption~\ref{asm:moments}, the random vectors $\btheta_0(\X)$ and $\Delta\btheta(\X)$ are integrable, so the right-hand side is integrable.
Hence, by dominated convergence,
\begin{align*}
\lim_{t\ra 0}\tfrac{\h_{\E}(\q+t v)-\h_{\E}(\q)}{t}
&=
v^\intercal\bEs\left[\btheta_0(\X)+\Delta\btheta(\X)\one\{\q^\intercal \Delta\btheta(\X)>0\}\right].
\end{align*}
Thus the directional derivative of $\h_{\E}$ at $\q$ in direction $v$ exists and is linear in $v$. It follows that $\h_{\E}$ is differentiable at $\q$, with gradient
\begin{align*}
\nabla_{\q}\h_{\E}(\q)=\bEs\left[\btheta_0(\X)+\Delta\btheta(\X)\one\{\q^\intercal \Delta\btheta(\X)>0\}\right].
\end{align*}
Lemma \ref{lem:grad-supp-func} follows from applying the law of iterated expectations again.
\end{proof}

\section{Additional Proofs for Appendix \ref{app:auxiliary:1}}
\begin{proof}[Prof of Lemma \ref{prop:compact_convex}]
Recall that $\mathcal A \equiv\{a:\sX\to[0,1]\ \text{Borel-measurable}\}$.

\noindent\textbf{Part (i).} Each coordinate map $a\mapsto e_g^\iota(a)$ is affine in $a$ by \eqref{eq:eg_iota}. Since $\mathcal A$ is convex and $(e_r^A(a),e_b^A(a),e_r^F(a)-e_b^F(a))$ is an affine image of $a$, its range $\E$ is convex.\smallskip

\noindent\textbf{Part (ii).} Define
\begin{align*}
K\equiv\{u\in L^\infty(\Omega,\mathfrak F,\bP):
u \text{ is }\sigma(\X)\text{-measurable and }0\le u\le 1 \text{ a.s.}\}.
\end{align*}
As $(\sX,\mathcal B(\sX))$ is standard Borel, the Doob--Dynkin lemma \citep[Prop. 3, Ch. 1]{RaoSwift2006} yields $K=\{a(\X):a\in\mathcal A\}$ up to a.s. equality.
Hence, $\E=\bE[\btheta_0(\X)]+T(K)$, for
\begin{align}
T(u)\equiv\bE[u\Delta\btheta(\X)]\in\R^3.\label{eq:def_T}
\end{align}
We first show that $K$ is weak-* compact in $L^\infty$, where $L^\infty$ is equipped with the topology $\sigma(L^\infty,L^1)$.
Since $0\le u\le 1$ a.s. for every $u\in K$, we have $\|u\|_\infty\le 1$, so $K$ is contained in the closed unit ball of $L^\infty$. 
Since $L^\infty=(L^1)^*$, the Banach-Alaoglu theorem yields that the closed unit ball of $L^\infty$ is compact for the weak-* topology $\sigma(L^\infty, L^1)$.

To prove weak-* compactness of $K$, it suffices to show that $K$ is weak-* closed, since $K$ is contained in the weak-* compact unit ball of $L^\infty$. Accordingly, let $u_\alpha\in K$ and suppose $u_\alpha\to u$ weak-* in $L^\infty$. We show that $u\in K$.
First, $u$ is $\sigma(\X)$-measurable. Indeed, for every $g\in L^1(\Omega,\mathfrak F,\bP)$, $\bE\left[u_\alpha\bigl(g-\bE[g| \sigma(\X)]\bigr)\right]=0$, because $u_\alpha$ is $\sigma(\X)$-measurable. Passing to the limit using weak-* convergence gives $\bE\left[u\left(g-\bE[g| \sigma(\X)]\right)\right]=0$, for all $g\in L^1$.
Taking $g=\one_A$ for $A\in\mathfrak F$, we obtain
\begin{align*}
\int_A ud\bP=\int_A \bE[u|\sigma(\X)]d\bP,\quad \forall A\in\mathfrak F,
\end{align*}
hence $u=\bE[u|\sigma(\X)]$ a.s. Therefore $u$ is $\sigma(\X)$-measurable.

Next, the bounds $0\le u\le 1$ a.s. are preserved. To show $u\ge 0$ a.s., suppose otherwise.
Then $\bP(u<0)>0$, so for some $\varepsilon>0$, $A_\varepsilon\equiv\{u\le -\varepsilon\}$ has positive probability. 
Taking $g=\one_{A_\varepsilon}\in L^1$, weak-* convergence yields $\bE[u_\alpha g]\to \bE[ug]$.
Since $u_\alpha\ge 0$ a.s., we have $\bE[u_\alpha g]\ge 0$ for all $\alpha$, whereas $\bE[ug]=\bE[u\one_{A_\varepsilon}]
\le -\varepsilon\,\bP(A_\varepsilon)<0$, a contradiction. Thus $u\ge 0$ a.s. Applying the same argument to $1-u_\alpha$ shows
$u\le 1$ a.s.
Therefore $u\in K$, so $K$ is weak-* closed. Hence $K$ is weak-* compact.

Now consider the map $T:K\to\R^3$ in \eqref{eq:def_T}.
This map is well-defined because $u\in L^\infty$ with $|u|\le 1$ a.s. and each component of $\Delta\btheta(\X)$ is in $L^1$ by Assumption~\ref{asm:moments}, so each component of $u\Delta\btheta(\X)$ is in $L^1$. 
Moreover, for each component $j=1,2,3$, the map from $u$ to the component of $\bE[u\Delta\btheta(\X)]$ is weak-* continuous on $L^\infty$, since each component of $\Delta\btheta(\X)$ is in $L^1$. 
Therefore $T:K\to\R^3$ is continuous when $K$ is endowed with the weak-* topology and $\R^3$ with its usual Euclidean topology.
Since $K$ is weak-* compact and $T$ is continuous, $T(K)$ is compact in $\R^3$. Therefore, $\E=\bE[\btheta_0(\X)]+T(K)$ is compact.\smallskip

\noindent\textbf{Part (iii).}
Let $\mathcal A_{\mathrm{disp}}\equiv\{\varepsilon:\sX\to\R\ \text{measurable}:\ |\varepsilon(\X)|\le 1/2\ \text{$P_X$-a.s.}\}=\{\varepsilon:\sX\to\R\ \text{measurable}:\ \|\varepsilon\|_\infty\le 1/2\}$ where
$\|\varepsilon\|_\infty\equiv\mathrm{ess}\sup_{x\sim P_X}|\varepsilon(x)|$. 
For any algorithm $a\in\mathcal A$, let $\varepsilon(\X)\equiv a(\X)-\tfrac12$ be the corresponding displacement in $\mathcal A_\mathrm{disp}$. Let $e^0=\bE\bigl[\btheta_0(\X)]+\tfrac12\bE[\Delta\btheta(\X)]$.
Then for any $a\in\mathcal A$ we have $e(a)-e^0=\bE[\varepsilon(\X)\Delta\btheta(\X)]$.
Hence, we can write $\E=e^0+\E^0$, where $\E^0$ is the set of feasible displacements:
\begin{equation*}
\E^0\equiv\left\{\bE[\varepsilon(\X)\Delta\btheta(\X)],~ \varepsilon:\sX\to[-\tfrac12,\tfrac12]\ \text{measurable}\right\},
\end{equation*}
and the equality holds because every measurable $a:\sX\to[0,1]$ corresponds to a unique $\varepsilon\in[-\tfrac12,\tfrac12]$ and vice versa.
We next show that $\E^0$ contains a Euclidean ball around $0$, which implies $\mathcal E=e^0+\E^0$ contains a ball around $e^0$.
Fix $q\in\R^3$ with $\|q\|=1$. Consider the support function of $\E^0$ in direction $q$:
\begin{align*}
    \h_{\E^0}(q)\equiv\sup_{e\in \E^0}\q^\intercal e 
=\sup_{\|\varepsilon\|_\infty\le 1/2}\bE\left[\varepsilon(\X)\q^\intercal\Delta\btheta(\X)\right].
\end{align*}
Pointwise maximization yields the maximizer
$\varepsilon_q(\X)=\tfrac12\,\mathrm{sign}\big(\q^\intercal\Delta\btheta(\X)\big)$ (with any value in $[-\tfrac12,\tfrac12]$ on the null set where $\q^\intercal\Delta\btheta(\X)=0$),
so
\begin{equation}\label{eq:support_abs}
\h_{\E^0}(\q)=\tfrac12\bE\left[|\q^\intercal\Delta\btheta(\X)|\right].
\end{equation}
By Assumption \ref{asm:margin}, there exists a constant $C>0$ such that for all $\delta>0$, $\bP(|\q^\intercal\Delta\btheta(\X)|\le \delta)\le C\delta^m$, 
uniformly in $\q\in\bS^2\equiv\{\q\in\R^3:\Vert q\Vert=1\}$.
Choose $\delta_0\equiv (1/(2C))^{1/m}>0$. Then for every $\q\in\bS^2$, $\bP(|\q^\intercal\Delta\btheta(\X)|>\delta_0)\ge 1-C\delta_0^m = \tfrac12$, and hence $\bE\left[|\q^\intercal\Delta\btheta(\X)|\right]\ge\delta_0\bP(|\q^\intercal\Delta\btheta(\X)|>\delta_0)\ge\tfrac{\delta_0}{2}$.
Combining with \eqref{eq:support_abs} yields the uniform bound
\begin{equation}\label{eq:uniform_support_lb}
\h_{\E^0}(\q)\ge\tfrac{\delta_0}{4}
\qquad\forall \q\in\bS^2.
\end{equation}
Hence, \eqref{eq:uniform_support_lb} implies that $\E^0$ contains the Euclidean ball $\bB(0,r)$ with radius $r=\delta_0/4$ \citep[see, e.g.,]
[Section 1.7]{sch93}. 
Therefore, $e^0+\bB(0,r)\subseteq \E$
so $\E$ has a nonempty interior in $\R^3$, and hence $\dim(\E)=3$.
\end{proof}

\begin{proof}[Proof of Corollary \ref{corr:one-sided-geometry}]
If $\q\in\widetilde{\bS}^{2}_{+}$, then for
$\e=\ss_\E(-\q)$ we have $\e_3\geq0$ and $\q_3\geq0$, so
$\q_3\e_3\geq0$. Lemma~\ref{lemma:PF-F-support-map} yields $\e\in\sF$.

Conversely, let $\e=\ss_\E(-\q)\in\sF$ for some $\q\in\widetilde{\bS}^{2}$, so that $\q_3\e_3\geq0$. 
If $\q_3\geq0$, then $\q\in\widetilde{\bS}^{2}_{+}$. If $\q_3<0$, then $\e_3=0$ and for any $\eta\in\E$,
\begin{align*}
\q_1\e_1+\q_2\e_2
=
\q^\intercal\e
\leq
\q^\intercal\eta
=
\q_1\eta_1+\q_2\eta_2+\q_3\eta_3
\leq
\q_1\eta_1+\q_2\eta_2,
\end{align*}
where the last inequality uses $\q_3<0$ and $\eta_3\geq0$. The case
$\q_1=\q_2=0$ would force $\eta_3=0$ for all $\eta\in\E$, contradicting
that $\E$ has nonempty interior. Hence $\q_1+\q_2>0$. After normalizing
$[\q_1,\q_2,0]^\intercal$, we obtain a vector in
$\widetilde{\bS}^{2}_{+}$ exposing the same point $\e$. This proves \eqref{eq:sF_representation_tildeS+}.

The identities for $\sIM_{12}$, $\sIM_{13}$, and $\sIM_{23}$ follow
because positive rescaling of a normal vector does not change the minimizers. For example,  $\{\q\in\widetilde{\bS}^{2}_{+}:\q_3=0\}$ consists exactly of positive normalizations of $[\alpha,1-\alpha,0]^\intercal$ for $\alpha\in[0,1]$. 
Thus its support image is $\sIM_{12}$, which equals $\sPF$ by Lemma~\ref{lemma:PF-F-support-map}. 
The other two edges follow by the same argument.

Under Assumption~\ref{asm:no-kink}, the support map is injective. 
Hence the map $T:\widetilde{\bS}^{2}_{+}\to\sF$, $T(\q)=\ss_\E(-\q)$, is injective. It is also continuous by Berge's maximum theorem, and it is
surjective by the representation in \eqref{eq:sF_representation_tildeS+}.
Therefore $T$ is a continuous bijection from the compact space $\widetilde{\bS}^{2}_{+}$ onto the Hausdorff space $\sF\subset\R^3$, and hence $T$ is a homeomorphism. 
Thus, the images of the three boundary edges of
$\widetilde{\bS}^{2}_{+}$ are the relative boundary edge curves of $\sF$.
In particular, $\sIM_{12}=T\left(\{\q\in\widetilde{\bS}^{2}_{+}:\q_3=0\}\right)=\sPF$.
Since $\{\q\in\widetilde{\bS}^{2}_{+}:\q_3=0\}$ is a proper subset of $\widetilde{\bS}^{2}_{+}$, injectivity of $T$ implies $\sIM_{12}=\sPF\subsetneq\sF$.
\end{proof}

\begin{proof}[Proof of Corollary \ref{corr:crossing-one-sided-pareto}]
It is enough to prove the result when
$\sPF\subset\{\e\in\R^3:\e_3>0\}$; the other case follows by applying the
reflection $(\e_1,\e_2,\e_3)\mapsto(\e_1,\e_2,-\e_3)$.

Since $\E$ is convex and contains points with positive and negative third coordinate, $\E\cap\{\e\in\R^3:\e_3=0\}\neq\emptyset$.
Hence, $\pi\left(\E\cap\{\e_3=0\}\right)$ is a nonempty compact subset of $\E^A$. 
Choose $c^0\in\arg\min_{c\in\pi(\E\cap\{\e_3=0\})}(c_1+c_2)$.
Then $c^0\in\E^A$ and $d(c^0)=0$ for $d(\cdot)$ defined in \eqref{eq:min_disparity}, because $c^0$ has a feasible lift with third coordinate equal to zero.

We first show that $c^0\in\sF^A$. Suppose not. Then there exists $c'\in\E^A$ such that
\begin{align*}
c'_1\leq c^0_1,\qquad c'_2\leq c^0_2,\qquad d(c')\leq d(c^0)=0,
\end{align*}
with at least one strict inequality. 
Since $d(c')\geq0$, we have $d(c')=0$.
Therefore $c'\in\pi(\E\cap\{\e_3=0\})$. 
Moreover, the strict inequality cannot occur only in the fairness coordinate, because $d(c')=d(c^0)=0$. 
Hence at least one of $c'_1<c^0_1$ or $c'_2<c^0_2$ holds, so $c'_1+c'_2<c^0_1+c^0_2$, contradicting the definition of $c^0$. 
Thus $c^0\in\sF^A$.

Next, $c^0\notin\sPF^A$. 
Indeed, since $c^0\in\pi(\E\cap\{\e_3=0\})$, the point $\e^0\equiv(c^0_1,c^0_2,0)$ belongs to $\E$. 
If $c^0\in\sPF^A$, then $\e^0\in\sPF$: otherwise there
would exist $\eta\in\E$ with $\eta_1\leq\e^0_1$, $\eta_2\leq\e^0_2$, and at least one strict inequality, and then $\pi(\eta)$ would Pareto-dominate $c^0$ in $\E^A$. 
This contradicts $c^0\in\sPF^A$. 
Therefore $\e^0\in\sPF$, contradicting the maintained case
$\sPF\subset\{\e_3>0\}$. 
Hence $c^0\notin\sPF^A$.

We have shown that $c^0\in\sF^A\setminus\sPF^A$.
It remains to check that $\sPF^A\subseteq\sF^A$. If
$c\in\sPF^A$ but $c\notin\sF^A$, then some $c'\in\E^A$ satisfies $c'_1\leq c_1$, $c'_2\leq c_2$, $d(c')\leq d(c)$, with at least one strict inequality. 
The strict inequality cannot occur only in the fairness coordinate: if $c'_1=c_1$ and $c'_2=c_2$, then $c'=c$, hence $d(c')=d(c)$. 
Thus at least one of $c'_1<c_1$ or $c'_2<c_2$ holds, contradicting $c\in\sPF^A$. 
Therefore $\sPF^A\subseteq\sF^A$.
Together with $c^0\in\sF^A\setminus\sPF^A$, this gives
\begin{align*}
\sPF^A\subsetneq\sF^A.
\end{align*}

Finally, because $c^0\in\sF^A$, $d(c^0)=0$, and
$\e^0=(c^0_1,c^0_2,0)\in\E$, Proposition~\ref{prop:FequivLLMO_proj} implies $\e^0\in\sF$.
But $\e^0\notin\sPF$, since $\e^0_3=0$ while
$\sPF\subset\{\e_3>0\}$. 
Since Lemma~\ref{lemma:PF-F-support-map} gives
$\sPF\subseteq\sF$, we obtain $\sPF\subsetneq\sF$.
\end{proof}

\section{Additional Proofs for Appendix \ref{app:auxiliary:2}}
\begin{proof}[Proof of Lemma \ref{lem:continuous_h_frakAq}]
By \citet[Thm. 2.1.35]{mol17}, the support function of $\cM_q$ is given by
\begin{align*}
    \h_{\cM_q}(v)=\bE\!\left[\sup_{y\in K_q(X)}v^\intercal y\right].
\end{align*}
Using the definition of $K_q(X),\,\kappa_0(\v,q;X),\,\kappa_1(\v,q;X)$, we obtain~\eqref{eq:sup_fun_mathfrakAq}.

To prove \eqref{eq:cont:h_frakAq}, we first use that  $K_q(x)$ and $K_p(x)$ are compact sets and  $y \to v^\intercal y$ is $\| v \|$-Lipschitz to conclude that, for any $p,q,v\in\bS^2$,
$$ \left| \sup_{y\in K_q(X)}v^\intercal y- \sup_{y\in K_p(X)}v^\intercal y \right| \le \| v \| d_H(K_q(x),K_p(x))~,$$
where $d_H(K_q(x),K_p(x))$ denotes the Hausdorff distance between $K_q(x)$ and $K_p(x)$. This implies that $\sup_{v\in\bS^2}|\h_{\cM_p}(v)-\h_{\cM_q}(v)| \le \bE[d_H(K_q(X),K_p(X))]$; therefore, to show \eqref{eq:cont:h_frakAq}, it is sufficient to show
$ \bE[d_H(K_q(X),K_p(X))] \lesssim\|p-q\|^{m/(m+1)}~.$ If $\|p-q\| > 1$, then the inequality holds by using that $d_H(K_q(X),K_p(X))] \le 2H_0$. 

Suppose $\|p-q\| \le 1$. Take $\delta \equiv (1+2M_0) \|p-q\|^{1/(m+1)} > 2 M_0 \|p-q\|$. Define the event $G_\delta=
\bigcap_{j=0,1}
\{|q^\intercal\Delta\theta(X;\lambda_j)|>\delta\}.$ We show in claim 1 below that $P(G_\delta) \lesssim \delta^m$ and in claim 2 below that conditional on $G_\delta$, we have $d_H(K_p(X),K_q(X))
\le\frac{4H_0 M_0 \|p-q\|}{\delta}$. Therefore, 
\begin{align*}
     \bE[d_H(K_q(X),K_p(X))] &=  \bE[d_H(K_q(X),K_p(X)) \one\{G_\delta^c\}] + \bE[d_H(K_q(X),K_p(X)) \one\{G_\delta\}] \\
     &\lesssim 2H_0   \delta^m + \frac{4H_0 M_0 \|p-q\|}{\delta} \propto \| p- q\|^{m/(m+1)}  ~,
\end{align*}
where we used that $\delta =  (1+2M_0) \|p-q\|^{1/(m+1)} $. This completes the proof of \eqref{eq:cont:h_frakAq}. 

\noindent \textbf{Claim 1:} $P(G_\delta^c) \lesssim \delta^m$. To prove this, we note that $G_\delta^c = \bigcup_{j=0,1}
\{|q^\intercal\Delta\theta(X;\lambda_j)| \le \delta\}$; therefore, Assumption~\ref{asm:endpoint-margin} implies the claim.

\noindent \textbf{Claim 2:} $d_H(K_p(X),K_q(X))
\le\frac{4H_0 M_0 \|p-q\|}{\delta}$ conditional on $G_\delta$. To prove this claim, we fix $x$ such that $\bigcap_{j=0,1}
\{|q^\intercal\Delta\theta(x;\lambda_j)|>\delta\}$ holds. Then, for any $y \in K_q(x)$ we will construct $y' \in K_p(x)$ such that $\|y-y'\| \le \frac{4H_0 M_0 \|p-q\|}{\delta}$, which is sufficient to prove the claim. 

For $x \in G_\delta$, we have $
|(p-q)^\intercal(\Delta A(x)-2u)|\le M_0 \|p - q\|$ for any $u \in U(x)$. Define $r=2M_0 \|p-q\|/\delta$ and $b= B(x) \lambda_1$.

Suppose $y \in K_q^0(x)$. Then, there is $u \in U(x)$ such that $y=A_0(x)+u$ and $q^\intercal(\Delta A(x) - 2u) \ge 0$. Define $y'= A_0(x) + u'$, where $u'=(1-r)u$. It follows that $u' \in U(x)$ since $U(x)$ is convex and includes the zero vector. We now show that $y' \in K_p^0(x)$, which is equivalent to prove that $p^\intercal(\Delta A(x) - 2u') \ge 0$. To see this, 
\begin{align*}
    p^\intercal(\Delta A(x) - 2u') &= (1-r)p^\intercal (\Delta A(x) -2u)+r p^\intercal \Delta A(x) \\
    &\ge -(1-r)M_0 \|p-q\| +r\delta/2 \ge 0~,
\end{align*}
where the first inequality holds since $p^\intercal(\Delta A(x) - 2u) \ge -M_0 \|p-q\|$ and  $p^\intercal \Delta A(x) > \delta/2$ (which holds because $q^\intercal \Delta A(x) > \delta$ and $x \in G_\delta$). Finally, $y-y' = u-u'= r u$; therefore, $\|y-y'\| \le r \|u\| \le r ( \| A_0(x) + u\| +  \|A_0(x) +0\|) \le \tfrac{4 M_0 \|p-q\| H_0}{\delta}$.

Suppose $y \in K_q^1(x)$. Then, there is $u \in U(x)$ such that $y=A_1(x) - u$ and $q^\intercal (\Delta A(x) - 2u) \le 0$. Define  $y'= A_1(x) - u'$, where $u'=(1-r)u + r b$. It follows that $u' \in U(x)$ since $U(x)$ is convex and $b \in U(x)$. We now show that $y' \in K_p^1(x)$, which is equivalent to prove that $p^\intercal(\Delta A(x) - 2u') \le 0$. To see this,
\begin{align*}
    p^\intercal(\Delta A(x) - 2u') &= (1-r)p^\intercal (\Delta A(x) -2u)+r p^\intercal (\Delta A(x) - 2b)\\
    &\le (1-r)M_0 \|p-q\| - r\delta/2 \le 0~,
\end{align*}
where the first inequality holds since $p^\intercal(\Delta A(x) - 2u) \le M_0 \|p-q\|$ and  $p^\intercal ( \Delta A(x) - 2b) < -\delta/2$ (which holds because $q^\intercal (\Delta A(x) - 2b) < -\delta$ and $x \in G_\delta$). Finally, $y-y' =   r(b - u)$; therefore, $\|y-y'\| \le r \|b - u\| \le r ( \| A_1(x) - u\| +  \|A_1(x) - b\|) \le \tfrac{4 M_0 \|p-q\| H_0}{\delta}$.
\end{proof}

\begin{proof}[Proof of Lemma \ref{lem:lambda_t}]
    By definition of $\e \in \cM_q$, we have  $\e = \bE[ a(\X) \Delta \btheta(\X;\lambda) + \btheta_0(\X;\lambda)]$, 
    \begin{align}\label{eq:aux1_lemma_lambda_t}
        q^\intercal\Delta\btheta(X;\lambda)\ge0\ \text{on }\{a(X)=0\},
\quad
q^\intercal\Delta\btheta(X;\lambda)\le0\ \text{on }\{a(X)=1\},
\quad \bP_\X\text{-a.s.}
    \end{align}
 Recall $\lambda_0=(0,0)^\intercal$ and  $\lambda_1=(1,1)^\intercal$. Since  $q_1,q_2 \ge 0$, \(B(X)\) has nonnegative entries, and $\Delta\btheta(X;\lambda)=\Delta A(X)-2B(X)\lambda(X)$, we obtain
$$q^\intercal\Delta\theta(X;\lambda_0)
\ge q^\intercal\Delta\theta(X;\lambda)
\ge q^\intercal\Delta\theta(X;\lambda_1)~.$$
Assumption~\ref{asm:endpoint-margin} and \eqref{eq:aux1_lemma_lambda_t} imply $q^\intercal\Delta\btheta(X;\lambda_0)>0$ on $\{a(X)=0\}$ and $q^\intercal\Delta\btheta(X;\lambda_1)<0$ on $\{a(X)=1\}$, $ P_X\text{-a.s.}$  
Fix $t\in(0,1)$. Since $\lambda(X;a)=\lambda_{a(X)}$ and $\Delta\btheta$ is affine in $\lambda$,
\begin{align*}
\Delta\btheta(X;\lambda_t) =(1-t)\Delta\btheta(X;\lambda) +t\Delta\btheta(X;\lambda_{a(X)})~.    
\end{align*} 
It follows that $q^\intercal\Delta\btheta(X;\lambda_t)>0$ on $\{a(X)=0\}$ and $q^\intercal\Delta\theta(X;\lambda_t)<0$ on $\{a(X)=1\}$, $\bP_X\text{-a.s.}$ 
Therefore, for any  $b\in\mathcal{A}(\sX)$,
\begin{align*}
q^\intercal\e(b;\lambda_t)-q^\intercal\e(a;\lambda_t) &=\bE[(b(\X)-a(\X))\q^\intercal\Delta\btheta(\X;\lambda_t)]\\
&=\bE[|b(\X)-a(\X)|\,|\q^\intercal\Delta\btheta(\X;\lambda_t)|]
\ge0,
\end{align*}
with equality only if $b=a$,  $\bP_X$-a.s. Hence,
$\e_t\equiv \e(a;\lambda_t)$ is the unique minimizer of $q^\intercal\eta$ over $\eta\in \E(\lambda_t)$.
Finally, the fairness coordinate does not depend on the completion, so
\(q_3\e_{t,3}=q_3\e_3\ge0\).
Suppose that some \(\eta\in\mathcal E(\lambda_t)\) FA-dominates \(\varepsilon_t\). Then
\[
q^\intercal\eta
\le q_1\e_{t,1}+q_2\e_{t,2}
   +|q_3|\,|\e_{t,3}|
=q^\intercal\e_t.
\]
Uniqueness of the minimizer implies $\eta=\e_t$, contradicting strict improvement in at least one coordinate. This implies that $\e(a;\lambda_t)\in \sF(\lambda_t)$. Finally, define $V=\Delta\btheta(\X;\lambda)$ and $W=\Delta\btheta(\X;\lambda(\cdot;a))$. Then, Lemma~\ref{lem:PI-segment-margin} implies that $\lambda_t$ satisfies a margin conditions similar to Assumption~\ref{asm:margin} for Lebesgue-almost every $t \in (0,1)$. 
\end{proof}

\begin{proof}[Proof of Lemma \ref{lem:PI-segment-margin}]
Choose $L<\infty$ such that $\|V\|,\|W\|\le L$ almost surely,
and let $(V_i,W_i)$, $i=1,2,3$, be independent copies of $(V,W)$.
Define
$$ Z_t^{(i)}=(1-t)V_i+tW_i,
\qquad D(t)=\det(Z_t^{(1)},Z_t^{(2)},Z_t^{(3)}).  $$
Thus $D(t)$ is a random polynomial of degree at most three.
Set $s=m/8$, so that $0<s<m$ and $s<1/3$.
By Claims~2--3 below and Tonelli's theorem, we obtain
\begin{align*}
    \int_0^1\bE|D(t)|^{-s}\,dt \le c_s\bE|D(1)|^{-s}<\infty.
\end{align*} 
Consequently, $A_t\equiv \mathbb E|D(t)|^{-s}<\infty$ for almost every $t \in (0,1)$.
For any such $t$, Claim~1 and Markov's inequality imply
$$\bP(|v^\intercal Z_t|\le \delta)^3\le(3L^2 \delta)^s A_t.$$
Taking cube roots and the supremum over $v\in\mathbb S^2$ yields
the conclusion with $C_t=(3L^2)^{s/3}A_t^{1/3}$, since $s/3=m/24$. The exceptional set of $t$ does not depend on $v$.  

\noindent \textbf{Claim 1:} $\bP(|v^\intercal Z_t|\le \delta)^3 \le\mathbb P\bigl(|D(t)|^{-s}\ge(3L^2\delta)^{-s}\bigr)$ for any $t\in(0,1)$, $v\in\mathbb S^2$, and $\delta>0$, where we used the convention $0^{-s}=+\infty$. To see this, choose orthonormal coordinates with first axis $v$. The expansion of the determinant along its first row and Hadamard's inequality for the corresponding minors imply
$$|D(t)|\le\sum_{i=1}^3|v^\intercal Z_t^{(i)}|
\prod_{j\ne i}\|Z_t^{(j)}\|.$$
Hence $|v^\intercal Z_t^{(i)}|\le \delta$ for all $i$ implies $|D(t)|\le3L^2\delta$.
Independence of the three copies gives
\begin{align*}
    \bP(|v^\intercal Z_t|\le \delta)^3
=\bP\left(\bigcap_{i=1}^3\{|v^\intercal Z_t^{(i)}|\le \delta\}\right)
\le\bP(|D(t)|\le3L^2\delta),
\end{align*} 
which proves the claim.

\noindent \textbf{Claim 2:} $\int_0^1|D(t)|^{-s}\,dt\le c_s|D(1)|^{-s}$ for a positive finite constant $c_s$, depending only on $s$, when $D(1)\ne0$. By Remez’s inequality \citep[Theorem 5.1.1]{borwein2012polynomials}, rescaled to  $[0,1]$, any polynomial $P$ of degree at most three and measurable set $E\subseteq[0,1]$ of positive measure satisfy 
$$\sup_{t\in[0,1]}|P(t)| \le \left(\frac{4}{\operatorname{Leb}(E)}\right)^3 \sup_{t\in E}|P(t)|.$$
By taking $E= \{t\in(0,1):|P(t)|\le \delta\}$ and using $\sup_{[0,1]}|P|\ge|P(1)|$, we obtain 
\begin{align*}
    \operatorname{Leb}\{t\in(0,1):|P(t)|\le \delta\}
\le\min\left\{1,4\left(\frac{\delta}{|P(1)|}\right)^{1/3}\right\},
\qquad \delta>0.
\end{align*}
Define $f(t)=|P(t)|/|P(1)|$. Integrating the sublevel bound gives
\begin{align*}
\int_0^1 f(t)^{-s}\,dt
&\le1+s\int_0^1 \delta^{-s-1}
\operatorname{Leb}\{t\in(0,1):f(t)\le \delta\}\,d\delta\\
&\le1+cs\int_0^1\delta^{1/3-s-1}\,d\delta
=1+\frac{cs}{1/3-s} =: c_s<\infty.
\end{align*} 
Multiplying by $|P(1)|^{-s}$ and taking $P=D$ proves the claim. 

\noindent \textbf{Claim 3:}  $\mathbb E|D(1)|^{-s}<\infty$. Note that \eqref{eq:condition_PI-segment-margin} and $s<m$ imply
$$K_s\equiv \sup_{v\in \bS^2} \bE|v^\intercal W|^{-s}
\le1+Cs\int_0^1 \delta^{m-s-1}\,d \delta
=1+\frac{Cs}{m-s}<\infty.$$
Gram--Schmidt implies $|D(1)|=R_1R_2R_3,$
where $R_1=\|W_1\|$, $R_2=\operatorname{dist}(W_2,\operatorname{span}\{W_1\})$, and $R_3=\operatorname{dist}(W_3,\operatorname{span}\{W_1,W_2\})$. 
The first factor dominates $|e_1^\intercal W_1|$.
For the second factor, choose a measurable unit vector perpendicular to $W_1$;
for the third, choose a measurable unit vector perpendicular to
$\operatorname{span}\{W_1,W_2\}$.
Each distance dominates the corresponding absolute projection.
Independence and the uniform bound defining $K_s$ therefore imply
\[
\mathbb E[R_3^{-s}\mid W_1,W_2]\le K_s,\qquad
\mathbb E[R_2^{-s}\mid W_1]\le K_s,\qquad
\mathbb E[R_1^{-s}]\le K_s.
\]
Successive conditioning yields
\[
\mathbb E|D(1)|^{-s}
=\mathbb E[R_1^{-s}R_2^{-s}R_3^{-s}]
\le K_s^3<\infty.
\]
The conditional bounds also show that the relevant degeneracies have
probability zero. This proves Claim~3 and completes the proof.
\end{proof}

\end{appendix}
\end{document}